\documentclass[12pt]{article}  

\usepackage{amssymb}
\usepackage{amsthm}
\usepackage{graphicx}
\usepackage{amsmath}
\usepackage{amsfonts}
\usepackage{color}
\usepackage{centernot}
\usepackage{mathtools}
\usepackage{stmaryrd}
\newtheorem{theorem}{Theorem}
\newtheorem{lemma}{Lemma}
\newtheorem{proposition}{Proposition}
\newtheorem{definition}{Definition}
\newtheorem{corollary}{Corollary}
\newtheorem*{conjecture*}{Conjecture}
\numberwithin{equation}{section}

\begin{document}

\title{Hidden spatiotemporal symmetries\\ and intermittency in turbulence}
\author{Alexei A. Mailybaev}
\date{Instituto de Matem\'atica Pura e Aplicada -- IMPA, Rio de Janeiro, Brazil \\
Email: alexei@impa.br}

\maketitle

\begin{abstract}
We consider general infinite-dimensional dynamical systems with the Galilean and spatiotemporal scaling symmetry groups. Introducing the equivalence relation with respect to temporal scalings and Galilean transformations, we define a representative set containing a single element within each equivalence class. Temporal scalings and Galilean transformations do not commute with the evolution operator (flow) and, hence, the equivalence relation is not invariant. Despite of that, we prove that a normalized flow with an invariant probability measure can be introduced on the representative set, such that symmetries are preserved in the statistical sense. We focus on hidden symmetries, which are broken in the original system but restored in the normalized system. The central motivation and application of this construction is the intermittency phenomenon in turbulence. We show that hidden symmetries yield power law scaling for structure functions, and derive formulas for their exponents in terms of normalized measures. The use of Galilean transformation in the equivalence relation leads to the Quasi--Lagrangian description, making the developed theory applicable to the Euler and Navier--Stokes systems.
\end{abstract}


\section{Introduction}

Symmetry principles play important role in understanding the laws of nature~\cite{brading2003symmetries}. In particular, they provide powerful tools for the analysis of complex systems through self-similarity and renormalization~\cite{cardy1996scaling}. In this work, we focus on symmetries, which shape the modern understanding of developed turbulence: the Galilean and spatiotemporal scaling groups~\cite{frisch1999turbulence,falkovich2009symmetries}. Symmetry considerations are central in Kolmogorov’s theory of 1941 \cite{kolmogorov1941local}, which assumes a homogeneous, isotropic and scale invariant stationary state. These symmetries are understood in the statistical sense, i.e., being satisfied by probabilistic quantities rather than exact solutions of equations of motion. This is an important distinction, since probabilistic formulations may lead to additional symmetries; see e.g.~\cite{kraichnan1965lagrangian,kraichnan1975remarks,oberlack2010new,waclawczyk2014statistical,mailybaev2020hidden,oberlack2022turbulence}. Whether or not, and in which sense solutions are symmetric is an important issue, both for the theory and applications. For example, the broken scale invariance of statistically stationary solutions underlines  the still not well understood phenomenon of intermittency in turbulence~\cite{frisch1999turbulence}. 

In this work, we investigate the particular role of symmetries that do not commute with the flow (evolution) operator $\Phi^t$. The two fundamental symmetries of this kind are temporal scalings and Galilean transformations. Their commutation with $\Phi^t$ relates states at different times or translated in physical space. We prove that such noncommutativity is responsible for the existence of sophisticated ``hidden'' symmetries of statistical solutions: these symmetries are broken in the original formulation but can be restored using equivalence relations. 

{\color{black}The suggested new formalism follows and gives a rigorous foundation to several phenomenological ideas in the turbulence theory. Their origin lies in the famous work of Kolmogorov in 1962~\cite{kolmogorov1962refinement}, where the concept of ``multipliers" first appeared (as they were called later). Kolmogorov’s hypothesis of self-similarity for these multipliers, which are ratios of velocity differences at distinct scales, can be seen as the first manifestation of the hidden symmetry. This idea was inspired by the theory of multiplicative stochastic processes and further discussed in~\cite{benzi1993intermittency,chen2003kolmogorov,eyink2003gibbsian,biferale2017optimal,PhysRevX.11.021063}. Another idea came from the work of Parisi and Frisch in 1983 on the multifractal model~\cite{frisch1985singularity,frisch1999turbulence}. Those authors remarked that ``Since the Navier-Stokes equations (in the zero viscosity limit) are invariant under the group of
scaling transformations (defined in eq.(2.2)) for \textit{any} value of $h$, singularities of arbitrary exponents (and mixtures thereof) are consistent with the equations." The hidden-symmetry formalism presented below naturally unifies the ideas of Kolmogorov with those of Parisi-Frisch. We prove that our construction fuses the one-parameter family of space-time scaling symmetries (depending on $h$) into the single hidden symmetry, therefore, reducing the Parisi--Frisch argument to the restoration of the hidden symmetry alone and in the usual sense. Existence of such kind of symmetry in intermittent turbulence was also anticipated in the work of She and Leveque in 1994 on the log-Poisson model~\cite{she1994universal}, where the
authors wrote that “We believe that this relation is a consequence of some hidden (statistical) symmetries in the solution of the Navier-Stokes equations.”}

\subsection{Spatiotemporal symmetries}
\label{secEuler}

Let us introduce a group of space-time symmetries of interest by examining the Euler system, which describes a flow of ideal incompressible fluid of unit density. Its equations have the form
	\begin{equation}
	\frac{\partial \mathbf{u}}{\partial t}
	+ \mathbf{u}\cdot \nabla \mathbf{u} = -\nabla p,\quad 
	\nabla \cdot \mathbf{u} = 0,
	\label{eqE3b}
	\end{equation}
where $\mathbf{u}(\mathbf{r},t) \in \mathbb{R}^d$ is the velocity field and $p(\mathbf{r},t) \in \mathbb{R}$ is the pressure in physical space $\mathbf{r} \in \mathbb{R}^d$ of dimension $d$. Given a solution $\mathbf{u}(\mathbf{r},t)$ the following relations generate new solutions as\\
	\begin{equation}
	\begin{array}{rll}
		\textrm{temporal translation:}& 
		\mathbf{u}(\mathbf{r},t) \mapsto \mathbf{u}(\mathbf{r},t'+t),&
		t' \in \mathbb{R};
		\\[2pt]
		\textrm{spatial translation:}& 
		\mathbf{u}(\mathbf{r},t) \mapsto \mathbf{u}(\mathbf{r}+\mathbf{r}',t),&
		\mathbf{r}' \in \mathbb{R}^d;
		\\[2pt]
		\textrm{rotation:}& 
		\mathbf{u}(\mathbf{r},t) \mapsto \mathbf{Q}^{-1}\mathbf{u}(\mathbf{Q}\mathbf{r},t),&
		\mathbf{Q} \in \mathrm{O}(d);
		\\[2pt]
		\textrm{Galilean transformation:}& 
		\mathbf{u}(\mathbf{r},t) \mapsto \mathbf{u}(\mathbf{r}+\mathbf{v}t,t)-\mathbf{v},&
		\mathbf{v} \in \mathbb{R}^d;
		\\[2pt]
		\textrm{temporal scaling:}& 
		\mathbf{u}(\mathbf{r},t) \mapsto \mathbf{u}(\mathbf{r},t/a)/a,&
		a > 0;
		\\[2pt]
		\textrm{spatial scaling:}& 
		\mathbf{u}(\mathbf{r},t) \mapsto b\mathbf{u}(\mathbf{r}/b,t),&
		b > 0,
	\end{array}	
	\label{eq2}
	\end{equation}
where $\mathrm{O}(d)$ is the orthogonal group;
the pressure is not included because it can be expressed through velocity~\cite{frisch1999turbulence}. Transformations (\ref{eq2}) generate the sum of Galilean and spatiotemporal scaling groups. 

We now write transformations (\ref{eq2}) in terms of the evolution operator (flow) $\Phi^t$ and mappings acting on velocity fields at a fixed time. In this description, points of the configuration space $\mathcal{X}$ are time-independent velocity fields $x = \mathbf{u}(\mathbf{r})$, and the flow $\Phi^t: \mathcal{X} \mapsto \mathcal{X}$ relates velocity fields at different times with the property $\Phi^{t_1+t_2} = \Phi^{t_1} \circ \Phi^{t_2}$ for any $t_1$ and $t_2$. The flow $\Phi^t$ is associated with temporal translations, and remaining relations in (\ref{eq2}) taken at $t = 0$ yield the maps $s:\mathcal{X} \mapsto \mathcal{X}$ as
	\begin{equation}
	\begin{array}{rlll}
		s^{\mathbf{r}'}_{\mathrm{s}}:& 
		\mathbf{u}(\mathbf{r}) \mapsto \mathbf{u}(\mathbf{r}+\mathbf{r}'),&
		\mathbf{r}' \in \mathbb{R}^d, & (\textrm{spatial translation})
		\\[2pt]
		s^{\mathbf{Q}}_{\mathrm{r}}:& 
		\mathbf{u}(\mathbf{r}) \mapsto \mathbf{Q}^{-1}\mathbf{u}(\mathbf{Q}\mathbf{r}),&
		\mathbf{Q} \in \mathrm{O}(d), & (\textrm{rotation})
		\\[2pt]
		s^{\mathbf{v}}_{\mathrm{g}}:& 
		\mathbf{u}(\mathbf{r}) \mapsto \mathbf{u}(\mathbf{r})-\mathbf{v},&
		\mathbf{v} \in \mathbb{R}^d, & (\textrm{Galilean transformation})
		\\[2pt]
		s^{a}_{\mathrm{ts}}:& 
		\mathbf{u}(\mathbf{r}) \mapsto \mathbf{u}(\mathbf{r})/a,&
		a > 0, & (\textrm{temporal scaling})
		\\[2pt]
		s^{b}_{\mathrm{ss}}:& 
		\mathbf{u}(\mathbf{r}) \mapsto b\mathbf{u}(\mathbf{r}/b),&
		b > 0. & (\textrm{spatial scaling})
	\end{array}	
	\label{eq3}
	\end{equation}

Table~\ref{tab1} describes commutation relations for the flow $\Phi^t$ and all mappings in (\ref{eq3}) in agreement with time-dependent transformations (\ref{eq2}). In our study, {\color{black}we will not refer to any particular system, except in explicit examples, but instead consider Tab.~\ref{tab1} as a definition, which is  based on fundamental physical properties of space and time. Namely,} we assume the existence of flow $\Phi^t$ and other maps from Tab.~\ref{tab1} acting on some configuration space $\mathcal{X}$ and generating a group with the composition operation. 

{\color{black}The assumed existence of a flow (or semiflow) operator $\Phi^t$ deserves a special remark, because it is a still unresolved issue for the Euler equations (\ref{eqE3b}); see e.g. \cite{gibbon2008three}. In the traditional approach of developed turbulence~\cite{frisch1999turbulence}, symmetries of Tab.~\ref{tab1} are considered in the asymptotic sense, corresponding to the inviscid limit of Navier--Stokes equations. The Navier--Stokes system is supposed to have a unique solution, though this has not yet been rigorously proven~\cite{fefferman2006existence}. Having this approach in mind (developed in more details in Section~\ref{subsec_conj}), we assume the existence of a flow map $\Phi^t$, therefore, bypassing the lack of global-in-time existence and uniqueness results for particular systems of interest. On the other hand, recent studies~\cite{mailybaev2016spontaneously,biferale2018rayleigh,mailybaev2021spontaneously} indicate that the inviscid limit yields spontaneously stochastic solutions, in which case the map $\Phi^t$ is defined as acting on probability distributions (for both velocity fields~\cite{thalabard2020butterfly} and particle trajectories~\cite{falkovich2001particles,eyink2006turbulent,eyink2007turbulent}) rather than on specific deterministic states. We expect that the hidden symmetry formalism presented here can later be extended to such systems, along with the development of the theory of spontaneous stochasticity.}

\begin{table}
\begin{center}
\normalsize
 \begin{tabular}{| c | l l l l l l |} 
\hline
	& $\ \Phi^t\qquad\qquad\quad$ & 
	$s^{\mathbf{r}}_{\mathrm{s}}\qquad\qquad$ & 
	$s^{\mathbf{Q}}_{\mathrm{r}}\qquad\quad\quad$ & 
	$s^{\mathbf{v}}_{\mathrm{g}}\qquad\qquad\quad$ &
	$s^{a}_{\mathrm{ts}}\qquad\quad$ & 
	$s^{b}_{\mathrm{ss}}\qquad\quad$ 
\\[3pt] \hline 
	$\ \ \Phi^t\ \ $ & 
	$\ \Phi^{t_1+t_2}$ & 
	$s^{\mathbf{r}}_{\mathrm{s}} \circ \Phi^t$ & 
	$s^{\mathbf{Q}}_{\mathrm{r}} \circ \Phi^t$ & 
	$s^{\mathbf{v}t}_{\mathrm{s}} \circ s^{\mathbf{v}}_{\mathrm{g}} \circ \Phi^t$ &
	$s^a_{\mathrm{ts}} \circ \Phi^{t/a}$ & 
	$s^{b}_{\mathrm{ss}} \circ \Phi^t$ 
\\[3pt] 
	$s^{\mathbf{r}}_{\mathrm{s}}$ & 
	$\ \Phi^t \circ s^{\mathbf{r}}_{\mathrm{s}}$ & 
	$s^{\mathbf{r}_1+\mathbf{r}_2}_{\mathrm{s}}$ & 
	$s^{\mathbf{Q}}_{\mathrm{r}} \circ s^{\mathbf{Q}\mathbf{r}}_{\mathrm{s}}$ & 
	$s^{\mathbf{v}}_{\mathrm{g}} \circ s^{\mathbf{r}}_{\mathrm{s}}$ &
	$s^{a}_{\mathrm{ts}} \circ s^{\mathbf{r}}_{\mathrm{s}}$ & 
	$s^{b}_{\mathrm{ss}} \circ s^{\mathbf{r}/b}_{\mathrm{s}}$ 
\\[3pt] 
	$s^{\mathbf{Q}}_{\mathrm{r}}$ & 
	$\ \Phi^t \circ s^{\mathbf{Q}}_{\mathrm{r}}$ & 
	$s^{\mathbf{Q}^{-1}\mathbf{r}}_{\mathrm{s}} \circ s^{\mathbf{Q}}_{\mathrm{r}}$ & 
	$s^{\mathbf{Q}_1\mathbf{Q}_2}_{\mathrm{r}}$ &
	$s^{\mathbf{Q}^{-1}\mathbf{v}}_{\mathrm{g}} \circ s^{\mathbf{Q}}_{\mathrm{r}}$ &
	$s^{a}_{\mathrm{ts}} \circ s^{\mathbf{Q}}_{\mathrm{r}}$ &
	$s^{b}_{\mathrm{ss}} \circ s^{\mathbf{Q}}_{\mathrm{r}}$ 
\\[3pt] 
	$s^{\mathbf{v}}_{\mathrm{g}}$ & 
	$\ s^{-\mathbf{v}t}_{\mathrm{s}} \circ \Phi^t \circ s^{\mathbf{v}}_{\mathrm{g}}$ & 
	$s^{\mathbf{r}}_{\mathrm{s}} \circ s^{\mathbf{v}}_{\mathrm{g}}$ & 
	$s^{\mathbf{Q}}_{\mathrm{r}} \circ s^{\mathbf{Q}\mathbf{v}}_{\mathrm{g}}$ & 
	$s^{\mathbf{v}_1+\mathbf{v}_2}_{\mathrm{g}}$ &
	$s^{a}_{\mathrm{ts}} \circ s^{a\mathbf{v}}_{\mathrm{g}}$ & 
	$s^{b}_{\mathrm{ss}} \circ s^{\mathbf{v}/b}_{\mathrm{g}}$ 
\\[3pt] 
	$s^{a}_{\mathrm{ts}}$ & 
	$\ \Phi^{at} \circ s^a_{\mathrm{ts}}$ & 
	$s^{\mathbf{r}}_{\mathrm{s}} \circ s^{a}_{\mathrm{ts}}$ & 
	$s^{\mathbf{Q}}_{\mathrm{r}} \circ s^{a}_{\mathrm{ts}}$ & 
	$s^{\mathbf{v}/a}_{\mathrm{g}} \circ s^{a}_{\mathrm{ts}}$ &
	$s^{a_1a_2}_{\mathrm{ts}}$ &
	$s^b_{\mathrm{ss}} \circ s^a_{\mathrm{ts}}$ 
\\[3pt] 
	$s^{b}_{\mathrm{ss}}$ & 
	$\ \Phi^t \circ s^{b}_{\mathrm{ss}}$ & 
	$s^{b\mathbf{r}}_{\mathrm{s}} \circ s^{b}_{\mathrm{ss}}$ & 
	$s^{\mathbf{Q}}_{\mathrm{r}} \circ s^{b}_{\mathrm{ss}}$ & 
	$s^{b\mathbf{v}}_{\mathrm{g}} \circ s^{b}_{\mathrm{ss}}$ &
	$s^{a}_{\mathrm{ts}} \circ s^{b}_{\mathrm{ss}}$ & 
	$s^{b_1b_2}_{\mathrm{ss}}$ 
\\[3pt] \hline
\end{tabular}
\end{center}
\caption{Commutation relations among the flow $\Phi^t$ and symmetry mappings (\ref{eq3}); the primes are dropped for simplicity. In these relations, the left-hand side is understood as $(\textrm{row map}) \circ (\textrm{column map})$ and the right-hand side is given in the main part of the table. For the diagonal elements, one assumes the index $1$ for the row and $2$ for the column.}
\label{tab1}
\end{table}

Focusing on statistical properties of the flow, we consider an invariant probability measure $\mu$ on the configuration space $\mathcal{X}$. The invariance signifies that the push-forward $\Phi^t_\sharp \mu = \mu$ for any time. Then, we consider symmetries in the statistical sense, as transformations of $\mu$ preserving its invariance. For example, one can see using the commutation relations of Tab.~\ref{tab1} that all maps in (\ref{eq3}), except for Galilean transformations, are symmetries: a push-forward of $\mu$ by these maps yield invariant measures. Galilean transformations become symmetries under an extra homogeneity condition for the invariant measure: $\left(s^{\mathbf{r}}_{\mathrm{s}}\right)_\sharp \mu = \mu$ for any translation $\mathbf{r}$ in physical space.

\subsection{Quotient construction} 

Our study will be developed around the two groups 
	\begin{eqnarray}
	\mathcal{H} & = & \big\{s^{a}_{\mathrm{ts}} \circ s^{\mathbf{v}}_{\mathrm{g}}:\ 
	a > 0,\ \mathbf{v} \in \mathbb{R}^d \big\},
	\label{eq_I_1a}
	\\[3pt]
	\mathcal{G} & = & \big\{s^{\mathbf{Q}}_{\mathrm{r}} \circ 
	s^b_{\mathrm{ss}}:\ 
	\mathbf{Q} \in \mathrm{O}(d),\,
	b > 0\big\}.
	\label{eq_I_1b}
	\end{eqnarray}
The group $\mathcal{H}$ contains maps $h: \mathcal{X} \mapsto \mathcal{X}$ generated by temporal scalings and Galilean transformations, which do not commute with the flow; see Tab.~\ref{tab1}. The commutation of $\Phi^t$ with $s^{a}_{\mathrm{ts}}$ leads to a different time $t/a$, while the commutation of $\Phi^t$ with $s^{\mathbf{v}}_{\mathrm{g}}$ contains an extra spatial translation $s^{\mathbf{v}t}_{\mathrm{s}}$. Maps $g: \mathcal{X} \mapsto \mathcal{X}$ of the group $\mathcal{G}$ are generated by spatial rotations and scalings, which commute with $\Phi^t$. Spatial translations $s^{\mathbf{r}}_{\mathrm{s}}$, which are not included in $\mathcal{H}$ and $\mathcal{G}$, will play an auxiliary role in our study.

Using the group $\mathcal{H}$, we define the equivalence relation between two states as
	\begin{equation}
	x \sim x' \quad \textrm{if} \quad
	x' = h(x),\ h \in \mathcal{H}.
	\label{eq_I_2}
	\end{equation}
Equivalence classes 
	\begin{equation}
	\mathcal{E}(x) = \{x' \in \mathcal{X}: x' \sim x \} 
	\label{eq_I_2b}
	\end{equation}
form a partition of the configuration space $\mathcal{X}$. Because of noncommutativity, this partition is not invariant with respect to the flow: generally, $\Phi^t(x)$ and $\Phi^t(x')$ are not equivalent for initially equivalent states $x \sim x'$; see Fig.~\ref{fig1}(a). However, due to the specific form of commutation relations, the equivalence can be ``repaired'' as follows. Using relations of Tab.~\ref{tab1}, we have
	\begin{equation}
	s^{-\mathbf{v}t}_{\mathrm{s}} \circ
	\Phi^{at} \circ s^{a}_{\mathrm{ts}} \circ s^{\mathbf{v}}_{\mathrm{g}} 
	= s^{a}_{\mathrm{ts}} \circ s^{\mathbf{v}}_{\mathrm{g}}
	\circ \Phi^{t}.
	\label{eq_I_3ext}
	\end{equation}
Hence, we can write
	\begin{equation}
	s^{\mathbf{r}}_{\mathrm{s}} \circ \Phi^{t'}(x') = h \circ \Phi^t(x), \quad
	t' = at,\quad \mathbf{r} = -\mathbf{v}t,
	\label{eq_I_3}
	\end{equation}
for $x' = h(x)$ with a general element $h = s^{a}_{\mathrm{ts}} \circ s^{\mathbf{v}}_{\mathrm{g}}$ of the group (\ref{eq_I_1a}). Thus, all initially equivalent states $x \sim x'$ are fit into the same equivalence class at larger times, if one assumes the specific time synchronization $t' = at$ and the extra spatial translation $\mathbf{r} = -\mathbf{v}t$ for each $x'$, as shown in Fig.~\ref{fig1}(b). This construction is determined by a selected representative element $x$, with respect to which all other equivalent states are ``synchronized''.

\begin{figure}[t]
\centering
\includegraphics[width=0.8\textwidth]{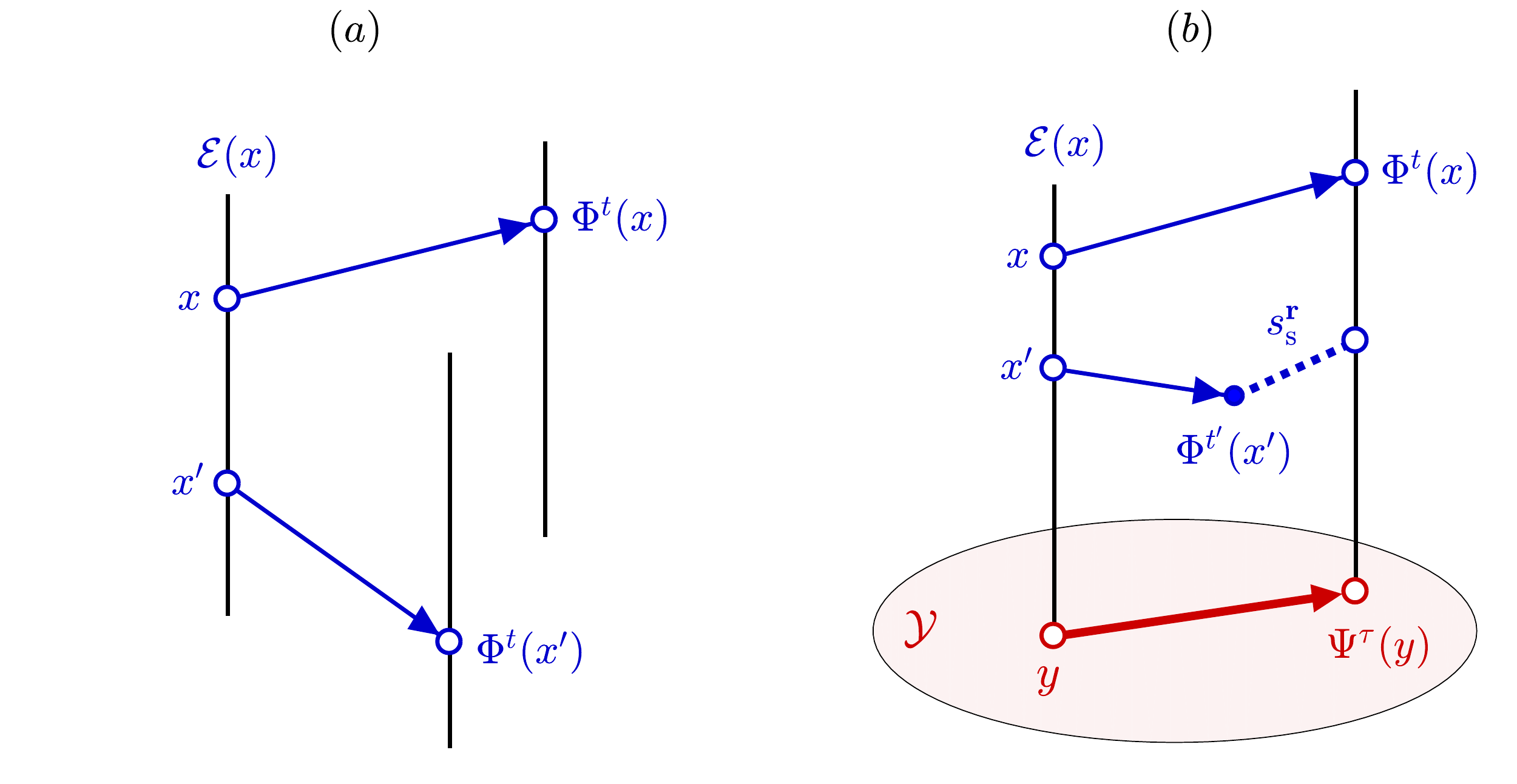}
\caption{Structure of configuration space $\mathcal{X}$ with a partition to equivalence classes (straight vertical lines) with respect to the symmetry group $\mathcal{H}$. (a) Due to noncommutativity with the flow, the equivalence relation $x \sim x'$ is not invariant: the states $\Phi^t(x)$ and $\Phi^t(x')$ are generally not equivalent. (b) The equivalence can be ``repaired'' by choosing a specific time $t'$ and  an extra spatial translation $s^{\mathbf{r}}_{\mathrm{s}}$, fitting the initially equivalent states $x \sim x'$ into the equivalence class of $\Phi^t(x)$ at a later time. Such construction can be introduced globally by synchronizing the flow with respect to a representative set $\mathcal{Y}$, which contains a single state from every equivalence class. This construction induces the dynamics in $\mathcal{Y}$ governed by a new normalized flow $\Psi^\tau$.}
\label{fig1}
\end{figure}

In this paper, we develop such a quotient-like construction globally in the configuration space $\mathcal{X}$ by introducing a representative set $\mathcal{Y} \subset \mathcal{X}$, which contains a single element $y \in \mathcal{Y}$ within each equivalence class $\mathcal{E}(x)$; see Fig.~\ref{fig1}(b). As a result, we reduce the original dynamical system in $\mathcal{X}$ to the dynamical system in $\mathcal{Y}$, which we call the \textit{normalized} system. We prove the following properties of this construction:
\begin{itemize}
\item 
There is a normalized flow $\Psi^\tau: \mathcal{Y} \mapsto \mathcal{Y}$ on the representative set, which is induced by $\Phi^t$ and the equivalence relation (\ref{eq_I_2}); see Fig.~\ref{fig1}(b).
\item 
The normalized flow $\Psi^\tau$ has the invariant measure $\nu$, which is explicitly related to the original invariant measure $\mu$. 
\item 
The group (\ref{eq_I_1b}) defines statistical symmetries in the normalized system. We introduce a transformation $\nu \mapsto g_\star\nu$ for any $g \in \mathcal{G}$, akin to the push-forward. This transformation preserves the group structure and the invariance of a measure with respect to $\Psi^\tau$.
\item
For any given $h \in \mathcal{H}$ and $g \in \mathcal{G}$, the symmetry of $\mu$ implies the symmetry of $\nu$ in the form
	\begin{equation}
	(g \circ h)_\sharp \mu = \mu \quad \Rightarrow \quad g_\star \nu = \nu.
	\label{eq_I_3add}
	\end{equation}
{\color{black}The converse is not true in general.}
\item
The property of statistical symmetry, $g_\star \nu = \nu$ for a given element $g \in \mathcal{G}$, does not depend on a choice of the representative set $\mathcal{Y}$.
\end{itemize}

Notice that the transformation from original to normalized system is time-dependent. In general, such transformations do not preserve statistical properties, e.g. the measure invariance. In fact, the listed properties follow in a nontrivial way from the specific commutation relations of Tab.~\ref{tab1}.

\subsection{Hidden symmetries, multifractality, intermittency and sweeping effects}
\label{subsec_hid}

The main motivation of the developed construction is related to the interplay between statistical symmetries in the original and normalized systems. For understanding a general idea, {\color{black}let us consider $g = s^b_{\mathrm{ss}}$ with $b = 2$ corresponding to the change of spatial scale by a factor of two and $h^a = s^a_{\mathrm{ts}}$ determining the temporal scaling with a particular factor $a > 0$. 
Using relations (\ref{eq2}) and (\ref{eq3}) for velocity fields, we see that the combined symmetry $g \circ h^a$ is associated with the spatiotemporal scaling transformation of the form
	\begin{equation}
	\mathbf{u}(\mathbf{r},t) \ \mapsto \ 2^{1-\alpha}\,
	\mathbf{u}\left(\frac{\mathbf{r}}{2},\frac{t}{2^\alpha}\right),
	\label{eq_I_5}
	\end{equation}
where $\alpha = \log_2 a$. According to (\ref{eq_I_3add}), every space-time symmetry $(g \circ h)_\sharp \mu = \mu$ implies $g_\star \nu = \nu$, but not vice versa. In particular, we can have situations when 
	\begin{equation}
	(g \circ h^a)_\sharp \mu \ne \mu,\quad 
	g_\star \nu = \nu,
	\label{eq_I_8}
	\end{equation}
where the first condition refers to any $a > 0$; see \cite{mailybaev2021solvable} for a rigorous example. 
This means that the normalized measure $\nu$ remains symmetric, while all symmetries of the original measure are broken. This is what we call the \textit{hidden symmetry}: a statistical symmetry is restored only in the normalized system.

Our central application is the demonstration that the hidden symmetry provides a rigorous foundation for the multifractal theory in turbulence~\cite{frisch1985singularity,sreenivasan1991fractals,frisch1999turbulence}. This phenomenological theory models an intermittent turbulent state as a sum of statistical behaviours (singularities) featuring different scaling laws (\ref{eq_I_5}) and supported in subspaces of different fractal dimensions. The intermittency is quantified using structure functions of different orders $p$ defined as the mean value $S_p(\ell) = \langle \|\delta_\ell \mathbf{u}\|^p \rangle$ for a difference of fluid velocities $\delta_\ell \mathbf{u} = \mathbf{u}(\mathbf{r}')-\mathbf{u}(\mathbf{r})$ at a distance $\ell  = \|\mathbf{r}'-\mathbf{r}\| > 0$. The multifractal statistics yields the asymptotic power law
	\begin{equation}
	S_p(\ell) \propto \ell^{\zeta_p}	
	\label{eq_I_9}
	\end{equation}
at small $\ell$ with the exponent $\zeta_p$ depending nonlinearly on $p$.}
In this work, we derive asymptotic power laws (\ref{eq_I_9}) from the assumption of hidden scaling symmetry (\ref{eq_I_8}). This derivation provides formulas for the exponent $\zeta_p$ in terms of Perron--Frobenius eigenvalues of operators constructed for the symmetric normalized measure $\nu$. We show that the resulting exponents $\zeta_p$ can be anomalous, i.e., depending nonlinearly on $p$. This leads us to the conjecture that the developed turbulent state in the inertial interval (where the dynamics is governed by the Euler system) possesses a hidden scaling symmetry (\ref{eq_I_8}). {\color{black}In fact, the formalism developed here was used in the subsequent works for verifying the hidden self-similarity and its implications in shell models of turbulence~\cite{mailybaev2020hidden,mailybaev2021solvable,mailybaev2022shell} and in the Navier--Stokes system~\cite{mailybaev2022hidden}.}

Finally, we mention the role of Galilean transformations in the equivalence relation of our quotient construction. Galilean transformations yield a normalized system in the form analogous to the Quasi--Lagrangian representation in fluid dynamics, i.e., describing the system in a reference frame moving with a Lagrangian (fluid) particle {\color{black}\cite{belinicher1987scale,l1991scale}}. As a consequence, the quotient construction removes the so-called sweeping effect caused by a large-scale motion, the well-known obstacle for describing statistical properties at small scales~\cite{frisch1999turbulence}. This makes the developed theory applicable to  real turbulence problems. Remarkably, our quotient construction imposes extra algebraic conditions, one of which corresponds to incompressibility in fluid dynamics.

\subsection{Structure of the paper} 

In Section~\ref{sec4}, we consider a simpler quotient construction by excluding Galilean transformations, i.e., the equivalence relation is considered only with respect to temporal scalings. We introduce the representative set $\mathcal{Y}$, the normalized flow $\Psi^\tau$, the invariant normalized measure $\nu$ and the group action $g_\star$, and investigate their basic properties. In Section~\ref{sec_shell1}, this procedure is carried out explicitly for a shell model of turbulence, providing the evidence of hidden scaling symmetry.

Section~\ref{sec_int} presents our central application. It shows that the hidden scaling symmetry implies asymptotic scaling laws for structure functions. The scaling exponents are obtained in terms of  Perron--Frobenius eigenvalues by exploiting the symmetry of the normalized measure $\nu$. {\color{black}These results are confirmed analytically and numerically for anomalous exponents of intermittent statistics in shell models~\cite{mailybaev2021solvable,mailybaev2022shell}.}

Section~\ref{sec_Gal} develops a quotient construction for the equivalence relation with respect to Galilean transformations. We show that this construction is possible assuming additional properties of the measure $\mu$. Remarkably, these properties have the physical meaning of spatial homogeneity and incompressibility, and the resulting normalized system is analogous to the Quasi--Lagrangian description in fluid dynamics. In Section~\ref{sec_fus}, we develop the final quotient construction, in which the equivalence takes into account both Galilean transformations and temporal scalings. We show how this construction can be applied to the study of turbulence in the Euler and Navier--Stokes systems. The Conclusion section contains a short summary. 

\section{Quotient construction with temporal scalings}
\label{sec4}

Let us consider a {\color{black}probability measure space $(\mathcal{X},\Sigma,\mu)$. Because of applications we have in mind, the space is assumed to be infinite-dimensional. By definition~\cite{cornfeld2012ergodic}, the flow operator $\Phi^t:\mathcal{X} \mapsto \mathcal{X}$} is a one-parameter group of one-to-one measurable maps such that $\Phi^{t_1} \circ \Phi^{t_2} = \Phi^{t_1+t_2}$ for all times. The flow  {\color{black}must also be} measurable as a function of $(x,t) \in \mathcal{X} \times \mathbb{R}$.
We will use the following notions.

{\color{black}
\begin{definition}\label{del_measure}
Here we introduce three interconnected concepts: the invariant measure, the symmetry map and the symmetric measure:
\begin{itemize}
\item A probability measure $\mu$ is said to be invariant for the flow $\Phi^t$ if the push-forward (image) $\Phi^t_\sharp \mu = \mu$ for all times. 
\item We call a one-to-one measurable map $s: \mathcal{X} \mapsto \mathcal{X}$ symmetry, if the invariance of any measure $\mu$ implies the invariance of $s_\sharp \mu$.  A set of symmetries $s \in \mathcal{S}$ with a group operation given by composition $s_1 \circ s_2$ is called a symmetry group. 
\item Let $s$ be a symmetry. A given measure $\mu$ is said to be symmetric with respect to $s$ if $s_\sharp\mu = \mu$. In the opposite situation, $s_\sharp\mu \ne \mu$, we say that the symmetry is broken.
\end{itemize}
\end{definition}
}

We emphasize that symmetries in this definition are understood in the \textit{statistical} sense: they are defined through their action on invariant probability measures. This, in particular, implies that symmetries do not necessarily commute with the flow $\Phi^t$. 

We will always assume that the measure $\mu$ is invariant. In this section, we consider a symmetry group given by a direct sum 
	\begin{equation}
	\mathcal{S} = \mathcal{H}_{\mathrm{ts}}+\mathcal{G}. 
	\label{eq2_0N}
	\end{equation}
Here $ \mathcal{H}_{\mathrm{ts}}$ is a one-parameter group of temporal scalings
	\begin{equation}
	\mathcal{H}_{\mathrm{ts}} = \big\{s^{a}_{\mathrm{ts}}:\ a > 0\big\}.
	\label{eq2_1}
	\end{equation}
We will adopt the shorter notation $h^a = s^{a}_{\mathrm{ts}}$. {\color{black}The group $\mathcal{G}$ can be taken in the form (\ref{eq_I_1b}), containing compositions of spatial rotations and scalings. Then,} elements $h^a \in \mathcal{H}_{\mathrm{ts}}$ and $g \in \mathcal{G}$ are one-to-one measurable maps in $\mathcal{X}$ satisfying the commutation relations {\color{black}(see Tab.~\ref{tab1})}
	\begin{eqnarray}
	\label{eq2_0ha}
	& h^{a_1} \circ h^{a_2} = h^{a_1a_2}, 
	\\
	& \Phi^t \circ g = g \circ \Phi^t,\quad
	g \circ h^a = h^a \circ g,
	\label{eq2_2com}
	\\
	& \Phi^t \circ h^a = h^a \circ \Phi^{t/a}.
	\label{eq2_2}
	\end{eqnarray}
Using these relations, it is straightforward to check that any element $s \in \mathcal{S}$ is a symmetry in the sense of Definition~\ref{del_measure}. Relations (\ref{eq2_0ha})--(\ref{eq2_2}) are all we need to know about the symmetry group {\color{black}for further derivations.} Notice that Galilean transformations will not be considered until Section~\ref{sec_Gal}.

\subsection{Normalized flow and invariant measure}\label{subsec_red}

Let us consider the equivalence relation with respect to temporal scalings $\mathcal{H}_{\mathrm{ts}}$ as
	\begin{equation}
	x \sim x' \quad \textrm{if} \quad x' = h^a(x),\ a > 0.
	\label{eq2_Eq1}
	\end{equation}
For each $x \in \mathcal{X}$, this relation defines the equivalence class
	\begin{equation}
	\mathcal{E}_{\mathrm{ts}}(x) = \{x' \in \mathcal{X}: x' \sim x \}.
	\label{eq2_Eq2}
	\end{equation}
Because of commutation relation (\ref{eq2_2}), for the equivalent states (\ref{eq2_Eq1}) we have
	\begin{equation}
	\Phi^{at}(x') = h^a \circ \Phi^t(x).
	\label{eq2_Eq3}
	\end{equation}
Hence, the equivalence relation is not preserved by the flow: the states $\Phi^t(x)$ and $\Phi^t(x')$ are generally not equivalent at the same time $t > 0$. However, the equivalence can be restored by considering a different time $t' = at$ for the state $x'$, which yields $\Phi^{t'}(x') \sim \Phi^t(x)$; see Fig.~\ref{fig2}. Such time synchronization requires a choice of a representative element $x$ in the equivalence class, and can be introduced globally using a representative set consisting of these elements. 

\begin{figure}[t]
\centering
\includegraphics[width=0.43\textwidth]{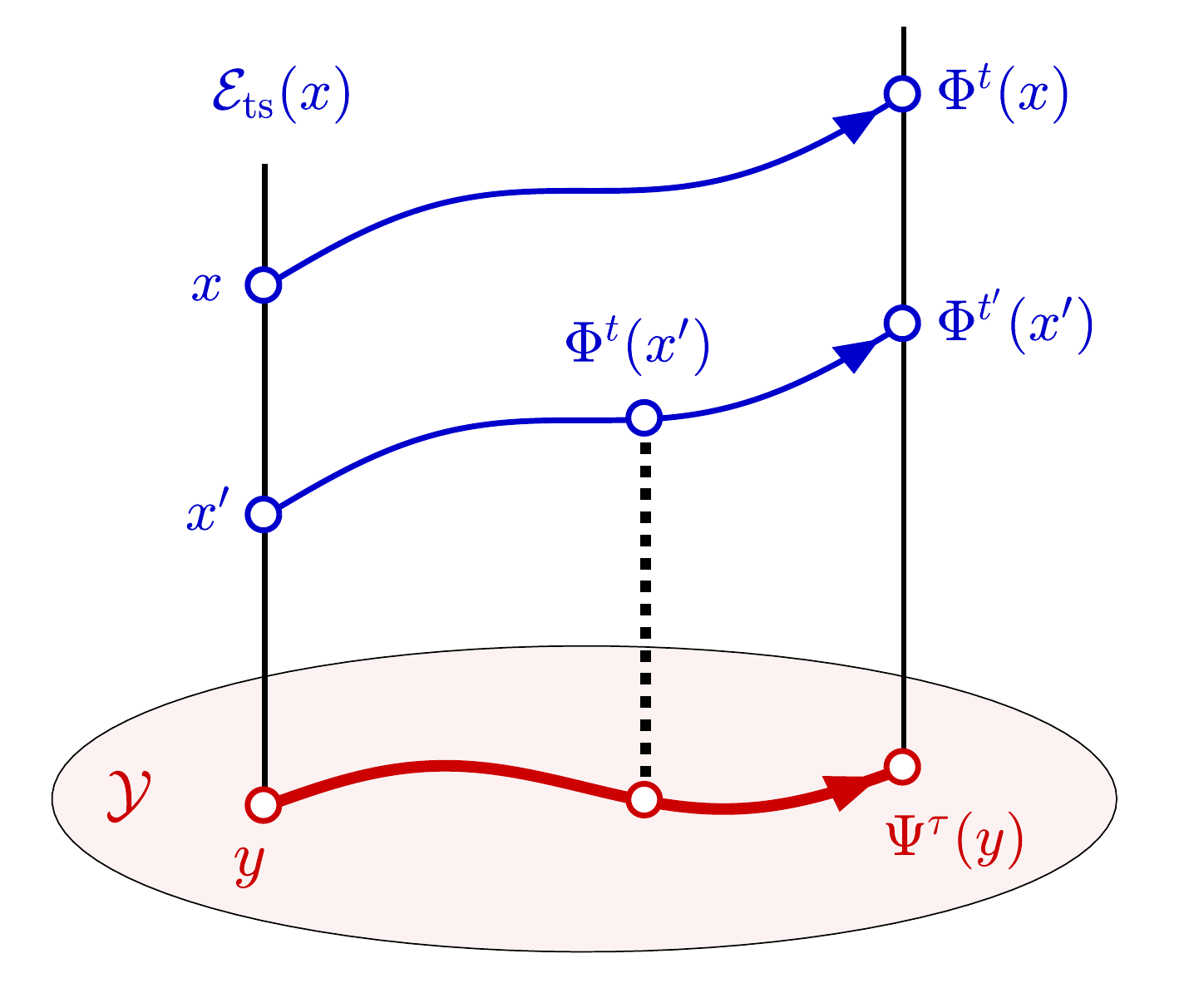}
\caption{Structure of configuration space $\mathcal{X}$ with a partition to equivalence classes (straight vertical lines) $\mathcal{E}_{\mathrm{ts}}(x)$. The equivalence relation $x \sim x'$ is not invariant: the states $\Phi^t(x)$ and $\Phi^t(x')$ are generally not equivalent as shown by the dotted line. The equivalence can be restored by choosing a different time $t' = at$ for $x'$. Such construction is introduced globally by synchronizing the flow with respect to a representative set $\mathcal{Y}$, which contains a single state from every equivalence class. This yields a normalized flow $\Psi^\tau$ in $\mathcal{Y}$.}
\label{fig2} 
\end{figure}

\begin{definition}
\label{def1}
We call $\mathcal{Y} \subset \mathcal{X}$ a representative set (with respect to the group $\mathcal{H}_{\mathrm{ts}}$), if the following properties are satisfied. For any $x \in \mathcal{X}$, there exists a unique value $a = A(x) > 0$ such that $h^a(x) \in \mathcal{Y}$. The function $A: \mathcal{X} \mapsto \mathbb{R}_+$ is measurable with $\int A \,d\mu < \infty$. 
\end{definition}

Thus, a representative set $\mathcal{Y}$ contains a single state within every equivalence class. From   Definition~\ref{def1} and relation (\ref{eq2_0ha}) it follows that the function $A(x)$ has the property
	\begin{equation}
	A \circ h^a(x) = \frac{A(x)}{a}, \quad A(y) = 1
	\label{eq2_2Ap}
	\end{equation}
for any $h^a \in \mathcal{H}_{\mathrm{ts}}$, $x \in \mathcal{X}$ and $y \in \mathcal{Y}$. We introduce a measurable projector $P: \mathcal{X} \mapsto \mathcal{Y}$ as 
	\begin{equation}
	P(x) = h^{A(x)}(x). 
	\label{eq2_2P}
	\end{equation}
We will need the following known property of invariant measures under a change of time. 

\begin{proposition}[\cite{cornfeld2012ergodic}]
\label{prop_time}
For a positive measurable function $A(x)$, one can introduce a new flow $\Phi^\tau_A$ with a new time $\tau \in \mathbb{R}$ defined by the relations
	\begin{equation}
	\Phi^\tau_A(x) = \Phi^t(x),\quad
	\tau = \int_0^t A\circ \Phi^s (x) ds.
	\label{eq2_A1}
	\end{equation}
The flow $\Phi^\tau_A$ has the invariant measure $\mu_A$, which is absolutely continuous with respect to $\mu$ as 
	\begin{equation}
	\frac{d\mu_A}{d\mu} = \frac{A(x)}{\int Ad\mu}. 
	\label{eq2_A1ac}
	\end{equation}
\end{proposition}

We adopt the subscript notation $\mu_A$ for transformation (\ref{eq2_A1ac}) from now on. In (\ref{eq2_A1}), the function $A(x)$ plays the role of a ``relative speed'' between the original and new times. By construction, $\mu_A$ is a probability measure. For consecutive changes of time with relative speeds $A_1(x)$ and $A_2(x)$, one can verify the relations
	\begin{equation}
	(\mu_{A_1})_{A_2} = \mu_{A}, \quad A(x) = A_1(x)A_2(x).
	\label{eq2_A1r}
	\end{equation}

We now \textit{normalize} the system by reducing the dynamics to the representative set $\mathcal{Y}$. This is the central part of our construction, which yields a normalized flow $\Psi^\tau$ and a corresponding normalized measure $\nu$ on $\mathcal{Y}$ by synchronizing the original time $t$ in $\mathcal{X}$ with the time $\tau$ in $\mathcal{Y}$; see Fig.~\ref{fig2}.  

\begin{theorem}
\label{theorem0}
The map 
	\begin{equation}
	\Psi^\tau(y) = P\circ \Phi_A^\tau(y)
	\label{eq2_AmF}
	\end{equation}
with $y \in \mathcal{Y}$ defines a flow in the representative set. It has the invariant probability measure 
	\begin{equation}
	\nu = P_\sharp \mu_A.
	\label{eq2_Am}
	\end{equation}
\end{theorem}

For all proofs, see Subsection~\ref{sec_2proofs}. Notice that the invariance of measure (\ref{eq2_Am}) is not a trivial fact, because it depends on the measure $\mu_A$ on the full space $\mathcal{X}$ while the flow (\ref{eq2_AmF}) is determined by $\Phi_A^\tau$ restricted to $\mathcal{Y}$. 
The important property of $\nu$ is that it is not affected by temporal scalings:

\begin{proposition} \label{prop_h}
All invariant measures $\widetilde{\mu} = h^a_\sharp \mu$ with $a > 0$ yield the same normalized measure $\nu = P_\sharp \widetilde{\mu}_A$ by Theorem~\ref{theorem0}.
\end{proposition}

In applications, one often explores statistical properties of a system using test functions (also called observables), which are averaged with respect to time for particular solutions or with respect to statistical ensembles. Let us consider measurable functions $\varphi: \mathcal{X} \mapsto \mathbb{R}$ for the original system and $\psi: \mathcal{Y} \mapsto \mathbb{R}$ for the normalized system. We introduce their temporal and ensemble averages as
	\begin{eqnarray}
	&&\displaystyle
	\langle \varphi \rangle_t(x) = \lim_{t \to \infty} \frac{1}{t}\int_0^t \varphi \circ \Phi^s(x) \,ds,\quad
	\langle \varphi \rangle_\mu = \int \varphi\, d\mu,
	\label{eq2_Avr1}
	\\[7pt]
	&&\displaystyle
	\langle \psi \rangle_\tau(y) = \lim_{\tau \to \infty} 
	\frac{1}{\tau}\int_0^\tau \psi \circ \Psi^\sigma(y) \,d\sigma,\quad
	\langle \psi \rangle_\nu = \int \psi\, d\nu,
	\label{eq2_Avr2}
	\end{eqnarray}
where the limits are assumed to exist; in general, the temporal averages depend on the initial state $x$ or $y$.

\begin{proposition} \label{prop_erg}
For averages (\ref{eq2_Avr1}) and (\ref{eq2_Avr2}) the following relations hold
	\begin{equation}
	\langle \psi \rangle_\tau(y) = \frac{\langle \varphi \rangle_t(x)}{\langle A \rangle_t(x)},
	\quad
	\langle \psi \rangle_\nu = \frac{\langle \varphi \rangle_\mu}{\langle A \rangle_\mu},
	\label{eq2_Avr3}
	\end{equation}
where $y = P(x)$, $\varphi(x) = \psi \circ P(x) A(x)$, and averages of $A(x)$ are assumed to be finite and nonzero.
\end{proposition}

This Proposition shows that both temporal and ensemble averages of any observable $\psi(y)$ in the normalized system are related to respective averages of the observable $\varphi(x)$ in the original system. Hence, the normalized system inherits some of ergodic properties of the original flow: if temporal and ensemble averages are equal for $\varphi(x)$ and $A(x)$ in the original system, the same is true for $\psi(y)$ in the normalized system. Recall that, in the definition of SRB (physical) measures~\cite{eckmann1985ergodic}, such equality is assumed for almost all initial states and bounded continuous test functions.

\subsection{Symmetries of the normalized measure}
\label{subsec_SNM}

Here we are going to extend the symmetry group $\mathcal{G}$ to the normalized system. First, let us establish the action of symmetries on the normalized measure $\nu$.

\begin{theorem} \label{prop_g}
Consider invariant measures $\mu$ and $g_\sharp \mu$ of the flow $\Phi^t$ for some $g \in \mathcal{G}$. We denote by $\nu$ and $g_\star \nu$ the corresponding invariant measures of the flow $\Psi^\tau$ given by Theorem~\ref{theorem0}. Then,
	\begin{equation}
	g_\star \nu = (P \circ {g})_\sharp\nu_C, \quad
	C = A \circ g,
	\label{eq2_A7}
	\end{equation}
where $\nu_C$ is an absolutely continuous measure with respect to $\nu$ such that 	
	\begin{equation}
	\frac{d\nu_C}{d\nu} = \frac{C(y)}{\int C d\nu}. 
	\label{eq2_A7b}
	\end{equation}
\end{theorem}

Here (\ref{eq2_A7b}) is the change-of-time transformation (\ref{eq2_A1ac}), which is applied to the normalized measure $\nu$. In the following, we assume $\int C d\nu = \int A \circ g \,d\nu < \infty$ for all $g \in \mathcal{G}$, implying that all measures $g_\star \nu$ exist. By Theorem~\ref{prop_g}, elements of the group $\mathcal{G}$ define transformations of normalized invariant measures through the relation $\nu \mapsto g_\star \nu$, which is a normalized counterpart of the push-forward $\mu \mapsto g_\sharp \mu$ for the original measure. Therefore, $g_\star$ preserves the group structure: 

\begin{corollary}
\label{theorem1}
For any $g$ and $g' \in \mathcal{G}$, we have
	\begin{equation}
	(g' \circ g)_\star\nu = g'_{\star} \big(g_{\star}\nu\big),
	\label{eq2_A7c}
	\end{equation}
where the action of $g_\star$ is defined by (\ref{eq2_A7}) and (\ref{eq2_A7b}).
\end{corollary}

We say that the normalized measure $\nu$ is symmetric with respect to $g$ if $g_\star \nu = \nu$. Combining Proposition~\ref{prop_h} and Theorem~\ref{prop_g}, we see that this relation is not sensitive to temporal scalings:

\begin{corollary}\label{corr_sym_mes}
If the measure $\mu$ is symmetric with respect to a composition $g \circ h^a$ for some $g \in \mathcal{G}$ and $h^a \in \mathcal{H}_{\mathrm{ts}}$, then the normalized measure $\nu$ is symmetric with respect to $g$:
	\begin{equation}
	\label{eq2_A8ex}
	(g \circ h^a)_\sharp \mu = \mu \quad \Rightarrow \quad g_\star \nu = \nu.
	\end{equation}
\end{corollary}

We see that the normalized system inherits the symmetry group $\mathcal{G}$ in the statistical sense. {\color{black}As we mentioned in Section~\ref{subsec_hid}, the normalized measure $\nu$ may be symmetric while $\mu$ is not, manifesting a ``hidden'' form of symmetry.} 

A specific form of the normalized system depends on a choice of the representative set $\mathcal{Y}$. The next statement ensures that all choices are equivalent as far as the symmetry of the normalized measure is concerned.

\begin{theorem}
\label{theorem3}
Assume that the normalized measure $\nu$ from Theorem~\ref{theorem0} is symmetric with respect to $g \in \mathcal{G}$ for some representative set: $g_\star \nu = \nu$. Then the same is true for any representative set.
\end{theorem}

It is useful to express $g_\star \nu$ in terms of the original measure $\mu$.
\begin{proposition}
\label{prop_muC}
Under conditions of Theorem~\ref{prop_g}, the following relation holds:
	\begin{equation}
	g_\star \nu = (P \circ {g})_\sharp\mu_C.
	\label{eq2_A7_mu}
	\end{equation}
\end{proposition}

In summary, we developed a quotient-like construction for the flow $\Phi^t$ with respect to the group of temporal scalings $\mathcal{H}_{\mathrm{ts}}$. It yields the normalized flow $\Psi^\tau$ with the normalized invariant measure $\nu$, which are not sensitive to temporal scalings. Symmetries of the remaining group $\mathcal{G}$ persist in the form of transformations $g_\star \nu$ for normalized  invariant measures. 

\subsection{Proofs of Theorems~\ref{theorem0}--\ref{theorem3} and Propositions~\ref{prop_h}--\ref{prop_muC}}
\label{sec_2proofs}

We will need the following lemmas:

\begin{lemma}
\label{lemmaMs}
For any measurable map $f: \mathcal{X} \mapsto \mathcal{X}$ and positive measurable functions $B: \mathcal{X} \mapsto \mathbb{R}_+$ and $B':\mathcal{X} \mapsto \mathbb{R}_+$  the following relations hold:
	\begin{eqnarray}
	\label{eq2_D2exD}
	(f_\sharp \mu)_B & = & f_\sharp \,\mu_{B \circ f}, \\
	(f_\sharp \mu_B)_{B'} & = & f_\sharp \mu_F, \quad F = (B' \circ f)B.
	\label{eq2_C2g}
 	\end{eqnarray}
\end{lemma}

\begin{proof}
Equality of these measures can be verified by integrating them with a measurable function $\varphi: \mathcal{X} \mapsto \mathbb{R}$.
Using (\ref{eq2_A1ac}) and the classical change-of-variables formula for a push-forward measure, one has
	\begin{equation}
	\label{eq2_prL1_1}
	\int \varphi \, d(f_\sharp \mu)_B 
	= \frac{\int \varphi B  \, d(f_\sharp \mu)}{\int B \, d(f_\sharp \mu)} 
	= \frac{\int (\varphi  \circ f) (B \circ f)  \, d\mu}{\int B \circ f \, d\mu} 
	= \int \varphi \circ f \, d\mu_{B \circ f}
	= \int \varphi\, d\left(f_\sharp \,\mu_{B \circ f}\right),
 	\end{equation}
proving (\ref{eq2_D2exD}). Equality (\ref{eq2_C2g}) is obtained by combining (\ref{eq2_A1r}) and (\ref{eq2_D2exD}).
\end{proof}

\begin{lemma}
\label{lemma1}
The maps $\Phi_A^\tau$ and $h^a$ commute for any $a > 0$ and time $\tau$.
\end{lemma}

\begin{proof}
Using expressions (\ref{eq2_A1}) and (\ref{eq2_2}), we write 
	\begin{equation}
	h^a \circ \Phi_A^\tau(x)  = h^a \circ \Phi^t(x)  = \Phi^{at} \circ h^a(x).
	\label{eq2_A3}
	\end{equation}
Similarly, using (\ref{eq2_A1}) {\color{black}for the state $x_1 = h^a(x)$, we express
	\begin{equation}
	\Phi_A^\tau \circ h^a(x) 
	= \Phi_A^\tau(x_1)
	= \Phi^{t_1}(x_1)
	= \Phi^{t_1} \circ h^a(x),
	\label{eq2_A4}
	\end{equation}
}where the time $t_1$ is determined by the equation 
	\begin{equation}
	\tau 
	= \int_0^{t_1} A\circ \Phi^s(x_1)ds 
	= \int_0^{t_1} A\circ \Phi^s \circ h^a(x)ds.
	\label{eq2_A5}
	\end{equation}
Using (\ref{eq2_2}) and (\ref{eq2_2Ap}) in (\ref{eq2_A5}), we obtain
	\begin{equation}
	\tau
	= \int_0^{t_1} A\circ h^a \circ \Phi^{s/a}(x)ds
	= \int_0^{t_1} A\circ \Phi^{s/a}(x)\,\frac{ds}{a}
	= \int_0^{t_1/a} A\circ \Phi^{s'}(x) ds',
	\label{eq2_A6}
	\end{equation}
where the last equality follows from the change of integration variable $s' = s/a$.
Comparing (\ref{eq2_A6}) with the second expression in (\ref{eq2_A1}), we find $t = t_1/a$.  Then, expressions (\ref{eq2_A3}) and (\ref{eq2_A4}) yield the commutativity property $h^a \circ \Phi_A^\tau = \Phi_A^\tau \circ h^a$.
\end{proof}

\begin{proof}[Proof of Theorem~\ref{theorem0}]
Let us show that
	\begin{equation}
	\label{eq2_D2}
	P\circ \Phi_A^\tau \circ P(x) = P\circ \Phi_A^\tau(x)
	\end{equation}
for any $x \in \mathcal{X}$. In the left-hand side, we use (\ref{eq2_2P}) and the commutation relation of Lemma~\ref{lemma1}, which yields $P\circ \Phi_A^\tau \circ P(x) = P\circ h^{A(x)} \circ \Phi_A^\tau(x)$. Then, equality (\ref{eq2_D2}) follows from the projector property (see Definition~\ref{def1})
	\begin{equation}
	\label{eq2_PP}
	P \circ h^a = P.
	\end{equation}

By definitions (\ref{eq2_AmF}) and (\ref{eq2_Am}), we have 
	\begin{equation}
	\label{eq2_D1ex}
	\Psi^\tau_\sharp \nu 
	= (P\circ \Phi_A^\tau)_\sharp \left(P_\sharp \mu_A\right)
	= (P\circ \Phi_A^\tau \circ P)_\sharp \mu_A
	= (P\circ \Phi_A^\tau)_\sharp \mu_A
	= P_\sharp \mu_A = \nu,	
	\end{equation}
where we used (\ref{eq2_D2}) and the invariance of the measure $\mu_A$ for the flow $\Phi_A^\tau$ by Proposition~\ref{prop_time}. Hence, the measure $\nu$ is invariant for the normalized flow $\Psi^\tau$.

It remains to prove the property $\Psi^{\tau_1}\circ\Psi^{\tau_2} = \Psi^{\tau_1+\tau_2}$. Using definition (\ref{eq2_AmF}), we  have
	\begin{equation}
	\label{eq2_D1}
	\Psi^{\tau_1}\circ\Psi^{\tau_2} 
	= P\circ \Phi_A^{\tau_1} \circ P \circ \Phi_A^{\tau_2} 
	= P\circ \Phi_A^{\tau_1} \circ \Phi_A^{\tau_2} 
	= P\circ \Phi_A^{\tau_1+\tau_2}
	= \Psi^{\tau_1+\tau_2},
	\end{equation}
where we used (\ref{eq2_D2}) and the flow relation $\Phi_A^{\tau_1} \circ \Phi_A^{\tau_2} = \Phi_A^{\tau_1+\tau_2}$. 
\end{proof}

\begin{proof}[Proof of Proposition \ref{prop_h}]
By Theorem~\ref{theorem0}, the normalized measure for $\widetilde{\mu} = h^a_\sharp \mu$ is found as
	\begin{equation}
	\label{eq2_D2ex}
	\widetilde{\nu} = P_\sharp \widetilde{\mu}_A
	= P_\sharp \left(h^a_\sharp \mu\right)_A.
 	\end{equation}
Taking $f = h^a$ and $B = A$ in (\ref{eq2_D2exD}), we obtain
	\begin{equation}
	\label{eq2_D2exC}
	\left(h^a_\sharp \mu\right)_A = h^a_\sharp \mu_{A\circ h^a} 
	= h^a_\sharp \mu_{A/a} = h^a_\sharp \mu_A,
 	\end{equation}
where we used (\ref{eq2_2Ap}) and the observation that dividing by a constant $a$ in $A(x)/a$ does not change the measure (\ref{eq2_A1ac}). Using (\ref{eq2_D2exC}) in (\ref{eq2_D2ex}) yields
	\begin{equation}
	\label{eq2_D3ex}
	\widetilde{\nu} 
	= P_\sharp \left(h^a_\sharp \mu_A\right) 
	= (P \circ h^a)_\sharp \mu_A = P_\sharp \mu_A = \nu, 
 	\end{equation}
where we used the projector property (\ref{eq2_PP}).
\end{proof}

\begin{proof}[Proof of Proposition \ref{prop_erg}]
Using (\ref{eq2_AmF}), we have
	\begin{equation}
	\label{eq2_prP3_1}
	\int_0^\tau \psi \circ \Psi^\sigma(y) \,d\sigma
	= \int_0^\tau \psi \circ P\circ \Phi_A^\sigma(y) \,d\sigma.
	\end{equation}
By Proposition \ref{prop_time}, we substitute $\Phi^\sigma_A(y) = \Phi^s(y)$ and $d\sigma = A\circ \Phi^s (y) ds$. This yields
	\begin{equation}
	\label{eq2_prP3_2}
	\int_0^\tau \psi \circ \Psi^\sigma(y) \,d\sigma
	= \int_0^t \psi \circ P\circ \Phi^s(y) \,A \circ \Phi^s(y)\,ds,
	\end{equation}
where 
	\begin{equation}
	\label{eq2_prP3_3}
	\tau = \int_0^t A \circ \Phi^s(y) \,ds.
	\end{equation}

Taking $y = P(x) = h^{A(x)}(x)$ from (\ref{eq2_2P}) and using commutation relations (\ref{eq2_2}), we reduce (\ref{eq2_prP3_2}) to the form
	\begin{equation}
	\label{eq2_prP3_4}
	\begin{array}{rcl}
	\displaystyle
	\int_0^\tau \psi \circ \Psi^\sigma(y) \,d\sigma
	& = & {\color{black}
	\displaystyle
	\int_0^t \psi \circ P \circ \Phi^s \circ h^{A(x)}(x) \,
	A \circ \Phi^{s} \circ h^{A(x)}(x)\,ds 
	}
	\\[12pt]
	& = & 
	\displaystyle
	\int_0^t \psi \circ P \circ h^{A(x)} \circ \Phi^{s/A(x)}(x) \,
	A \circ h^{A(x)} \circ \Phi^{s/A(x)}(x)\,ds 
	\\[12pt]
	& = & 
	\displaystyle
	\int_0^t \psi \circ P\circ \Phi^{s/A(x)}(x) \,A \circ \Phi^{s/A(x)}(x)\,\frac{ds}{A(x)},
	\end{array}
	\end{equation}
where the {\color{black}third} equality follows from properties (\ref{eq2_2Ap}) and (\ref{eq2_PP}).
Denoting $\varphi = (\psi \circ P) A$ and performing the linear change of time $s' = s/A(x)$, expression (\ref{eq2_prP3_4}) becomes (dropping the primes)
	\begin{equation}
	\label{eq2_prP3_5}
	\int_0^\tau \psi \circ \Psi^\sigma(y) \,d\sigma
	= \int_0^T \varphi \circ \Phi^s(x) \,ds, \quad T = \frac{t}{A(x)}.
	\end{equation}	
Similarly, (\ref{eq2_prP3_3}) is reduced to the form
	\begin{equation}
	\label{eq2_prP3_6}
	\tau = \int_0^T A \circ \Phi^s(x) \,ds.
	\end{equation}
Since the average $\langle A \rangle_t(x)$ of a positive function $A$ is assumed to be finite and nonzero, one can see from (\ref{eq2_prP3_6}) and (\ref{eq2_Avr1}) that the limit $T \to \infty$ implies $\tau \to \infty$ and vice versa. Hence, using (\ref{eq2_prP3_5}) and (\ref{eq2_prP3_6}), we have 
	\begin{equation}
	\label{eq2_prP3_7}
	\langle \psi \rangle_\tau (y) 
	= \lim_{\tau \to \infty} \frac{1}{\tau}\int_0^\tau \psi \circ \Psi^\sigma(y) \,d\sigma
	= \lim_{T \to \infty}
	\frac{\frac{1}{T}\int_0^{T}\varphi \circ \Phi^s(x) \,ds}{\frac{1}{T}
	\int_0^{T}A \circ \Phi^s(x) \,ds}
	=\frac{\langle \varphi \rangle_t (x)}{\langle A \rangle_t (x)}.
	\end{equation}

Using (\ref{eq2_Am}) in (\ref{eq2_Avr1})--(\ref{eq2_Avr2}) with the change-of-variables formula for the push-forward measure, we have
	\begin{equation}
	\label{eq2_prP3_B1}
	\langle \psi \rangle_\nu 
	= \int \psi\, d\nu = \int \psi \, d\left(P_\sharp \mu_A\right)
	= \int \psi \circ P \, d \mu_A
	= \frac{\int (\psi \circ P) A \, d \mu}{\int A\,d\mu} 
	= \frac{\langle \varphi \rangle_\mu}{\langle A \rangle_\mu}, 
	\end{equation}
where the last two equalities follow from (\ref{eq2_A1ac}) and $\varphi(x) = \psi \circ P(x)\, A(x)$.
\end{proof}

\begin{proof}[Proof of Theorem \ref{prop_g}]
We first derive two simple identities. The first is
	\begin{equation}
	P\circ g \circ P(x) = P\circ g \circ h^{A(x)}(x) = P \circ g(x),
	\label{eq2_B2x}
	\end{equation}
where we substituted (\ref{eq2_2P}) and used commutation relation (\ref{eq2_2com}) with projector property (\ref{eq2_PP}). The second identity is
	\begin{equation}
	C \circ P(x)A(x) = A \circ g \circ h^{A(x)}(x)A(x) = A \circ h^{A(x)} \circ g(x)A(x) 
	= A \circ g(x) = C(x),
	\label{eq2_H6}
	\end{equation}
where we used sequentially $C = A \circ g$, (\ref{eq2_2P}), (\ref{eq2_2com}) and (\ref{eq2_2Ap}). 

By Theorem~\ref{theorem0} applied to the measure $\widetilde{\mu} = g_\sharp \mu$, we have
	\begin{equation}
	\label{eq2_A7ex}
	g_\star \nu 
	= P_\sharp \,\widetilde{\mu}_A 
	= P_\sharp \left(g_\sharp \mu\right)_A 
	= P_\sharp \left( {g}_\sharp\mu_C\right)
	= (P \circ g)_\sharp \mu_C, 
	\end{equation}
where we used the equality $(g_\sharp \mu)_A = g_\sharp \mu_C$ following from general relation (\ref{eq2_D2exD}) with $C = A \circ g$. Using (\ref{eq2_B2x}) in (\ref{eq2_A7ex}), we write
	\begin{equation}
	\label{eq2_A7exM}
	g_\star \nu = 
	{\color{black}
	(P \circ g \circ P)_\sharp \mu_C}
	= (P \circ g)_\sharp \big( P_\sharp\mu_C \big).
	\end{equation}
Similarly, using (\ref{eq2_Am}) we express
	\begin{equation}
	\label{eq2_A7exN}
	\nu_C = (P_\sharp \mu_A)_C
	= P_\sharp \mu_F, \quad F = (C\circ P)\,A,
	\end{equation}
where we used relation (\ref{eq2_C2g}) written for the measure $\nu$ with $f = P$, $B = A$ and $B' = C$. Notice that $F = C$ by the identity (\ref{eq2_H6}). Hence, we obtain (\ref{eq2_A7}) by combining (\ref{eq2_A7exM}) and (\ref{eq2_A7exN}).
\end{proof}

\begin{proof}[Proof of Proposition \ref{prop_muC}]
Expression (\ref{eq2_A7_mu}) has been verified in (\ref{eq2_A7ex}).
\end{proof}

\begin{proof}[Proof of Theorem \ref{theorem3}]
Let us consider two different representative sets, $\mathcal{Y}$ and $\widetilde{\mathcal{Y}}$, with the corresponding projectors, $P(x) = h^{A(x)}(x) \in \mathcal{Y}$ and $\widetilde{P}(x) = h^{\widetilde{A}(x)}(x) \in \widetilde{\mathcal{Y}}$. By Theorem~\ref{theorem0}, the normalized measures are expressed as
	\begin{equation}
	\nu = P_\sharp \mu_A, \quad
	\widetilde{\nu} = \widetilde{P}_\sharp \mu_{\widetilde{A}}.
	\label{eq2_H3}
	\end{equation}
Using $\nu$ from (\ref{eq2_H3}) and $g_\star \nu$ from (\ref{eq2_A7_mu}), we write the symmetry condition $\nu = g_\star \nu$ as
	\begin{equation}
	P_\sharp \mu_A = (P \circ g)_\sharp \mu_C.
	\label{eq2_H4y}
	\end{equation}
Similarly, for the second representative set $\widetilde{\mathcal{Y}}$, the symmetry condition $\widetilde{\nu} = g_\star \widetilde{\nu}$ is equivalent to
	\begin{equation}
	\widetilde{P}_\sharp \mu_{\widetilde{A}} = (\widetilde{P} \circ {g})_\sharp \mu_{\widetilde{C}}.
	\label{eq2_H5z}
	\end{equation}
The proof will be completed by deriving (\ref{eq2_H5z}) from (\ref{eq2_H4y}).

Changing time in both sides of (\ref{eq2_H4y}) with the relative speed $\widetilde{A}(x)$ {\color{black} yields 
	\begin{equation}
	\left(P_\sharp \mu_A\right)_{\widetilde{A}} = \left((P \circ {g})_\sharp \mu_C\right)_{\widetilde{A}}.
	\label{eq2_H7pre}
	\end{equation}
Using (\ref{eq2_C2g}) in both sides of (\ref{eq2_H7pre}),} we have
	\begin{equation}
	P_\sharp \mu_F = (P \circ {g})_\sharp \mu_H,
	\label{eq2_H7}
	\end{equation}
where
	\begin{equation}
	F = (\widetilde{A}\circ P)\,A, \quad
	H = (\widetilde{A} \circ P\circ g)\,C.
	\label{eq2_H8}
	\end{equation}
The function $F(x)$ is expressed using (\ref{eq2_2P}) and (\ref{eq2_2Ap}) as
	\begin{equation}
	F(x) = \widetilde{A}\circ h^{A(x)}(x)A(x)
	= \widetilde{A}(x).
	\label{eq2_H9}
	\end{equation}
Writing $C(x) = A \circ g(x) = A(x_g)$ with $x_g = g(x)$, we similarly express the function $H(x)$ as
	\begin{equation}
	H(x) 
	= \widetilde{A} \circ P(x_g)\, A(x_g) 
	= \widetilde{A} \circ h^{A(x_g)}(x_g)\, A(x_g) 
	= \widetilde{A}(x_g)
	= \widetilde{A} \circ g (x)
	= \widetilde{C}(x).
	\label{eq2_H10}
	\end{equation}
Combining (\ref{eq2_H7}) with (\ref{eq2_H9}) and (\ref{eq2_H10}), we have
	\begin{equation}
	P_\sharp \mu_{\widetilde{A}} = (P \circ {g})_\sharp \mu_{\widetilde{C}}.
	\label{eq2_H12}
	\end{equation}
Applying the push-forward $\widetilde{P}_\sharp$ in both sides of this expression and using the relation $\widetilde{P}\circ P = \widetilde{P}$ analogous to (\ref{eq2_PP}), yields the required identity (\ref{eq2_H5z}).
\end{proof}

\section{Hidden scaling symmetry in a shell model of turbulence}
\label{sec_shell1}

In this section we consider a popular toy-model, called a shell model, which mimics turbulent dynamics of incompressible three-dimensional Navier--Stokes equations~\cite{gledzer1973system,ohkitani1989temporal,biferale2003shell}.
It is represented by complex variables $u_n \in \mathbb{C}$ called  \textit{shell velocities} and indexed by integer \textit{shell numbers} $n$. Shell velocities are interpreted as amplitudes of velocity fluctuations at wavenumbers $k_n = 2^n$. Thus, small wavenumbers (smaller $n$) describe large-scale motion and large wavenumbers (larger $n$) correspond to small-scale dynamics. Equations of motion are constructed in analogy with the Navier--Stokes system (preserving some of its symmetries and global inviscid invariants) and take the form~\cite{l1998improved}
	\begin{equation}
	\label{eq3_1}
	\frac{du_n}{dt} = B_n-\mathrm{Re}^{-1} k_n^2u_n+f_n, \quad n \ge 0.
	\end{equation}
Here $B_n$ is the quadratic nonlinear term
	\begin{equation}
	\label{eq3_1b}
	B_n = \left\{
	\begin{array}{ll}
	i(k_{n+1}u_{n+2}u_{n+1}^*
	-k_{n-1}u_{n+1}u_{n-1}^* 
	+k_{n-2}u_{n-1}u_{n-2}), & n > 1;\\[3pt]
	i(k_2u_3u_2^*-k_0u_2u_0^*), & n = 1;\\[3pt]
	ik_1u_2u_1^*, & n = 0,
	\end{array}
	\right.
	\end{equation}
where $n = 0$ and $1$ are ``boundary'' shell numbers, $i$ is the imaginary unit, and the asterisks denote complex conjugation.
We consider constant (time independent) forcing terms $f_n$, which are nonzero only for the boundary shells $n = 0$ and $1$. Equations (\ref{eq3_1}) are written in non-dimensional form with characteristic integral scales set to unity. The viscous term $\mathrm{Re}^{-1} k_n^2u_n$ is multiplied by the inverse of the dimensionless Reynolds number $\mathrm{Re} > 0$. 

Along with (\ref{eq3_1}), we consider a shell model for the Euler equations of ideal flow. It is given by the equations
	\begin{equation}
	\label{eq3_2}
	\frac{du_n}{dt} = i\left(k_{n+1}u_{n+2}u_{n+1}^*
	-k_{n-1}u_{n+1}u_{n-1}^* 
	+k_{n-2}u_{n-1}u_{n-2}\right), \quad n \in \mathbb{Z},
	\end{equation}
where variables $u_n$ are introduced for all integer shell numbers $n$. Equations (\ref{eq3_2}) are obtained from (\ref{eq3_1}) and (\ref{eq3_1b}) for $n > 1$ after removing the forcing and viscous terms. We refer to~\cite{constantin2007regularity} for analytical properties of equations (\ref{eq3_1})--(\ref{eq3_2}), including the issues of existence and uniqueness of solutions. 
	
\subsection{Symmetries} 
\label{subsec_shell_sym}

In this subsection, we present the formal analysis of scaling symmetries for the ideal system (\ref{eq3_2}). The state variable $x = (u_n)_{n \in \mathbb{Z}}$ consists of all shell velocities. We assume the existence of a flow $\Phi^t: \mathcal{X} \mapsto \mathcal{X}$ in a properly defined configuration space $x \in \mathcal{X}$. Having a solution $\big(u_n(t)\big)_{n\in\mathbb{Z}} = \Phi^t(x)$ of (\ref{eq3_2}), new solutions are given by
	\begin{equation}
	\begin{array}{rll}
		\textrm{temporal scaling:}& 
		u_n(t) \mapsto u_n(t/a)/a,&
		a > 0;
		\\[2pt]
		\textrm{spatial scaling:}& 
		u_n(t) \mapsto k_m u_{n+m} (t),&
		m \in \mathbb{Z}.
	\end{array}	
	\label{eq3_3}
	\end{equation}
In terms of the state $x$ considered at initial time $t = 0$, relations (\ref{eq3_3}) define the mappings $h^a: \mathcal{X} \mapsto \mathcal{X}$ and $g^m: \mathcal{X} \mapsto \mathcal{X}$ acting on each shell velocity as
	\begin{eqnarray}
		x' = h^a(x), && \quad
		u'_n = u_n/a, \quad
		a > 0;
		\label{eq3_4a}
		\\[3pt]
		x' = g^m(x),&& \quad
		u_n' = k_m u_{n+m},\quad
		m \in \mathbb{Z}.
	\label{eq3_4b}
	\end{eqnarray}
Notice that $h^{a_1}\circ h^{a_2} = h^{a_1a_2}$ and $g^{m_1}\circ g^{m_2} = g^{m_1+m_2}$. These maps generate the two groups 
	\begin{equation}
	\mathcal{H}_{\mathrm{ts}} = \{h^a: a > 0\}, \quad
	\mathcal{G} = \{g^m: m \in \mathbb{Z}\}.
 	\label{eq3_4ex}
	\end{equation}
We will write $g^1 = g$, which represents the primary spatial scaling with the {\color{black}unit} change of shell numbers $n \mapsto n+1$. It is straightforward to see from (\ref{eq3_3})--(\ref{eq3_4b}) that the flow $\Phi^t$ and elements of the groups $\mathcal{H}_{\mathrm{ts}}$ and $\mathcal{G}$ satisfy composition and commutation relations (\ref{eq2_0ha})--(\ref{eq2_2}). Hence the theory of Section~\ref{sec4} applies to the shell model.

\subsection{Normalized system}\label{subsec_shell_norm}

The representative set $\mathcal{Y}$ is defined by a positive function $A(x)$ satisfying the homogeneity property (\ref{eq2_2Ap}). Given $x = \big(u_n\big)_{n\in\mathbb{Z}} \in \mathcal{X}$ the corresponding representative state $y = \big(U_n\big)_{n\in\mathbb{Z}} \in \mathcal{Y}$ is determined by the projector (\ref{eq2_2P}) as
	\begin{equation}
	\label{eq3_A2}
	y = P(x),\quad
	U_n = \frac{u_n}{A(x)}.
	\end{equation}
As an example, we consider
 	\begin{equation}
	\label{eq3_B1}
	A(x) = \sqrt{\sum_{n < 0} k_n^2|u_n|^2}.
	\end{equation}
For turbulent solutions, the sum in (\ref{eq3_B1}) converges as a geometric progression with the main contribution from the largest (close to zero) shells numbers~\cite{mailybaev2020hidden}.

Given a solution $\big(u_n(t)\big)_{n\in\mathbb{Z}} = \Phi^t(x)$ of system (\ref{eq3_2}), we now derive formally the equations for the normalized solution $\big(U_n(\tau)\big)_{n\in\mathbb{Z}} = \Psi^\tau(y)$. The normalized flow is defined by Theorem~\ref{theorem0} as $\Psi^\tau = P\circ \Phi_A^\tau$, which depends on the synchronized time given by expression (\ref{eq2_A1}) as
	\begin{equation}
	\tau = \int_0^t A \circ \Phi^s(x)\, ds.
	\label{eq3_6}
	\end{equation}
Using expressions (\ref{eq3_A2})--(\ref{eq3_6}) in (\ref{eq3_2}), after a long but elementary derivation one obtains~\cite{mailybaev2020hidden} 
	\begin{equation}
	\label{eq3_7}
	\begin{array}{rcl}
	\displaystyle
	\frac{dU_n}{d\tau} 
	& = & \displaystyle
	i\left(k_{n+1}U_{n+2}U_{n+1}^*
	-k_{n-1}U_{n+1}U_{n-1}^* 
	+k_{n-2}U_{n-1}U_{n-2}\right)
	\\[10pt]
	&& \displaystyle
	+\,U_n\sum_{j < 0} k_j^3
	\left(2\pi_{j+1}-\frac{\pi_j}{2}
	-\frac{\pi_{j-1}}{4}\right),\quad
	\pi_j = \mathrm{Im}\left(U_{j-1}^*U_j^*\,U_{j+1}\right).
	\end{array}
	\end{equation}
These are equations satisfied by solutions $U_n(\tau)$ of the normalized system.
The condition $A(y) = 1$ on the representative set is written using (\ref{eq3_B1}) as 
	\begin{equation}
	\sum_{n < 0} k_n^2|U_n|^2 = 1.
	\label{eq3_7A}
	\end{equation}
One can check that this condition is invariant for system (\ref{eq3_7}).

Now let us analyze statistical symmetries of the normalized system. By Theorem~\ref{theorem0}, the invariant measure $\mu$ of the flow $\Phi^t$ in the original system (\ref{eq3_2}) yields the invariant measure 
	\begin{equation}
	\nu  = P_\sharp \mu_A 
	\label{eq3_7B2}
	\end{equation}
of the flow $\Psi^\tau$ in the normalized system (\ref{eq3_7}). For any scaling map $g^m \in \mathcal{G}$, Theorem~\ref{prop_g} and Proposition~\ref{prop_muC} yield the new invariant normalized measure as
	\begin{equation}
	g^m_\star \nu 
	= (P \circ g^m)_\sharp\nu_C 
	= (P \circ g^m)_\sharp\mu_C, 
	\quad C = A \circ g^m.
	\label{eq3_7B}
	\end{equation}

The transformation $\nu \mapsto g_\star\nu$ can be associated with changes of variables.
Indeed, expressions (\ref{eq3_7B}) imply transformations of state and time in the form
	\begin{equation}
	\begin{array}{l}
	y \mapsto y^{(m)} = P \circ g^m(y), 
	\\[3pt]
	d\tau \mapsto d\tau^{(m)} = A \circ g^m(y)\,d\tau. 
	\end{array}
	\label{eq3_7D}
	\end{equation}
Using projector (\ref{eq3_A2}) and the scaling map from (\ref{eq3_4b}), the first relation of (\ref{eq3_7D}) is written as
	\begin{equation}
	U_n \mapsto U_n^{(m)} 
	= \frac{k_m U_{n+m}}{A \circ g^m(y)}.
	\label{eq3_7E}
	\end{equation}
Using (\ref{eq3_4b}), (\ref{eq3_B1}) and (\ref{eq3_7A}), we derive 
	\begin{equation}
	\label{eq3_B4}
	A \circ g^m(y) 
	= \sqrt{\sum_{n < m} k_n^2|U_n|^2}
	= \left\{
	\begin{array}{ll}
		\displaystyle
		\bigg(1+\sum_{0 \le n < m} k_n^2|U_n|^2\bigg)^{1/2}, & m > 0; \\[15pt]
		1, & m = 0; \\[10pt]
		\displaystyle
		\bigg(1-\sum_{m \le n < 0} k_n^2|U_n|^2\bigg)^{1/2}, & m < 0.
	\end{array}
	\right.
	\end{equation}
One can check that these transformations define symmetries of the normalized system, i.e., equations (\ref{eq3_7}) are satisfied by the velocities $U_n^{(m)}$ as functions of the time $\tau^{(m)}$ for any $m$, provided that $U_n(\tau)$ satisfy (\ref{eq3_7}). Notice that, though symmetries (\ref{eq3_4b}) of the original system are linear, the respective transformation (\ref{eq3_7D})--(\ref{eq3_B4}) for the normalized system is nonlinear.

\subsection{Hidden scaling symmetry}\label{subsec_conj}

Let us return to the original system (\ref{eq3_1}) with forcing and viscous terms, which is used to model the developed turbulence for large Reynolds numbers, $\mathrm{Re} \gg 1$. Dynamics of this model is not yet well understood theoretically, featuring important open problems of turbulence theory. The widely accepted conjecture~\cite{frisch1999turbulence,biferale2003shell} is that shell variables can be grouped into three different ranges (see Fig.~\ref{figHS}): the range of low wavenumbers $k_n \sim 1$ (small $n$) in which the forces are applied, the range of large wavenumbers $k_n \gtrsim \mathrm{Re}^{3/4}$ (large $n$) in which the dynamics is dominated by the viscous term, and the intermediate range of wavenumbers
	\begin{equation}
	\label{eq3_2ex}
	1 \ll k_n \ll \mathrm{Re}^{3/4}
	\end{equation}
called the inertial interval. In the inertial interval, the forcing and viscous terms are negligible, which yields the equations of ideal system (\ref{eq3_2}). 

\begin{figure}
\centering
\includegraphics[width=0.8\textwidth]{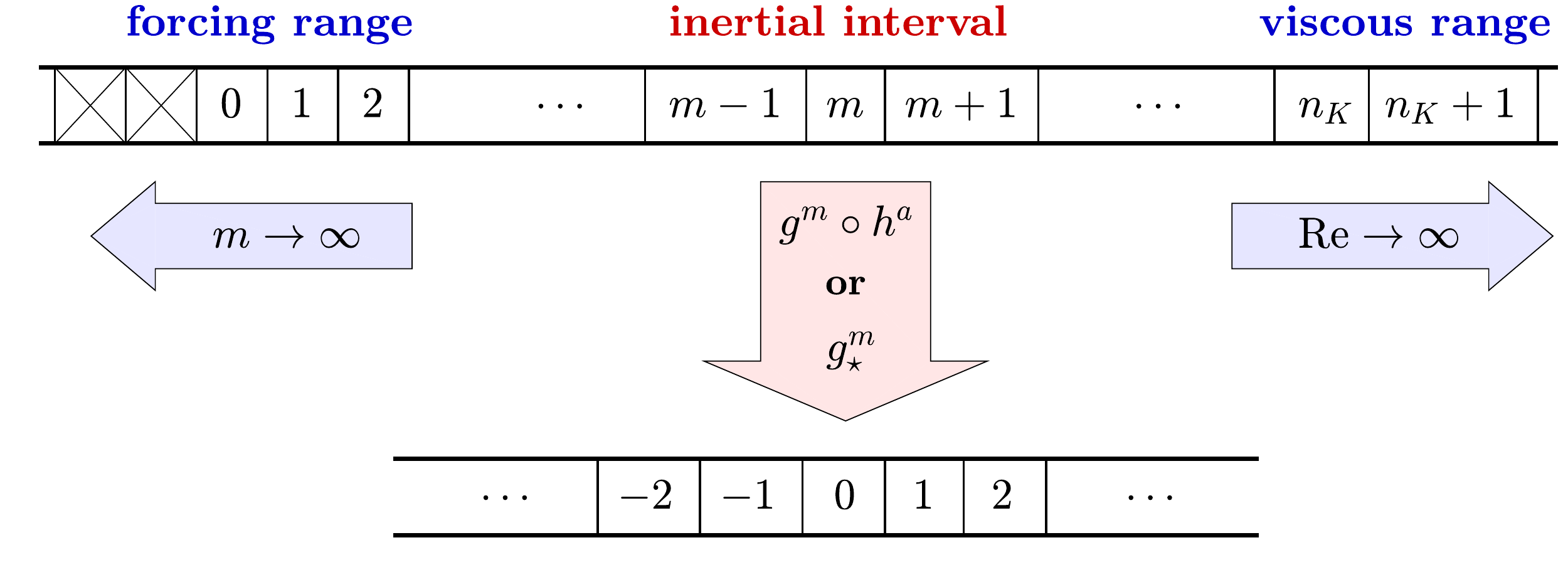}
\caption{Shell model describes a wide inertial interval, which separates a forcing range (shell numbers around zero) from a viscous range (shell numbers around $n_K = \log_2 \mathrm{Re}^{3/4}$). The double limit (\ref{eq3_T2lim}) extends the inertial interval to all shells. For invariant probability measures, this limit is expressed using spatiotemporal scalings $g^m \circ h^a$ in the original system or normalized scalings $g^m_\star$ in the normalized system.}
\label{figHS}
\end{figure}

For the scaling analysis, it is convenient to define shell variables for all $n \in \mathbb{Z}$. This can be done by assigning some constant values to $u_n$ with $n < 0$; see the crossed cells in Fig.~\ref{figHS}. We consider the system in a statistical equilibrium, i.e., assuming the existence of invariant measures denoted by $\mu^{\mathrm{Re}}$ and dependent on the Reynolds number. 
Using (\ref{eq3_4a}) and (\ref{eq3_4b}) we introduce the measure obtained by a combination of spatial and temporal scalings as
	\begin{equation}
	\label{eq3_T2a}
	s_\sharp \mu^{\mathrm{Re}}, \quad 
	s = g^m \circ h^a.
	\end{equation}
Notice that the choice $a = k_m^{2/3}$ corresponds to the scaling $u_n \mapsto k_m^{1/3}u_{n+m}$ assumed in the Kolmogorov theory~\cite{frisch1999turbulence,biferale2003shell}.
According to (\ref{eq3_4b}), the transformation $g^m$ performs the shift of shell numbers $n \mapsto n+m$. Hence, the inertial interval (\ref{eq3_2ex}) for measure (\ref{eq3_T2a}) is given by the conditions 
	\begin{equation}
	\label{eq3_C1ex}
	1 \ll k_{n+m} \ll \mathrm{Re}^{3/4}.
	\end{equation}
This condition is satisfied asymptotically for any $n \in \mathbb{Z}$ by considering the double limit 
	\begin{equation}
	\label{eq3_T2lim}
	\lim_{m \to \infty} \lim_{{\mathrm{Re}} \to \infty} 
	\end{equation}
Here, the first limit $\mathrm{Re} \to \infty$ moves the viscous range $k_{n+m} \sim \mathrm{Re}^{3/4}$ to infinitely large positive shells $n \to +\infty$, and the second limit $m \to \infty$ moves the forcing range $k_{n+m} \sim 1$ to infinitely large negative shells $n \to -\infty$; see Fig.~\ref{figHS}. Notice the importance of the limit order in this argument. 

As we mentioned above, the dynamics in the inertial interval is governed by the ideal system (\ref{eq3_2}). Hence, we can use the symmetry properties described in Sections~\ref{subsec_shell_sym} and \ref{subsec_shell_norm}. It is known that spatiotemporal scaling symmetries $s = g^m \circ h^a$ from (\ref{eq3_T2a}) are all broken in the inertial interval as a consequence of the intermittency phenomenon~\cite{frisch1999turbulence,biferale2003shell}, which signifies that the limit (\ref{eq3_T2lim}) of the measure (\ref{eq3_T2a}) does not exist. We now argue, that a similar limit may exist for the scaled normalized measure, which is defined according to (\ref{eq3_7B}) as
	\begin{equation}
	\label{eq3_T2bRe}
	g_\star^m \nu^{\mathrm{Re}} 
	= (P \circ g^m)_\sharp\mu^{\mathrm{Re}}_C,
	\quad C = A \circ g^m.
	\end{equation}
Precisely, the asymptotic symmetry condition is formulated as 

\begin{definition}
\label{conj1}
We say that the statistical stationary state of the shell model has a hidden scaling symmetry if the double limit  
	\begin{equation}
	\label{eq3_T2}
	\nu^\infty = \lim_{m \to \infty} \lim_{{\mathrm{Re}} \to \infty} 
	g^m_\star\nu^{\mathrm{Re}}
	\end{equation}
converges weakly {\color{black}(for a proper, e.g., standard product topology).} The limiting measure $\nu^\infty$ is symmetric: 
	\begin{equation}
	\label{eq3_T2b}
	g_\star \nu^\infty = \nu^\infty.
	\end{equation}
\end{definition}

Notice that (\ref{eq3_T2b}) follows from (\ref{eq3_T2}) and the group property (\ref{eq2_A7c}) in Corollary~\ref{theorem1}. We consider the convergence in (\ref{eq3_T2}) as a conjecture. Despite we are unable to prove it (the limit of high Reynolds numbers is still not well understood for the shell model), the hidden scaling symmetry can be tested by numerical simulations.

\subsection{Numerical results}\label{subsec_shell_nr}

Here we present a brief account of numerical results supporting the conjecture of hidden scaling symmetry; we refer to~\cite{mailybaev2020hidden} for further details on numerical simulations and statistical analysis. For approximating the limit (\ref{eq3_T2}), we took the very high Reynolds number $\mathrm{Re} = 2.5 \times 10^{11}$ leading to the large inertial interval $1 \ll k_m \ll \mathrm{Re}^{3/4} \approx k_{28}$. Equations (\ref{eq3_1}) and (\ref{eq3_1b}) with the forcing terms $f_0 = 2f_1$ and $f_1 = 1+i$ were integrated numerically for the variables $u_0,\ldots,u_{39}$ (with $u_n = 0$ for $n \ge 40$) in the large time interval $0 \le t \le 100$. 

Statistical properties of the normalized measure $g^m_\star\nu^{\mathrm{Re}}$ from (\ref{eq3_T2bRe}) can be accessed using Proposition~\ref{prop_erg}, which relates averages in the original and normalized system; see (\ref{eq2_Avr1})--(\ref{eq2_Avr3}). The results presented below are obtained by means of temporal averages, assuming that the temporal and statistical ensemble averages are equal (ergodicity property); the latter is a usual though not rigorously proven assumption. Using relations of Section~\ref{subsec_shell_norm}, analysis of the normalized measure $g^m_\star\nu^{\mathrm{Re}}$ reduces to computing temporal averages of the normalized and scaled shell velocities $U_n^{(m)}$ as functions of the normalized and scaled times $\tau^{(m)}$. Since the Reynold number is already taken very large, we test the convergence of the limit (\ref{eq3_T2}) by verifying that probability density functions (PDFs) of the variables $U_n^{(m)}$ do not depend on $m$ in the inertial interval. 

Figure~\ref{fig3} shows PDFs for the normalized velocities $U_{-2}^{(m)}$, $U_{-1}^{(m)}$, $U_0^{(m)}$ and $U_1^{(m)}$ for ten different values $m = 12,\ldots,21$ chosen in the central part of the inertial interval. The coincidence of curves for different $m$ provides a clear evidence of convergence. Figure~\ref{fig4} compares PDFs for $U_0^{(m)}$ with PDFs for the rescaled variables $k_m^{1/3}u_m$ considered in the Kolmogorov theory~\cite{frisch1999turbulence,biferale2003shell}; see also (\ref{eq3_T2a}) and the related discussion in Section~\ref{subsec_conj}. While Fig.~\ref{fig4}(a) confirms self-similarity for the normalized variable $U_0^{(m)}$ up to numerical fluctuations, Fig.~\ref{fig4}(b) demonstrates the symmetry breaking (a persistent drift of PDFs with a change of $m$) {\color{black}for the original variables}. 

\begin{figure}[t]
\centering
\includegraphics[width=0.4\textwidth]{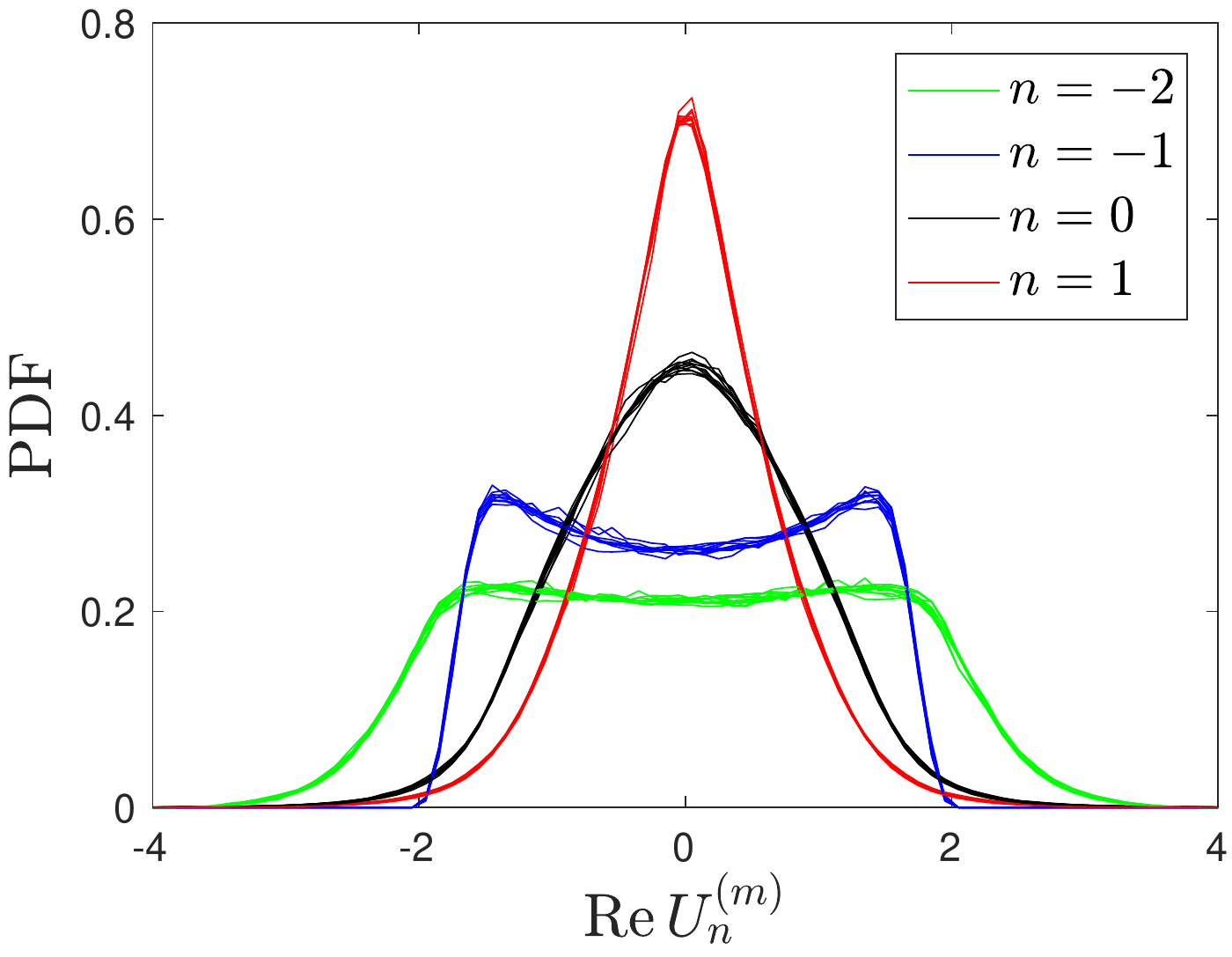}
\caption{PDFs of real parts of normalized and scaled shell velocities $U_{-2}^{(m)}$, $U_{-1}^{(m)}$, $U_0^{(m)}$ and $U_1^{(m)}$ (green, blue, black and red) computed numerically. For each velocity, ten PDFs are shown for $m = 12,\ldots,21$ in the inertial range. The collapse of PDFs onto a single profile verifies the hidden scaling symmetry. }
\label{fig3}
\end{figure}
\begin{figure}
\centering
\includegraphics[width=0.4\textwidth]{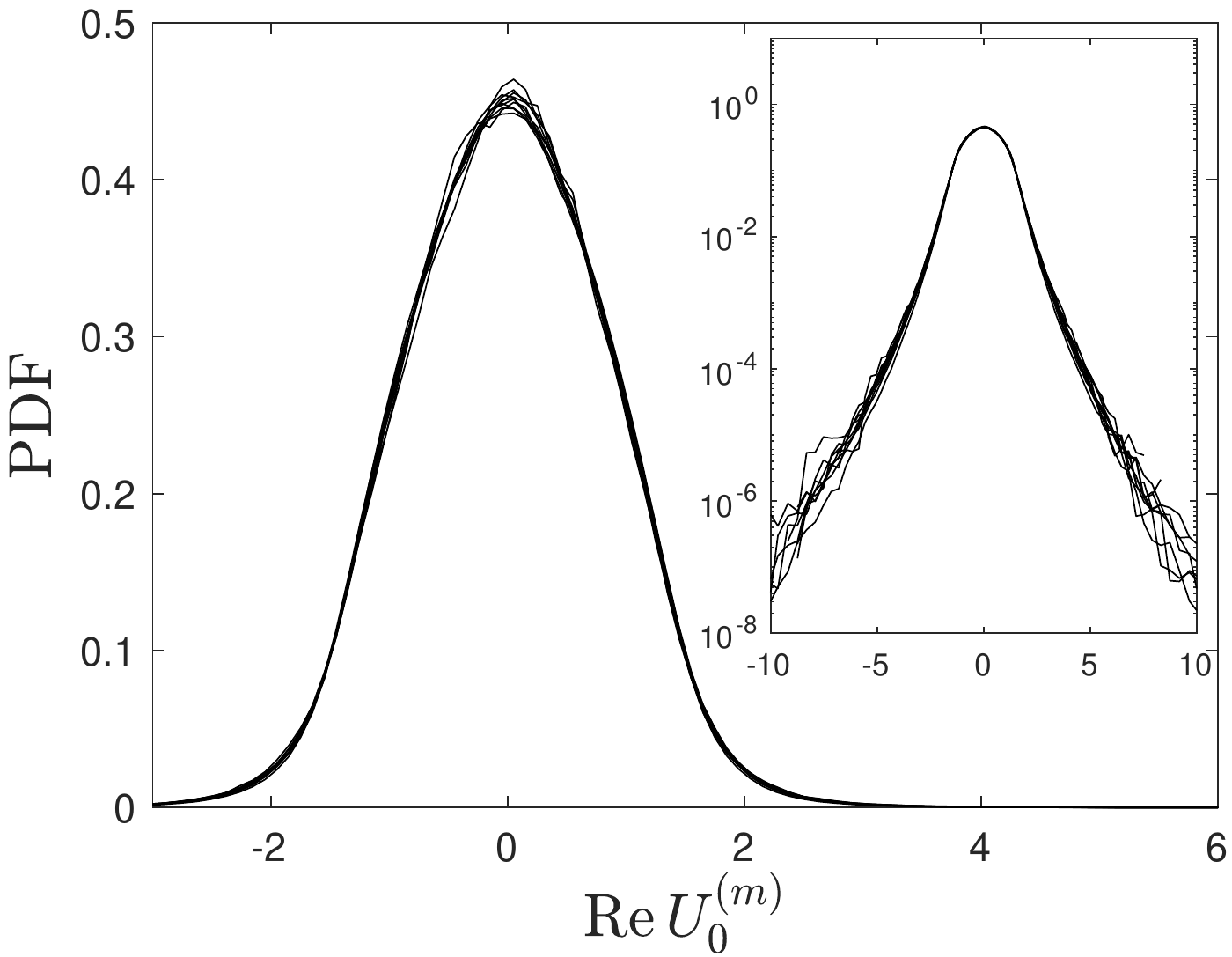}
\put(-95,165){(a)}
\hspace{10mm}
\includegraphics[width=0.4\textwidth]{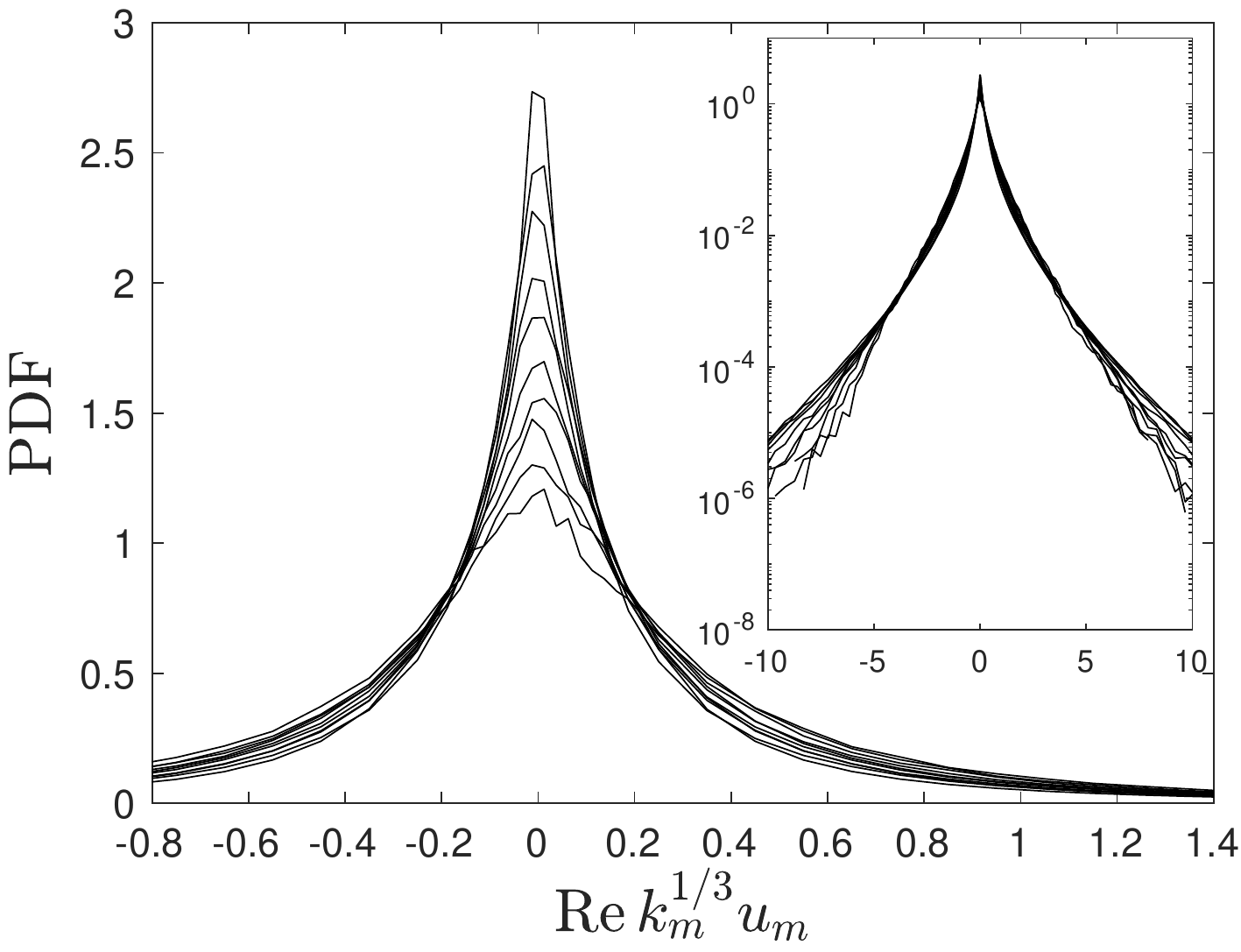}
\put(-100,165){(b)}
\caption{PDFs for real parts of (a) self-similar normalized velocity $U_0^{(m)}$ and (b) the shell velocities rescaled according to the Kolmogorov theory as $k_m^{1/3}u_m$ and demonstrating a symmetry breaking. Both figures show numerical results for $m = 12,\ldots,21$. The insets present the same graphs with a vertical logarithmic scale. }
\label{fig4}
\end{figure}

The hidden scaling symmetry reveals an interesting connection with so-called Kolmogorov multipliers~\cite{kolmogorov1962refinement,chen2003kolmogorov}. For the shell model, these multipliers are defined as ratios $u_{m+1}/u_m$.
Using numerical simulations, statistics of the Kolmogorov multipliers was shown to be universal~\cite{benzi1993intermittency,eyink2003gibbsian}, i.e., independent of the shell number $m$ in the inertial range. The multipliers can be expressed in terms of normalized shell velocities {\color{black}given by (\ref{eq3_7E}) and (\ref{eq3_A2})} as
  	\begin{equation}
	\label{eq3_13}
	 \frac{u_{m+1}}{u_m} = \frac{U_1^{(m)}}{U_0^{(m)}}.
	\end{equation}
According to our numerical observations, the right-hand side in (\ref{eq3_13}) has a universal (independent of $m$) statistics with respect to time $\tau^{(m)}$ due to the hidden scaling symmetry. 

{\color{black}In fact, the right-hand side in (\ref{eq3_13}) can also be replaced by $U_{j+1}^{(m-j)}/U_j^{(m-j)}$ for any integer $j$. This generalization provides universal statistics with respect any fixed time $\tau^{(m-j)}$ and justifies the earlier results~\cite{benzi1993intermittency,eyink2003gibbsian} on universal statistics of multipliers with respect to the original time $t$ as follows. Increasing $j$ yields a large separation between the scale of multiplier $u_{m+1}/u_m$ and the much larger scale of time $\tau^{(m-j)}$. It is natural to expect that the resulting statistics become independent of $j$ for large-scale times $\tau^{(m-j)}$, therefore, becoming the same as for the original time $t$. We refer to \cite{mailybaev2022hidden}, where such derivation is carried out in more detail for the Navier--Stokes system.}

\section{Intermittency}
\label{sec_int}

In fluid dynamics, intermittency refers to irregular alternation between {\color{black}concentrated turbulent and extended laminar-like motions} at high Reynolds numbers, which is traditionally quantified as anomalous scaling of structure functions~\cite{frisch1999turbulence}. For the shell model from Section~\ref{sec_shell1}, the structure function of degree $p > 0$ is defined as
  	\begin{equation}
	\label{eq3_14}
	S_p(k_n) = \int |u_n|^p d\mu,
	\end{equation}
which depends on the wavenumber $k_n = 2^n$ and a probability measure $\mu$ of a statistically stationary state. In the inertial interval (\ref{eq3_2ex}), structure functions feature the asymptotic power law scaling
  	\begin{equation}
	\label{eq3_15}
	S_p(k_n) \propto k_n^{-\zeta_p}.
	\end{equation}
The exponents can be defined by the double limit 
  	\begin{equation}
	\label{eq3_15ex}
	\zeta_p = -\lim_{k_n \to \infty} \lim_{\mathrm{Re} \to \infty} \frac{\log S_p(k_n)}{\log k_n}.
	\end{equation}
As described in Section~\ref{subsec_conj}, this limit corresponds to large wavenumbers in the asymptotically infinite inertial interval (\ref{eq3_2ex}). The \textit{anomaly} is understood as a nonlinear dependence of $\zeta_p$ on $p$, deviating from the prediction $\zeta_p = p/3$ of the Kolmogorov theory and implying the  broken scale invariance~\cite{frisch1999turbulence,falkovich2009symmetries}. 

Analysis of this section is based on the theory of Section~\ref{sec4}. Namely, we consider an invariant probability measure $\mu$ on $\mathcal{X}$ with the symmetry groups $\mathcal{H}_{\mathrm{ts}}$ and $\mathcal{G}$, which define the normalized system on the representative set $\mathcal{Y}$. Introducing generalized structure functions, we relate intermittency to the hidden scaling symmetry. The main result is a formula for anomalous exponents, which is obtained as a consequence of the scaling symmetry of the normalized invariant measure. {\color{black}Being derived within a general group--theoretical formulation, this result is applicable to different turbulence models; see \cite{mailybaev2021solvable,mailybaev2022shell} for different types of shell models and Section~\ref{subsec_NST} discussing the application to the Navier--Stokes system.

\subsection{Generalized structure functions}

Let $F: \mathcal{X} \mapsto \mathbb{R}$ be a measurable function with the homogeneity property
  	\begin{equation}
	\label{eq4_2}
	F \circ h^a(x) = \frac{F(x)}{a^p}
	\end{equation}
for any $h^a \in \mathcal{H}_{\mathrm{ts}}$ and $x \in \mathcal{X}$. Let us also fix an arbitrary symmetry $g \in \mathcal{G}$. 
We introduce the corresponding generalized structure function of order $p$ as
  	\begin{equation}
	\label{eq3_17}
	S_p(k_n) 
	= \frac{1}{k_n^p} \int F\circ g^n d\mu,
	\end{equation}
where $k_n = 2^n$.} 

{\color{black}In applications to turbulence, we interpret $g$ as a space scaling map, which doubles spatial resolution.} 
For the shell model example, structure functions (\ref{eq3_14}) are recovered by taking $F(x) = |u_0|^p$ with the symmetries (\ref{eq3_4a}) and (\ref{eq3_4b}). {\color{black}A similar representation for the Navier--Stokes system is discussed in Section~\ref{subs_6_4}.

First, let us describe the non-intermittent case, when the original measure $\mu$ is symmetric.}

\begin{proposition}
\label{prop_sym_mu}
Consider a measure $\mu$ satisfying the symmetry condition 
  	\begin{equation}
	\label{eq4_L1b}
	(g \circ h^a)_\sharp \mu = \mu
	\end{equation}
for some $g \in \mathcal{G}$ and $h^a \in \mathcal{H}_{\mathrm{ts}}$.
If $\int F d\mu$ is finite and nonzero, then structure functions (\ref{eq3_17}) have the power law scaling (\ref{eq3_15}) with the exponents
  	\begin{equation}
	\label{eq4_L1c}
	\zeta_p = (1- \log_2 a)p.
	\end{equation}
\end{proposition}

\begin{proof}
Using properties (\ref{eq2_0ha}) and (\ref{eq2_2com}), we express
  	\begin{equation}
	\label{eq4_L1_g}
	g^n = h^{1/a^n} \circ (g\circ h^a)^n.
	\end{equation}
Substituting this formula into (\ref{eq3_17}) yields
  	\begin{equation}
	\label{eq4_L1}
	S_p(k_n)  = 
	\frac{1}{k_n^p}\int F \circ h^{1/a^n} \circ (g\circ h^a)^n \, d\mu.
	\end{equation}
Using the change-of-variables formula for a push-forward measure with symmetry assumption (\ref{eq4_L1b}) and relation (\ref{eq4_2}) yields
  	\begin{equation}
	\label{eq4_L1aa}
	S_p(k_n)  = 
	\frac{1}{k_n^p}\int F \circ h^{1/a^n}\, d\mu
	= \frac{a^{np}}{k_n^p}\int F d\mu 
	= k_n^{-\zeta_p} \int F d\mu
	\end{equation}
with exponent (\ref{eq4_L1c}), where we took into account that $k_n = 2^n$. 
\end{proof}

For example, by taking $a = 2^{2/3}$, Proposition~\ref{prop_sym_mu} recovers the exponents $\zeta_p = p/3$ of the Kolmogorov  theory~\cite{frisch1999turbulence}. As we already mentioned, it is known that the exponents $\zeta_p$ depend nonlinearly on $p$ in the intermittent turbulence~\cite{frisch1985singularity,frisch1999turbulence} and, hence, all symmetries (\ref{eq4_L1b}) must be broken. 

{\color{black}
\subsection{Structure functions in terms of multipliers}
\label{subsec_iter}

We turn now to the intermittent case assuming that, despite all scaling symmetries are broken for the measure $\mu$, the hidden scaling symmetry is recovered for the normalized measure $\nu$.
In addition to (\ref{eq3_17}), we introduce the normalized structure functions 
  	\begin{equation}
	\label{eq4_V4}
	N_p(k_n) 
	= \frac{1}{k_n^p} \int F\circ g^n d\nu,
	\end{equation}
in which the integration is performed with respect to the normalized measure $\nu$ on the representative set $\mathcal{Y}$. In the next Subsection~\ref{sec_intS} we will see that structure functions (\ref{eq3_17}) and (\ref{eq4_V4}) have the same scaling asymptotics. In the present technical subsection, we derive iterative relations for the integral (\ref{eq4_V4}), which are used later for determining the scaling exponents $\zeta_p$. 

Our description uses the idea of Kolmogorov multipliers~\cite{kolmogorov1962refinement,benzi1993intermittency,chen2003kolmogorov}, which are ratios of velocity increments at different scales. Given a state $x$, we introduce the generalized multiplier $\sigma_n(x)$ as a similar ratio
  	\begin{equation}
	\label{eq4_T5}
	\sigma_n(x) = \frac{A\circ g^{n+1}(x)}{A\circ g^{n}(x)},
	\end{equation}
where the function $A(x)$ is computed at two different scales defined by the scaling maps $g^{n+1}$ and $g^{n}$. The important property of multipliers is that they are invariant with respect to time scalings: for $x' = h^a(x)$ with any $a > 0$ one has
  	\begin{equation}
	\label{eq4_T5eq}
	\sigma_n(x') = \frac{A\circ g^{n+1}\circ h^a (x)}{A\circ g^{n}\circ h^a(x)}
	= \frac{A \circ h^a \circ g^{n+1} (x)}{A \circ h^a \circ g^{n}(x)} = \frac{A\circ g^{n+1}(x)/a}{A\circ g^{n}(x)/a} = \sigma(x),\quad 
	x' = h^a(x),
	\end{equation}
where we used commutativity of $g$ with $h^a$ and (\ref{eq2_2Ap}). In particular, $\sigma(x) = \sigma(y)$ for $y = P(x)$.

We will use the sequences 
  	\begin{equation}
	\label{eqIS_3}
	\boldsymbol{\sigma}_- = (\sigma_{-1},\sigma_{-2},\ldots) \in \mathcal{S}_-
	\end{equation}
considered as infinite-dimensional vectors in the space $\mathcal{S}_- = \mathbb{R}_+^{\infty}$ with the standard product topology. Adding extra components $\sigma_0$ and $f$, we similarly introduce the sequences
  	\begin{equation}
	\label{eqIS_3sqB}
	\boldsymbol{\sigma}_\ominus = (\sigma_0,\boldsymbol{\sigma}_-) = (\sigma_0,\sigma_{-1},\sigma_{-2},\ldots) \in \mathbb{R}_+ \times \mathcal{S}_-.
	\end{equation}
  	\begin{equation}
	\label{eqIS_3sqC}
	\boldsymbol{\phi} = (f,\boldsymbol{\sigma}_\ominus) = (f,\sigma_0,\sigma_{-1},\sigma_{-2},\ldots) \in \mathbb{R} \times\mathbb{R}_+ \times \mathcal{S}_-.
	\end{equation}
Functions $\sigma_n(y)$ from (\ref{eq4_T5}) and $f = F(y)$ define the mappings from $\mathcal{Y}$ to the spaces (\ref{eqIS_3})--(\ref{eqIS_3sqC}), which we denote as
  	\begin{equation}
	\label{eqIS_3mA}
	\boldsymbol{\sigma}_- = \mathbf{P}_-(y),\quad
	\boldsymbol{\sigma}_\ominus = \mathbf{P}_\ominus(y),\quad
	\boldsymbol{\phi} = \mathbf{P}_\phi (y).
	\end{equation}
We will also need the shift map $\mathbf{S}: \mathbb{R}_+ \times \mathcal{S}_- \mapsto \mathcal{S}_-$ defined by the relations
  	\begin{equation}
	\boldsymbol{\sigma}'_- = \mathbf{S}(\boldsymbol{\sigma}_\ominus),\quad \sigma'_{n} = \sigma_{n+1}, \quad n < 0.
	\label{eqIS_3shift}
	\end{equation}

Let us denote the scaled normalized measures as
  	\begin{equation}
	\label{eqIS_4}
	\nu^{(n)} = g^n_\star \nu,
	\end{equation}
where the hidden scaling operator $g_\star$ is given by expression (\ref{eq2_A7}) of Theorem~\ref{prop_g}.
We denote the images (pushforwards) of $\nu^{(n)}$ in the spaces (\ref{eqIS_3})--(\ref{eqIS_3sqC}) as
  	\begin{equation}
	\label{eqIS_4m}
	\nu^{(n)}_-(\boldsymbol{\sigma}_-) = (\mathbf{P}_-)_\sharp \nu^{(n)},\quad
	\nu^{(n)}_\ominus(\boldsymbol{\sigma}_\ominus) = (\mathbf{P}_\ominus)_\sharp \nu^{(n)},\quad
	\nu^{(n)}_\phi(\boldsymbol{\sigma}_\phi) = (\mathbf{P}_\phi)_\sharp \nu^{(n)},
	\end{equation}
where we specified the corresponding space variables in the parentheses. 
Finally, we assume that there exist conditional probability densities given by measurable functions $\rho^{(n)}_0(\sigma_0|\boldsymbol{\sigma}_-)$ and $\rho^{(n)}_F(f|\boldsymbol{\sigma}_\ominus)$ satisfying the standard defining relations
  	\begin{equation}
	\label{eqIS_5}
	d\nu^{(n)}_\ominus(\boldsymbol{\sigma}_\ominus) = \rho^{(n)}_0(\sigma_0|\boldsymbol{\sigma}_-)\, d\sigma_0\,d\nu^{(n)}_-(\boldsymbol{\sigma}_-),\quad
	d\nu^{(n)}_\phi(\boldsymbol{\sigma}_\phi) = \rho^{(n)}_F(f|\boldsymbol{\sigma}_\ominus)\, df\,d\nu^{(n)}_\ominus(\boldsymbol{\sigma}_\ominus).
	\end{equation}
Here the assumption that the densities are measurable is taken for convenience; one can use conditional measures provided by the disintegration theorem in a more general situation; see e.g.~\cite{chang1997conditioning}.  

The following theorem formulates structure functions in terms of multipliers; for the proof see Subsection~\ref{subs_pr4}. 

\begin{theorem}
\label{theorem4}
Generalized structure functions (\ref{eq4_V4}) for $n \ge 0$ can be expressed in the form
  	\begin{equation}
	\label{eqIS_6}
	N_p(k_n) 
	= \int f\, \rho^{(n)}_F(f|\boldsymbol{\sigma}_\ominus) \,df \,d\lambda_p^{(n)}(\boldsymbol{\sigma}_\ominus),
	\end{equation}
where 
  	\begin{equation}
	\label{eqIS_6b}
	d\lambda_p^{(n)}(\boldsymbol{\sigma}_\ominus)
	= \frac{c_n}{k_n^p} 
	\left(\prod_{j = 1}^{n} \sigma_{-j}^{p-1}\right)d\nu_\ominus^{(n)}(\boldsymbol{\sigma}_\ominus),\quad
	c_n = \prod_{j = 1}^{n} \int A \circ g\, d\nu^{(j-1)}.
	\end{equation}
The measures $\lambda_p^{(n)}(\boldsymbol{\sigma}_\ominus)$ satisfy the iterative relations with $\lambda_p^{(0)} = \nu^{(0)}_\ominus$ and
  	\begin{equation}
	\label{eqIS_7}
	d\lambda_p^{(n+1)}(\boldsymbol{\sigma}_\ominus)
	= 
	2^{-p}\sigma_{-1}^p\, 
	\rho_0^{(n+1)}(\sigma_0|\boldsymbol{\sigma}_-)
	 \,d\sigma_0 \,d\Lambda_p^{(n)}(\boldsymbol{\sigma}_-),
	\quad 
	\Lambda_p^{(n)}(\boldsymbol{\sigma}_-) = \mathbf{S}_\sharp \lambda_p^{(n)}(\boldsymbol{\sigma}_\ominus),
	\end{equation}
where the measure $\Lambda_p^{(n)}(\boldsymbol{\sigma}_-)$ is obtained using the shift operator (\ref{eqIS_3shift}).
\end{theorem}

We remark that expression (\ref{eqIS_7}) can be written as
  	\begin{equation}
	\label{eqIS_9}
	\lambda_p^{(n+1)} = \mathcal{L}_p^{(n+1)} [\lambda_p^{(n)}],
	\end{equation}
where $\mathcal{L}_p^{(n+1)}$ is a linear operator acting on measures $\lambda_p^{(n)}(\boldsymbol{\sigma}_\ominus)$. Notice that this operator does not preserve the probability property of measures, i.e., $\int d\lambda_p^{(n)} \ne 1$ in general.
}

\subsection{Anomalous exponents as Perron--Frobenius eigenvalues}
\label{sec_intS}

In this subsection, we show how the scaling power laws for structure functions appear as a consequence of the hidden scaling symmetry. 
We establish this connection in two steps. First, we relate the original (generalized) structure functions $S_p$ from (\ref{eq3_17}) with the normalized structure functions from (\ref{eq4_V4}). Then, using the symmetry condition for the normalized measure, $g_\star \nu = \nu$, and the iterative relation of Theorem~\ref{theorem4}, we derive the asymptotic power law scaling (\ref{eq3_15}) and determine the respective exponents $\zeta_p$. 

Let us use the shell model from Section~\ref{sec_shell1} as an example. System (\ref{eq3_1}) describes the evolution of shell variables $u_n(t)$ for $n \ge 0$, where $n = 0$ corresponds to the largest scale (lowest wavenumber $k_0 = 1$) of the system. {\color{black}For a proper definition of the scaling group, we must introduce the ``dummy'' shell variables} with $n < 0$; see the crossed cells in Fig.~\ref{figHS}. This is a purely formal procedure, because these variables are removed in the limit $m \to \infty$; see Fig.~\ref{figHS} and Definition~\ref{conj1} of the hidden scaling symmetry. Therefore, we are free to set $u_{-1} = 2$ and $u_n = 0$ for $n < -1$, which yields the sum $\sum_{n < 0}k_n^2|u_n|^2 = 1$. With such a choice the function $A(x)$ from (\ref{eq3_B1}) takes the constant value
  	\begin{equation}
	\label{eq4_V4ex}
	A(x) = 1
	\end{equation}
for all states $x$ of interest. According to Definition~\ref{def1} and Theorem~\ref{theorem0}, we conclude that the normalized measure $\nu$ coincides with $\mu$ and, hence, the structure functions (\ref{eq3_17}) coincide with their normalized counterparts (\ref{eq4_V4}):  
  	\begin{equation}
	\label{eq4_V4exS}
	S_p(k_n) = N_p(k_n).
	\end{equation}

One can imagine that the ``trick'' leading to (\ref{eq4_V4ex}) and (\ref{eq4_V4exS}) is applicable in other fluid models, which possess the largest (so-called integral) scale. It is related to a proper artificial extension of the state variable $x$ to larger scales. Obviously, this extension does not affect the asymptotic scaling properties referring to small-scale dynamics.

Now we turn to the hidden scaling symmetry introduced in Definition~\ref{conj1} as the double limit (\ref{eq3_T2}) in the inertial interval. This symmetry implies that the measure $\nu^{(n)}$ from (\ref{eqIS_4}) converges to the self-similar measure $\nu^\infty$, {\color{black}and the same refers to the corresponding projections (\ref{eqIS_4m}). In particular, we can write analogous limits for conditional probability densities from (\ref{eqIS_5}) as
	\begin{equation}
	\label{eq4_K3}
	\lim_{n \to \infty} \lim_{\mathrm{Re} \to \infty} \rho_0^{(n)} = \rho_0^\infty, \quad
	\lim_{n \to \infty} \lim_{\mathrm{Re} \to \infty} \rho_F^{(n)} = \rho_F^\infty.
	\end{equation}
	
We define the limiting operator $\mathcal{L}_p^\infty$ acting on measures $\lambda_p(\boldsymbol{\sigma}_\ominus)$ and corresponding to (\ref{eqIS_7})--(\ref{eqIS_9}) as
  	\begin{equation}
	\label{eq4_K3ex}
	\lambda'_p = \mathcal{L}_p^\infty[\lambda_p], \quad  
	d\lambda'_p(\boldsymbol{\sigma}_\ominus) = 
	2^{-p}\sigma_{-1}^p\, 
	\rho_0^{\infty}(\sigma_0|\boldsymbol{\sigma}_-)
	 \,d\sigma_0 \,d\Lambda_p(\boldsymbol{\sigma}_-),\quad
	 \Lambda_p(\boldsymbol{\sigma}_-) = \mathbf{S}_\sharp \lambda_p(\boldsymbol{\sigma}_\ominus).
	\end{equation}
where we replaced $\rho_0^{(n+1)}$ by its asymptotic form $\rho_0^\infty$.}
The operator $\mathcal{L}_p^\infty$ is linear and positive: it maps positive measures to positive measures. Hence, we can use the Krein--Rutman theorem under proper assumptions of compactness; see~\cite[\S 19.5]{deimling2010nonlinear} for a precise formulation. This theorem, generalizing the Perron--Frobenius theorem for matrices with positive entries {\color{black}\cite[Ch.~16]{lax2007linear},} proves the existence of the (maximum) Perron--Frobenius eigenvalue $R_p > 0$ with a positive eigenvector (probability measure) $\lambda_p^\infty$ satisfying the equation
  	\begin{equation}
	\label{eq4_KA21}
	\mathcal{L}_p^\infty[\lambda_p^\infty] = R_p \lambda_p^\infty.
	\end{equation}
{\color{black}The eigenvalue $R_p$ is simple and dominant: absolute values of all other eigenvalues of $\mathcal{L}_p^\infty$ are smaller than $R_p$ under the  assumption of strong positivity~\cite[\S 19.5]{deimling2010nonlinear}.

Let us write the iterative relations of Theorem~\ref{theorem4} in the operator form as
  	\begin{equation}
	\label{eq4_KA_it}
	\lambda_p^{(n)} = \mathcal{L}_p^{(n)} \circ \mathcal{L}_p^{(n-1)}\circ \cdots \mathcal{L}_p^{(1)} [\nu_\ominus^{(0)}], \quad
	\nu_\ominus^{(0)} = (\mathbf{P}_\ominus)_\sharp \nu.
	\end{equation}
The convergence properties (\ref{eq4_K3}) imply that the limiting operator $\mathcal{L}_p^{(n)} \to \mathcal{L}_p^\infty$ asymptotically for large $n$ in the inertial interval. Hence, the iterative procedure   (\ref{eq4_KA_it}) with a generic initial measure $\nu_\ominus^{(0)}$ converges for large $n$ to the dominant Perron--Frobenius mode. In particular,} the measures $\lambda_p^{(n)}$ converge, up to a positive scalar factor, to the Perron--Frobenius eigenvector $\lambda_p^\infty$. In this limit, each iteration reduces to multiplication by the Perron--Frobenius eigenvalue $R_p$.
 Precisely, these properties are formulated as
  	\begin{equation}
	\label{eq4_KA2}
	\lim_{n \to \infty} \lim_{\mathrm{Re} \to \infty}  
	\frac{\lambda_p^{(n)}}{\int d\lambda_p^{(n)}} = \lambda_p^\infty,\quad
	\lim_{n \to \infty} \lim_{\mathrm{Re} \to \infty}  
	\frac{\int d\lambda_p^{(n+1)}}{\int d\lambda_p^{(n)}} = R_p.
	\end{equation}
Using limits (\ref{eq4_K3}) and (\ref{eq4_KA2}) with expressions (\ref{eqIS_6}) and (\ref{eq4_V4exS}), yields the structure function asymptotically proportional to {\color{black}
  	\begin{equation}
	\label{eq4_K8PL}
	S_p(k_n) \propto R_p^n 
	\int f\, \rho^{\infty}_F(f|\boldsymbol{\sigma}_\ominus) \,df \,d\lambda_p^{\infty}(\boldsymbol{\sigma}_\ominus).
	\end{equation}
}Notice that the limits in (\ref{eq4_KA2}) are considered here as assumptions, which are naturally related to the hidden scaling symmetry. A precise formulation that guaranties the convergence would require technical details depending on a specific system under consideration.
Recalling that $k_n = 2^n$, we obtain the following formula for scaling exponents in (\ref{eq3_15}).  

\begin{corollary}\label{cor2AE}
Assuming limits (\ref{eq4_K3}) and (\ref{eq4_KA2}) and a finite nonzero value of the integral 
  	\begin{equation}
	\label{eq4_K7}
	\int f\, \rho^{\infty}_F(f|\boldsymbol{\sigma}_\ominus) \,df \,d\lambda_p^{\infty}(\boldsymbol{\sigma}_\ominus),
	\end{equation}
the structure function $S_p(k_n)$ has the asymptotic power law scaling (\ref{eq3_15}) in the inertial interval with the exponent
  	\begin{equation}
	\label{eq4_K8}
	\zeta_p = -\log_2 R_p,
	\end{equation}
where $R_p$ is the Perron--Frobenius eigenvalue; see (\ref{eq4_KA21}).
\end{corollary}

The important property of Corollary~\ref{cor2AE} is that exponents (\ref{eq4_K8}) can be anomalous, i.e., depending nonlinearly on $p$. {\color{black}For example, consider the probability density  $\rho_0^\infty(\sigma_0|\boldsymbol{\sigma}_-) = \rho(\sigma_0)$, which is independent of $\boldsymbol{\sigma}_-$. Recall that the Perron--Frobenius eigenvector is the only positive eigenvector (measure) solution of (\ref{eq4_KA21})~\cite{lax2007linear,deimling2010nonlinear}. In this case the eigenvalue problem (\ref{eq4_KA21}) with the operator (\ref{eq4_K3ex}) can solved by using the ansatz 
  	\begin{equation}
	\label{eq4_Kint1}
	d\lambda_p^\infty(\boldsymbol{\sigma}_\ominus) = \sigma_0^{-p}\,d\widetilde{\lambda}_p(\boldsymbol{\sigma}_\ominus),
	\end{equation}
which yields
  	\begin{equation}
	\label{eq4_Kint2}
	2^{-p}\rho(\sigma_0)\, d\sigma_0\, d\widetilde{\Lambda}_p(\boldsymbol{\sigma}_-) = R_p\sigma_0^{-p}\,d\widetilde{\lambda}_p(\boldsymbol{\sigma}_\ominus),
	\quad
	\widetilde{\Lambda}_p(\boldsymbol{\sigma}_-) = \mathbf{S}_\sharp \widetilde{\lambda}_p(\boldsymbol{\sigma}_\ominus).
	\end{equation}
After dividing by $\sigma_0^{-p}$, both sides can be integrated taking into account that $\int d\widetilde{\Lambda}_p(\boldsymbol{\sigma}_-) = \int d\widetilde{\lambda}_p(\boldsymbol{\sigma}_\ominus)$ due to the pushforward relation. This yields
  	\begin{equation}
	\label{eq4_Kint}
	R_p = 2^{-p} \int \sigma_0^p \rho(\sigma_0)\, d\sigma_0.
	\end{equation}
Expression (\ref{eq4_Kint}) defines the Perron--Frobenius eigenvalues $R_p$ through moments of the probability density $\rho(\sigma_0)$. As a consequence, the corresponding exponents $\zeta_p = -\log_2 R_p$ are anomalous (depend nonlinearly on $p$) in general, e.g., consider $\rho(\sigma_0)$ to be a normal distribution. Furthermore, one can show the well-known concave property of $\zeta_p$ as a function of $p$~\cite{frisch1999turbulence}: applying the Cauchy--Schwarz inequality to (\ref{eq4_Kint}) yields $R^2_{p+q} \le R_{2p} R_{2q}$ and, hence, $\zeta_{p+q} \ge (\zeta_{2p} + \zeta_{2q})/2$.}

{\color{black}We refer to the subsequent work~\cite{mailybaev2021solvable}, where relation (\ref{eq4_Kint}) is implemented rigorously for the anomalous statistics in a specially designed shell model. Also, we refer to~\cite{mailybaev2022shell} for the numerical verification of Corollary~\ref{cor2AE} in the Sabra shell model of turbulence, where the operator $\mathcal{L}_p^\infty$ is approximated using multi-dimensional histograms. This latter work uses the concept of turn-over times $T_m$, which are defined through the function $1/A(x)$ from (\ref{eq3_B1}) in Section~\ref{subsec_shell_norm}, and the derivation of anomalous power-laws is given in a different way using temporal averages. The two derivations are equivalent under the ergodicity assumption.

Another consequence of Corollary~\ref{cor2AE} is that the exponents $\zeta_p$ depend only on the time-homogeneity property (\ref{eq4_2}), i.e., they do not depend on a specific form of function $F$. Using this property, one can apply our results to integrated multi-time correlation functions studied in~\cite{l1997temporal,biferale1999multi}. Notice, however, that this scaling may change if the integral (\ref{eq4_K7}) vanishes of diverges.

In summary, we see that normalized measures with a hidden scaling symmetry define scaling exponents $\zeta_p$ in terms of Perron--Frobenius eigenvalues of the linear operators $\mathcal{L}_p^\infty$. These exponents may depend nonlinearly on $p$, i.e., be anomalous.

\subsection{Proof of Theorem~\ref{theorem4}}
\label{subs_pr4}

\begin{lemma}
\label{lem1}
The formula
  	\begin{equation}
	\label{eqIP_6}
	N_p(k_n) 
	= \frac{c_m}{k_n^p} \int 
	\left(\prod_{j = 1}^{m} \sigma_{-j}^{p-1}\right) F \circ g^{n-m}\, d \nu^{(m)}
	\end{equation}
is valid for any $m \ge 0$. 
\end{lemma}

\begin{proof}[Proof of Lemma~\ref{lem1}]
Expression (\ref{eqIP_6}) reduces to the definition (\ref{eq4_V4}) for $m = 0$ and $c_0 = 1$. Hence, we can prove it by induction assuming that (\ref{eqIP_6}) is valid for a given $m \ge 0$ and verifying the next value $m+1$. 

Let us first prove the identity
  	\begin{equation}
	\label{eqIP_5b}
	F \circ g^n (y)
	= C^p(y) \,F \circ g^{n-1} \circ P \circ g(y),\quad
	C(y) = A \circ g(y)
	\end{equation}
for any $y \in \mathcal{Y}$.
By definition (\ref{eq2_2P}), we have $P \circ g(y) = h^a \circ g(y)$ with $a = A\circ g(y) = C(y)$. Using the inverse map $(h^a)^{-1} = h^{1/a}$, we have 
  	\begin{equation}
	\label{eqIP_5cg}
	g(y) = h^{1/C(y)}\circ P\circ g(y). 
	\end{equation}
Hence,
  	\begin{equation}
	\label{eqIP_5c}
	F \circ g^n(y) = F \circ g^{n-1}\circ g(y)
	= F \circ g^{n-1} \circ h^{1/C(y)} \circ P \circ g(y).
	\end{equation}
Using commutativity of $g^{n-1}$ with $h^{1/C(y)}$ and relation (\ref{eq4_2}) yields (\ref{eqIP_5b}).

Second, let us prove the relations 
  	\begin{equation}
	\sigma_{n}(y) = \sigma_{n-1}\circ P \circ g(y),\quad 
	\sigma_0(y) = C(y)
	\label{eqIP_6x}
	\end{equation}
for any $n \in \mathbb{Z}$ and $y \in \mathcal{Y}$. Relation (\ref{eqIP_5b}) yields 
  	\begin{equation}
	\label{eqIS_1}
	A \circ g^{n}(y) 
	= C(y)\, A \circ g^{n-1} \circ P\circ g(y),
	\end{equation}
because the function $A(y)$ satisfies condition (\ref{eq4_2}) with $p = 1$; see (\ref{eq2_2Ap}). Using (\ref{eqIS_1})
in both numerator and denominator of definition (\ref{eq4_T5}), yields the first relation of (\ref{eqIP_6x}). The second relation follows from (\ref{eq4_T5}) for $n = 0$ because $A(y) = 1$; see (\ref{eq2_2Ap}).

Equality (\ref{eqIP_6}) is expressed using (\ref{eqIP_5b}) as
  	\begin{equation}
	\label{eqIP_7a}
	N_p(k_n) 
	= \frac{c_m}{k_n^p} 
	\int \left(\prod_{j = 1}^{m} \sigma_{-j}^{p-1} \right) 
	C^p \,F \circ g^{n-m-1} \circ P \circ g \, d \nu^{(m)}.
	\end{equation}
With the second relation of (\ref{eqIP_6x}), we replace $C^{p-1}$ in (\ref{eqIP_7a}) by $\sigma_0^{p-1}$, extending the product to $j = 0$ as
  	\begin{equation}
	\label{eqIP_7aa}
	N_p(k_n) 
	= \frac{c_m}{k_n^p} 
	\int \left(\prod_{j = 0}^{m} \sigma_{-j}^{p-1} \right) 
	C \,F \circ g^{n-m-1} \circ P \circ g \, d \nu^{(m)}.
	\end{equation}
The change of time transformation (\ref{eq2_A7b}) yields
  	\begin{equation}
	\label{eqIP_8prev}
	N_p(k_n) 
	= \frac{c_{m+1}}{k_n^p}
	\int \left(\prod_{j = 0}^{m} \sigma_{-j}^{p-1}\right) 
	F \circ g^{n-m-1} \circ P \circ g \, d \nu_C^{(m)},
	\quad
	d\nu_C^{(m)} = \frac{C\, d\nu^{(m)}}{\int C\, d\nu^{(m)}},
	\end{equation}
where we used expression (\ref{eqIS_6b}) for the coefficient $c_{m+1}$.
Using the first relation of (\ref{eqIP_6x}) in (\ref{eqIP_8prev}) and changing the product index $j \mapsto j+1$, we write
  	\begin{equation}
	\label{eqIP_8}
	N_p(k_n) 
	= \frac{c_{m+1}}{k_n^p}
	\int \left(\prod_{j = 1}^{m+1} \sigma_{-j} \circ P \circ g\right)^{p-1} 
	F \circ g^{n-m-1} \circ P \circ g \, d \nu_C^{(m)}.
	\end{equation}
Finally, the change of variables $y \mapsto P \circ g (y)$ reduces (\ref{eqIP_8}) to the form
  	\begin{equation}
	\label{eqIP_10}
	N_p(k_n) 
	= \frac{c_{m+1}}{k_n^p} 
	\int \left(\prod_{j = 1}^{m+1} \sigma_{-j}^{p-1}\right) 
	F \circ g^{n-m-1} \, d \nu'',
	\end{equation}
where $\nu'' = (P \circ g)_\sharp \nu_C^{(m)}$ is a pushforward measure given by the classical change-of-variables formula. 
This measure is expressed by the hidden symmetry transformation (\ref{eq2_A7}) and (\ref{eqIS_4}) as
  	\begin{equation}
	\label{eqIP_12}
	\nu'' = (P \circ g)_\sharp \nu_C^{(m)} 
	= g_\star \nu^{(m)} = g_\star (g^m_\star \nu) = g^{m+1}_\star \nu = \nu^{(m+1)}.
	\end{equation}
Expressions (\ref{eqIP_10}) and (\ref{eqIP_12}) prove the induction step: the formula (\ref{eqIP_6}) for $m+1$.
\end{proof}

Equality (\ref{eqIP_6}) written for $m = n$ takes the form
  	\begin{equation}
	\label{eqIP_13}
	N_p(k_n) 
	= \frac{c_n}{k_n^p} \int 
	\left(\prod_{j = 1}^{n} \sigma_{-j}^{p-1}\right) F d \nu^{(n)}.
	\end{equation}
Notice that the integral expression depends only on $F$ and $\sigma_{-1},\ldots,\sigma_{-n}$ as functions of $y \in \mathcal{Y}$, integrated with respect to the measure $\nu^{(n)}$ in the representative set $\mathcal{Y}$. By changing the integration variables from $y$ to $(f,\sigma_0,\sigma_{-1},\sigma_{-2},\ldots) = \mathbf{P}_\phi(y)$ with $f = F(y)$, we reduce formula (\ref{eqIP_13}) to the form 
  	\begin{equation}
	\label{eqIP_14}
	N_p(k_n) 
	= \frac{c_n}{k_n^p} \int 
	\left(\prod_{j = 1}^{n} \sigma_{-j}^{p-1}\right) f d \nu_\phi^{(n)}.
	\end{equation}
Here the measure $d \nu^{(n)}$ is substituted by its image $d \nu_\phi^{(n)}$ in the space $\boldsymbol{\phi} = (f,\sigma_0,\sigma_{-1},\sigma_{-2},\ldots)$; see (\ref{eqIS_3sqC}), (\ref{eqIS_3mA}) and (\ref{eqIS_4m}). Finally, using the second expression in (\ref{eqIS_5}), we prove relations (\ref{eqIS_6}) and (\ref{eqIS_6b}) of the theorem.
 
Let us now prove the iterative relations (\ref{eqIS_7}). Using the first relation of (\ref{eqIS_5}) and $k_n = 2^n$ in expression (\ref{eqIS_6b}) for $n+1$, we write
  	\begin{equation}
	\label{eqIS_7c}
	\lambda_p^{(n+1)}(\boldsymbol{\sigma}_\ominus)
	= 2^{-p}\sigma_{-1}^p\, 
	\rho_0^{(n+1)}(\sigma_0|\boldsymbol{\sigma}_-)
	\,d\sigma_0\, d\Lambda'(\boldsymbol{\sigma}_-),
	\end{equation}
where
  	\begin{equation}
	\label{eqIS_20}
	d\Lambda'(\boldsymbol{\sigma}_-) = 
	\frac{c_{n+1}}{k_{n}^p} 
	\left(\frac{1}{\sigma_{-1}}\prod_{j = 2}^{n+1} \sigma_{-j}^{p-1}\right)
	\,d\nu^{(n+1)}_-(\boldsymbol{\sigma}_-).
	\end{equation}
Comparing with (\ref{eqIS_7}), one can see that for the proof of the theorem it remains to show that
  	\begin{equation}
	\label{eqIS_20b}
	\Lambda'(\boldsymbol{\sigma}_-) =  \mathbf{S}_\sharp \lambda_p^{(n)}(\boldsymbol{\sigma}_\ominus).
	\end{equation}

Using (\ref{eqIS_4m}) we express the last measure in (\ref{eqIS_20}) as $\nu^{(n+1)}_- = (\mathbf{P}_-)_\sharp \nu^{(n+1)}$.
Combining this expression with (\ref{eqIP_12}) yields
  	\begin{equation}
	\label{eqIS_21}
	\nu^{(n+1)}_- = (\mathbf{P}_- \circ P \circ g)_\sharp \nu^{(n)}_C.
	\end{equation}
One can see using (\ref{eqIP_6x}) and definitions (\ref{eqIS_3})--(\ref{eqIS_3shift}) that $\mathbf{P}_- \circ P \circ g = \mathbf{S} \circ \mathbf{P}_\ominus$. Hence,
  	\begin{equation}
	\label{eqIS_22}
	\nu^{(n+1)}_- = (\mathbf{S} \circ \mathbf{P}_\ominus)_\sharp \nu^{(n)}_C.
	\end{equation}
For $\nu^{(n)}_C$, we use expressions (\ref{eq2_A7b}) with $C = A \circ g$ and (\ref{eqIP_6x}) as
  	\begin{equation}
	\label{eqIS_23}
	d\nu^{(n)}_C = \frac{\sigma_0 \,d\nu^{(n)}}{\int A\circ g \,d\nu^{(n)}}.
	\end{equation} 
Substituting (\ref{eqIS_23}) into (\ref{eqIS_22}) and using definitions (\ref{eqIS_3shift}) and (\ref{eqIS_4m}), we have
   	\begin{equation}
	\label{eqIS_24}
	\nu^{(n+1)}_-(\boldsymbol{\sigma}_-) = \mathbf{S}_\sharp \lambda_\ominus^{(n)}(\boldsymbol{\sigma}_\ominus),\quad
	d\lambda_\ominus^{(n)}(\boldsymbol{\sigma}_\ominus) = \frac{\sigma_0\,d\nu^{(n)}_\ominus(\boldsymbol{\sigma}_\ominus)}{\int A\circ g \,d\nu^{(n)}}.
	\end{equation}
Substituting (\ref{eqIS_24}) into (\ref{eqIS_20}) and using properties of the shift map (\ref{eqIS_3shift}) yields
  	\begin{equation}
	\label{eqIS_25}
	\Lambda'(\boldsymbol{\sigma}_-) = \mathbf{S}_\sharp \lambda'(\boldsymbol{\sigma}_\ominus),\quad
	d\lambda'(\boldsymbol{\sigma}_\ominus) =
	\frac{c_{n+1}}{k_{n}^p\int A\circ g \,d\nu^{(n)}} 
	\left(\prod_{j = 1}^{n} \sigma_{-j}^{p-1}\right)
	\,d\nu^{(n)}_\ominus(\boldsymbol{\sigma}_\ominus).
	\end{equation}
Using (\ref{eqIS_6b}) we finally derive
  	\begin{equation}
	\label{eqIS_26}
	d\lambda'(\boldsymbol{\sigma}_\ominus) =
	\frac{c_{n}}{k_{n}^p} 
	\left(\prod_{j = 1}^{n} \sigma_{-j}^{p-1}\right)
	\,d\nu^{(n)}_\ominus(\boldsymbol{\sigma}_\ominus)
	= d\lambda_p^{(n)}(\boldsymbol{\sigma}_\ominus).
	\end{equation}
Then, expressions (\ref{eqIS_25}) and (\ref{eqIS_26}) yield (\ref{eqIS_20b}).
}

\section{Quotient construction with Galilean transformations}
\label{sec_Gal}

In this section, we study the equivalence relation with respect to Galilean transformations. It is the second symmetry (in addition to temporal scalings), which does not commute with the flow; see Tab.~\ref{tab1} in Section~\ref{secEuler}. Here we develop a quotient construction similar to the one of Section~\ref{sec4}, but using different commutation relations. These two constructions will be put together in the next Section~\ref{sec_fus}. As before, we consider an infinite-dimensional probability measure space $(\mathcal{X},\Sigma,\mu)$ with the measure $\mu$ invariant for a measurable flow $\Phi^t: \mathcal{X} \mapsto \mathcal{X}$. 

\subsection{Symmetries and spatial homogeneity}
\label{subsec_symG}

We explore the equivalence relation with respect to the group of Galilean transformations:
	\begin{equation}
	\label{eq6_G1aa}
	\mathcal{H}_{\mathrm{g}} = \{s_\mathrm{g}^{\mathbf{v}}: \mathbf{v} \in \mathbb{R}^d\}.
 	\end{equation}
Additionally, we consider the group (\ref{eq2_0N}) from Section~\ref{sec4}. It is generated by rotations $s^{\mathbf{Q}}_{\mathrm{r}}$ and scaling maps $s^{a}_{\mathrm{ts}}$ and $s^{b}_{\mathrm{ss}}$ as
	\begin{equation}
	\label{eq6_G1}
	\mathcal{S} = \big\{s^{\mathbf{Q}}_{\mathrm{r}}  
	\circ s^{a}_{\mathrm{ts}} \circ s^{b}_{\mathrm{ss}} :
	\mathbf{Q} \in \mathrm{O}(d),\,
	a > 0,\, b > 0\big\}. 
 	\end{equation}
Also, we consider spatial translations $s_{\mathrm{s}}^{\mathbf{r}}$, which play an auxiliary role. Commutation relations for all these maps and the flow are defined by Tab.~\ref{tab1}. The central relation for this section is
	\begin{equation}
	\label{eq6_3}
	\Phi^t \circ s^{\mathbf{v}}_{\mathrm{g}} =
	s^{\mathbf{v}t}_{\mathrm{s}} \circ s^{\mathbf{v}}_{\mathrm{g}} \circ \Phi^t,
	\end{equation}
implying that Galilean transformations do not commute with the flow. The commuted states are translated by the distance $\mathbf{r}  = \mathbf{v}t$ in physical space $\mathbb{R}^d$. 

We say that the measure $\mu$ is (spatially) \textit{homogeneous} if
	\begin{equation}
	\label{eq6_4}
	\left(s_{\mathrm{s}}^{\mathbf{r}}\right)_\sharp\mu = \mu, \quad 
	\mathbf{r} \in \mathbb{R}^d.
	\end{equation}
This means that $\mu$ is symmetric with respect to all spatial translations. From now on, we restrict our study to homogeneous measures $\mu$. One can check using Tab.~\ref{tab1} that measures $s_\sharp \mu$ are  homogeneous for any map {\color{black}$s \in \mathcal{H}_{\mathrm{g}}$ or $\mathcal{S}$; see (\ref{eq6_G1aa}) and (\ref{eq6_G1}).} Also, due to commutation relation (\ref{eq6_3}), the homogeneity is a necessary and sufficient condition for the invariance of Galilean transformed measures $(s_g^{\mathbf{v}})_\sharp\mu$ under the flow $\Phi^t$. Thus, Galilean transformations are symmetries for homogeneous invariant measures in the sense of Definition~\ref{del_measure}. We remark that spatial homogeneity is a typical assumption in the theory of turbulence~\cite{frisch1999turbulence}.

\subsection{Representative set, periodicity and incompressibility}
\label{subsec_incomp}

We consider the equivalence relation with respect to the group $\mathcal{H}_{\mathrm{g}}$ as
	\begin{equation}
	x \sim x' \quad \textrm{if} \quad x' = s_{\mathrm{g}}^{\mathbf{v}}(x),\ \mathbf{v} \in \mathbb{R}^d,
	\label{eq6_Eq1}
	\end{equation}
i.e., two states are equivalent if they are related by a Galilean transformation for some velocity $\mathbf{v} \in \mathbb{R}^d$. Similarly to Definition~\ref{def1} of Section~\ref{subsec_red}, we introduce a representative set containing a single state from each equivalence class.

\begin{definition}
\label{def_G}
We call $\mathcal{Z} \subset \mathcal{X}$ a representative set (with respect to the group $\mathcal{H}_{\mathrm{g}}$), if the following properties are satisfied. For any $x \in \mathcal{X}$, there exists a unique velocity $\mathbf{v} = \mathbf{V}(x) \in \mathbb{R}^d$ such that $z = s_{\mathrm{g}}^{\mathbf{v}}(x) \in \mathcal{Z}$. The function $\mathbf{V}: \mathcal{X} \mapsto \mathbb{R}^d$ is measurable. 
\end{definition}

{\color{black}Consider any state $x_1 = s_{\mathrm{g}}^{\mathbf{v}}(x)$ from the equivalence class of $x \in \mathcal{X}$. By Definition~\ref{def_G} we have 
	\begin{equation}
	z_1 = s_{\mathrm{g}}^{\mathbf{v}_1}(x_1),\quad
	\mathbf{v}_1 = \mathbf{V}(x_1) = \mathbf{V} \circ s_{\mathrm{g}}^{\mathbf{v}}(x). 
	\label{eq6_2_1extA}
	\end{equation}
Using relation $s_{\mathrm{g}}^{\mathbf{v}_1} \circ s_{\mathrm{g}}^{\mathbf{v}} = s_{\mathrm{g}}^{\mathbf{v}_1+\mathbf{v}}$ from Tab.~\ref{tab1}, we have
	\begin{equation}
	z_1 = s_{\mathrm{g}}^{\mathbf{v}_1}(x_1) = s_{\mathrm{g}}^{\mathbf{v}_1} \circ s_{\mathrm{g}}^{\mathbf{v}}(x) 
	= s_{\mathrm{g}}^{\mathbf{v}_1+\mathbf{v}}(x) = s_{\mathrm{g}}^{\mathbf{V} \circ s_{\mathrm{g}}^{\mathbf{v}}(x)+\mathbf{v}}(x).
	\label{eq6_2_1extB}
	\end{equation}
On the other hand, since $x_1 \sim x$, the uniqueness of a representative state in each equivalence class implies that $z_1 = z = s_{\mathrm{g}}^{\mathbf{V}(x)}(x)$. We conclude that the function $\mathbf{V}(x)$ from Definition~\ref{def_G} has the properties
	\begin{equation}
	\mathbf{V} \circ s_{\mathrm{g}}^{\mathbf{v}}(x) 
	= \mathbf{V}(x)-\mathbf{v}, \quad \mathbf{V}(z) = \mathbf{0}
	\label{eq6_2_1}
	\end{equation}
for any $\mathbf{v} \in \mathbb{R}^d$.} 

Relation (\ref{eq6_3}) is different from the commutation relation (\ref{eq2_2}) for temporal scalings used in Section~\ref{sec4}. This difference affects our quotient construction, for which we need to impose some extra conditions.
The first condition is \textit{periodicity} of a measure $\mu$. It means that there exist linearly independent vectors $\mathbf{e}_1,\ldots,\mathbf{e}_d \in \mathbb{R}^d$ such that 
	\begin{equation}
	\label{eq6_5}
	x = s_{\mathrm{s}}^{\mathbf{e}_1}(x) = \ldots = s_{\mathrm{s}}^{\mathbf{e}_d}(x)
	\end{equation}
for almost every $x \in \mathcal{X}$ with respect to $\mu$.
The period vectors $\mathbf{e}_1,\ldots,\mathbf{e}_d$ {\color{black}may} depend on the measure $\mu$ but not on the state $x$.
Periodicity is not crucial for our construction, but it considerably simplifies the analysis. {\color{black} This property features periodic flows, which} are very common in the theory of turbulence~\cite{frisch1999turbulence}. 

The physical origin of a Galilean transformation is the change to a reference frame moving with a constant velocity $\mathbf{v}$ in physical space $\mathbb{R}^d$. Given a state $x \in \mathcal{X}$, we are going to use $\mathbf{V}(x)$ as a speed of a corresponding reference frame. Considering a solution  $\Phi^t(x)$, we now introduce a reference frame translated in physical space along some trajectory $\mathbf{r} = \mathbf{R}^t(x)$; see Fig.~\ref{fig8}. By $\widetilde{x} = s_{\mathrm{s}}^{\mathbf{r}} \circ \Phi^t (x)$ we represent the state at time $t$ in a reference frame moved to position $\mathbf{r} \in \mathbb{R}^d$. We  set the instantaneous speed of this reference frame to be $\mathbf{V}(\widetilde{x}) =  \mathbf{V} \circ s_{\mathrm{s}}^{\mathbf{r}} \circ \Phi^t (x)$. Assuming that $\mathbf{r} = \mathbf{0}$ at $t = 0$, we obtain the Cauchy problem for the trajectory $\mathbf{r} = \mathbf{R}^t(x)$ in the form
	\begin{equation}
	\frac{d\mathbf{R}^t}{dt} = \mathbf{v}_x(\mathbf{R}^t,t),\quad 
	\mathbf{R}^0 = \mathbf{0},
	\label{eq6_2A_1}
	\end{equation}
with the time-dependent velocity field 
	\begin{equation}
	\mathbf{v}_x(\mathbf{r},t) =  \mathbf{V} \circ s_{\mathrm{s}}^{\mathbf{r}} \circ \Phi^t (x)
	\label{eq6_2_2}
	\end{equation}
in physical space $\mathbb{R}^d$. 
Periodicity conditions (\ref{eq6_5}) and commutation relations of Tab.~\ref{tab1} yield the periodicity of velocity field (\ref{eq6_2_2}) as
	\begin{equation}
	\mathbf{v}_x(\mathbf{r},t) =  \mathbf{v}_x(\mathbf{r}+\mathbf{e}_1,t) 
	= \cdots = 
	\mathbf{v}_x(\mathbf{r}+\mathbf{e}_d,t).
	\label{eq6_2_2b}
	\end{equation}

\begin{figure}
\centering
\includegraphics[width=0.35\textwidth]{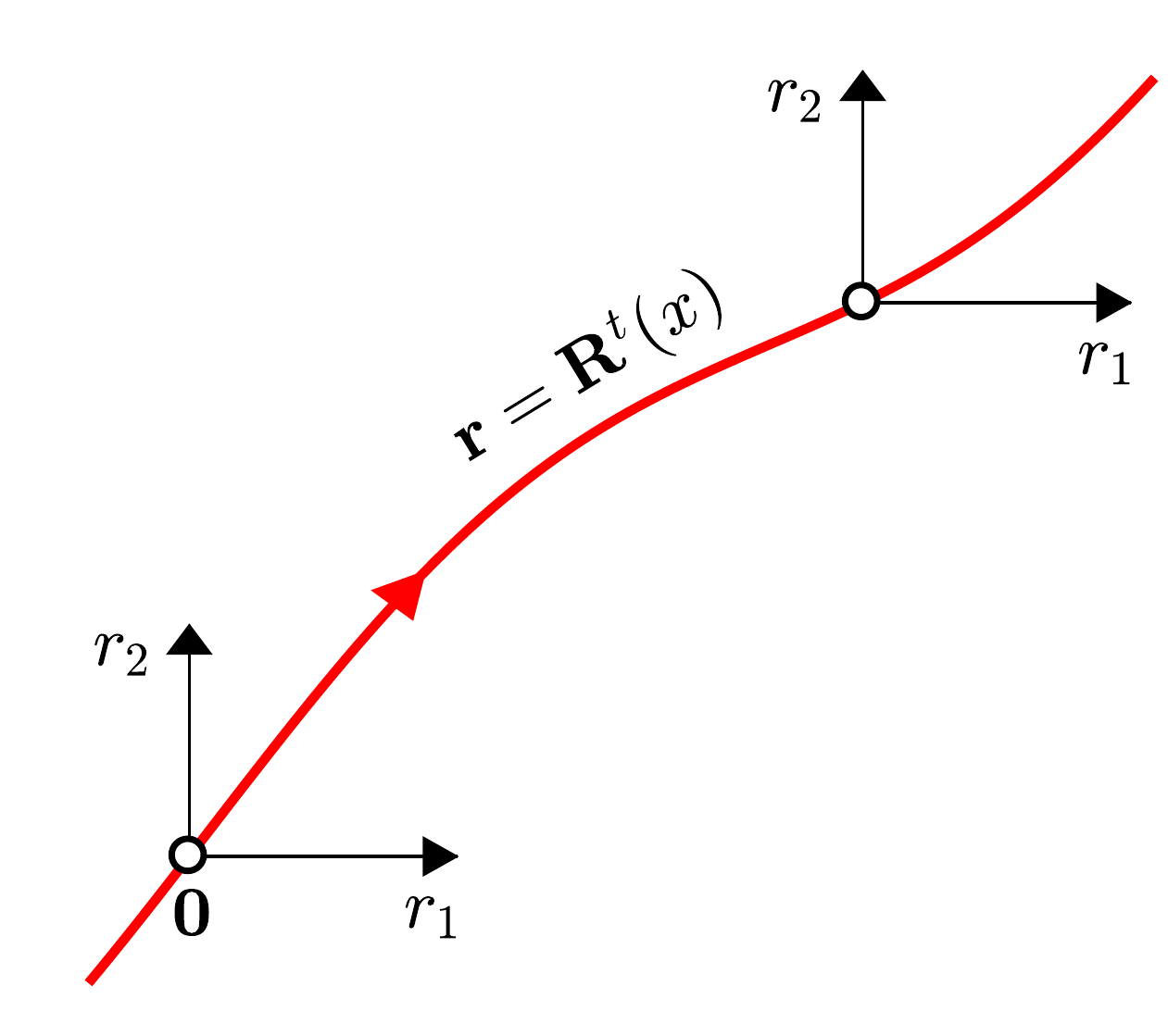}
\caption{Schematic graph of a trajectory $\mathbf{r} = \mathbf{R}^t(x)$ traced by a reference frame in physical space $\mathbb{R}^d$, which starts at $\mathbf{R}^0(x) = \mathbf{0}$ and moves with the speed $\mathbf{v}_x(\mathbf{r},t)$.}
\label{fig8}
\end{figure}

\begin{definition}
\label{def_G2}
Given the function $\mathbf{V}(x)$, we say that the flow is incompressible if the velocity field $\mathbf{v}_x(\mathbf{r},t)$ in (\ref{eq6_2_2}) is continuous in $(\mathbf{r},t)$, continuously differentiable in $\mathbf{r} = (r_1,\ldots,r_d)$, and 
	\begin{equation}
	\mathrm{div}\, \mathbf{v}_x = 0
	\label{eq6_2_3}
	\end{equation}
for all $x$, $\mathbf{r}$ and $t$, where $\mathrm{div} = \partial/\partial r_1+\cdots+\partial/\partial r_d$ is the divergence operator.
\end{definition}

We emphasize that (\ref{eq6_2_3}) is not the incompressibility condition for the phase-space volume in $\mathcal{X}$; instead, it refers to physical space $\mathbb{R}^d$ accessed by means of symmetries. By Picard's theorem (see, e.g.~\cite{teschl2012ordinary}), problem (\ref{eq6_2A_1})
has a unique local solution for velocity fields from Definition~\ref{def_G2}, and periodicity (\ref{eq6_2_2b}) ensures that the solution is defined globally in time. Additionally, we assume that $\mathbf{R}^t(x)$ is measurable as a function of $x$ and $t$. This assumption is natural, because one expects continuous dependence of solutions on initial states $x$ in well-posed problems; {\color{black}see, however, the remark in Section~\ref{secEuler}.}

To give an example, let us consider the Euler system from Section~\ref{secEuler} with $x = \mathbf{u}_0(\mathbf{r})$ representing a fluid velocity field at initial time. For the function $\mathbf{V}(x)$, the simplest choice is  
	\begin{equation}
	\mathbf{V}(x) = \mathbf{u}_0(\mathbf{0})
	\label{eq6_2_4}
	\end{equation}
corresponding to the velocity at $\mathbf{r} = \mathbf{0}$. Expression (\ref{eq6_2_2}) with the fluid velocity $\mathbf{u}(\mathbf{r},t) = \Phi^t(x)$ and $s_{\mathrm{s}}^{\mathbf{r}}$ from (\ref{eq3}) yield
	\begin{equation}
	\mathbf{v}_x(\mathbf{r},t) = \mathbf{u}(\mathbf{r},t). 
	\label{eq6_2_5}
	\end{equation}
We see that (\ref{eq6_2_3}) is exactly the fluid incompressibility condition, and $\mathbf{r} = \mathbf{R}^t(x)$ in (\ref{eq6_2A_1}) is the Lagrangian (particle) trajectory that starts at the origin at $t = 0$. 

\subsection{Normalized flow and invariant measure}
\label{subsec_NFG}

Let us introduce a measurable \textit{Galilean projector} $Q: \mathcal{X} \mapsto \mathcal{Z}$ as 
	\begin{equation}
	Q(x) = s_{\mathrm{g}}^{\mathbf{V}(x)}(x). 
	\label{eq6_3_1}
	\end{equation}
We now define the normalized flow with the invariant measure for the representative set $\mathcal{Z}$.
	
\begin{theorem}
\label{theorem_G}
Consider a flow $\Phi^t$ with an invariant measure $\mu$ and a representative set $\mathcal{Z}$, which satisfy the properties of homogeneity, periodicity and incompressibility. Then, the mapping 
	\begin{equation}
	\Omega^t(z) 
	= Q \circ s_{\mathrm{s}}^{\mathbf{R}^t(z)} \circ \Phi^t(z), \quad
	z \in \mathcal{Z},
	\label{eq6_3_2}
	\end{equation}
defines the flow $\Omega^t: \mathcal{Z} \mapsto \mathcal{Z}$ in the representative set $\mathcal{Z}$ with the normalized invariant measure
	\begin{equation}
	\zeta = Q_\sharp \mu. 
	\label{eq6_3_3}
	\end{equation}
Here $\mathbf{R}^t(z)$ is a solution of (\ref{eq6_2A_1}) assumed to be a measurable function of $z$ and $t$.
\end{theorem}

Proofs of all statements are collected in Section~\ref{subsec_proof_G}. The two assumptions of  homogeneity and incompressibility are crucial in the proof, where homogeneity allows averaging in physical space and incompressibility yields the volume-preserving property for a change of integration variables. Generally, violation of homogeneity or incompressibility breaks invariance of the normalized measure. The normalized flow (\ref{eq6_3_2}) can be seen as a reduction of every solution $\Phi^t(z)$ to a reference frame that moves along the trajectory $\mathbf{r} = \mathbf{R}^t(z)$ in physical space. Since these reference frames are different for different solutions, invariance of the normalized measure (\ref{eq6_3_3}) expressed in terms of $\mu$ is a remarkable property owing to homogeneity and incompressibility. 

The example in (\ref{eq6_2_4}) and (\ref{eq6_2_5}) provides the physical interpretation of the normalized system. It is the Quasi--Lagrangian representation {\color{black}\cite{belinicher1987scale,l1991scale}} for incompressible fluid dynamics, describing velocity fields in reference frames moving with selected fluid particles. 

Analogously to Proposition~\ref{prop_h} in Section~\ref{sec4}, we show that Galilean transformations act trivially on normalized measures.

\begin{proposition}
\label{prop_G1}
All Galilean transformed invariant measures $ \widetilde{\mu} = \left(s_{\mathrm{g}}^{\mathbf{v}}\right)_\sharp \mu$  with $\mathbf{v} \in \mathbb{R}^d$ yield the same normalized measure $\zeta = Q_\sharp \widetilde{\mu}$ by Theorem~\ref{theorem_G}.
\end{proposition}

\subsection{Symmetries in the normalized system}
\label{subsec_sym_G}

Here we show that the group $\mathcal{S}$ from (\ref{eq6_G1}) extends to the normalized system, provided that the velocity $\mathbf{V}(x)$ has proper transformation properties.

\begin{proposition}
\label{prop_Gsym}
Let us assume that the function $\mathbf{V}(x)$ satisfies the conditions
	\begin{equation}
	\label{eq6_4_S1}
	\mathbf{V} \circ s^{\mathbf{Q}}_{\mathrm{r}}(x) = \mathbf{Q}^{-1}\mathbf{V}(x),\quad
	\mathbf{V} \circ s^a_{\mathrm{ts}}(x) = \frac{\mathbf{V}(x)}{a},\quad
	\mathbf{V} \circ s^b_{\mathrm{ss}}(x) = b\mathbf{V}(x).
	\end{equation}
Then, the projector $Q$ commutes with all elements $s \in \mathcal{S}$. 
\end{proposition}

One can see that conditions (\ref{eq6_4_S1}) describe natural rules for transformations of velocity vectors under spatial rotations and scalings, as described by relations (\ref{eq3}) in Section~\ref{secEuler}. Since $\mathcal{Z} = Q(\mathcal{X})$, the commutativity in Proposition~\ref{prop_Gsym} implies that $s(\mathcal{Z}) = \mathcal{Z}$ and, hence, elements $s \in \mathcal{S}$ can be considered as the maps $s: \mathcal{Z} \mapsto \mathcal{Z}$.

\begin{theorem}
\label{theorem_G2}
Under conditions of Theorem~\ref{theorem_G} and Proposition~\ref{prop_Gsym}, mappings $s \in \mathcal{S}$ are symmetries of the normalized system, i.e., measures $s_\sharp\zeta$ are invariant for the normalized flow $\Omega^t$ (see Definition \ref{del_measure}). Commutation relations for these symmetries and the flow take the form
	\begin{equation}
	\label{eq6_4_S5}
	\Omega^t \circ s^{\mathbf{Q}}_{\mathrm{r}}
	= s^{\mathbf{Q}}_{\mathrm{r}} \circ \Omega^t, \quad 
	\Omega^t \circ s^a_{\mathrm{ts}}
	= s^a_{\mathrm{ts}} \circ \Omega^{t/a}, \quad
	\Omega^t \circ s^{b}_{\mathrm{ss}} 
	= s^{b}_{\mathrm{ss}} \circ \Omega^t,
	\end{equation}
the same as in the original system (see Tab.~\ref{tab1}).
\end{theorem}

We see that the normalized system inherits the symmetry group $\mathcal{S}$ of the original system together with the commutation relations. We remark that spatial translations $s_{\mathrm{s}}^{\mathbf{r}}$ could be extended to the normalized system in a similar way if one assumes $\mathbf{V} \circ s_{\mathrm{s}}^{\mathbf{r}}(x) = \mathbf{V}(x)$. But this condition is not of our interest: in the fluid dynamical representation it forbids the choice (\ref{eq6_2_4}), which associates the normalized flow with particle trajectories (\ref{eq6_2_5}). At the same time, one can check using (\ref{eq3}) that the choice (\ref{eq6_2_4}) satisfies all conditions in (\ref{eq6_4_S1}).

Similarly to Corollary~\ref{corr_sym_mes} from Section~\ref{subsec_SNM}, we formulate the symmetry relation between the original and normalized systems. It follows from Propositions~\ref{prop_G1} and \ref{prop_Gsym} as

\begin{corollary}\label{corr_sym_mes2}
If the measure $\mu$ is symmetric with respect to $s \circ s_{\mathrm{g}}^{\mathbf{v}}$ for some $s \in \mathcal{S}$ and $s_{\mathrm{g}}^{\mathbf{v}} \in \mathcal{H}_{\mathrm{g}}$, then the normalized measure $\zeta$ is symmetric with respect to $s$:
	\begin{equation}
	\label{eq2_A8_G}
	(s \circ s_{\mathrm{g}}^{\mathbf{v}})_\sharp \mu = \mu \quad \Rightarrow \quad 
	s_\sharp \zeta = \zeta.
	\end{equation}
\end{corollary}

\subsection{Proofs of Theorems~\ref{theorem_G} and \ref{theorem_G2} and Propositions \ref{prop_G1} and \ref{prop_Gsym}} 
\label{subsec_proof_G}

We first formulate and prove a few lemmas.

\begin{lemma}\label{lemma_G1}
The function $\mathbf{v}_x(\mathbf{r},t)$ from (\ref{eq6_2_2}) satisfies the following identities
	\begin{eqnarray}
	\mathbf{v}_z(\mathbf{r},t) & = & 
	\mathbf{v}_x\big(\mathbf{r}+\mathbf{V}(x)t,t\big)-\mathbf{V}(x), \quad z = Q(x);
	\label{eq6_P_1} \\[3pt]
	\mathbf{v}_x(\mathbf{r}+\mathbf{r}',t) & = & \mathbf{v}_{x'}(\mathbf{r},t), \quad
	x' = s_{\mathrm{s}}^{\mathbf{r}'}(x).
	\label{eq6_P_2}
	\end{eqnarray}
\end{lemma}

\begin{proof}
Using definitions (\ref{eq6_2_2}) and (\ref{eq6_3_1}), we write
	\begin{equation}
	\mathbf{v}_z(\mathbf{r},t) 
	= \mathbf{V} \circ s_{\mathrm{s}}^{\mathbf{r}} \circ \Phi^t (z) 
	= \mathbf{V} \circ s_{\mathrm{s}}^{\mathbf{r}} \circ \Phi^t \circ Q (x) 
	= \mathbf{V} \circ s_{\mathrm{s}}^{\mathbf{r}} \circ \Phi^t \circ s_{\mathrm{g}}^{\mathbf{V}(x)}(x). 
	\label{eq6_P_3}
	\end{equation}
Using commutation relations of Tab.~\ref{tab1}, we obtain
	\begin{equation}
	\mathbf{v}_z(\mathbf{r},t) 
	= \mathbf{V} \circ s_{\mathrm{g}}^{\mathbf{V}(x)} 
	\circ s_{\mathrm{s}}^{\mathbf{r}+\mathbf{V}(x)t} 
	\circ \Phi^t(x). 
	\label{eq6_P_4}
	\end{equation}
Using (\ref{eq6_2_1}) and (\ref{eq6_2_2}) in (\ref{eq6_P_4}) yields  the identity (\ref{eq6_P_1}) as
	\begin{equation}
	\mathbf{v}_z(\mathbf{r},t) 
	= \mathbf{V} \circ s_{\mathrm{s}}^{\mathbf{r}+\mathbf{V}(x)t} 
	\circ \Phi^t(x)-\mathbf{V}(x)
	= \mathbf{v}_x\big(\mathbf{r}+\mathbf{V}(x)t,t\big)-\mathbf{V}(x). 
	\label{eq6_P_5}
	\end{equation}
Equality (\ref{eq6_P_2}) is obtained using definition (\ref{eq6_2_2}) and relations of Tab.~\ref{tab1} as	\begin{equation}
	\mathbf{v}_x(\mathbf{r}+\mathbf{r}',t) 
	= \mathbf{V} \circ s_{\mathrm{s}}^{\mathbf{r}+\mathbf{r}'} \circ \Phi^t (x) 
	= \mathbf{V} \circ s_{\mathrm{s}}^{\mathbf{r}} \circ \Phi^t \circ s_{\mathrm{s}}^{\mathbf{r}'}  (x) 
	= \mathbf{V} \circ s_{\mathrm{s}}^{\mathbf{r}} \circ \Phi^t(x') 
	=  \mathbf{v}_{x'}(\mathbf{r},t). 
	\label{eq6_P_3ex}
	\end{equation}
\end{proof}

\begin{lemma}\label{lemma_G1b}
The function $\mathbf{R}^t(x)$ defined by (\ref{eq6_2A_1}) satisfies the following identities 
	\begin{eqnarray}
	\mathbf{R}^t(z) & = & \mathbf{R}^t(x)-\mathbf{V}(x)t, 
	\quad z = Q(x);
	\label{eq6_P_1b} \\[3pt]
	\mathbf{R}^{t_1+t_2}(x) & = & \mathbf{R}^{t_1}(x)+\mathbf{R}^{t_2}(x_1),
	\quad x_1 = s_{\mathrm{s}}^{\mathbf{R}^{t_1}(x)} \circ \Phi^{t_1}(x). 
	\label{eq6_P_2b}
	\end{eqnarray}
\end{lemma}

\begin{proof}
Using (\ref{eq6_2A_1}), we write the Cauchy problem for the function $\mathbf{R}^t(z)$ as
	\begin{equation}
	\frac{d}{dt}\, \mathbf{R}^t(z) 
	= \mathbf{v}_z\big(\mathbf{R}^t(z),t\big),\quad 
	\mathbf{R}^0(z) = \mathbf{0}.
	\label{eq6_P_6}
	\end{equation}
Using (\ref{eq6_P_1}), we write (\ref{eq6_P_6}) as
	\begin{equation}
	\frac{d}{dt}\, \mathbf{R}^t(z) +\mathbf{V}(x)
	= \mathbf{v}_x\big(\mathbf{R}^t(z)+\mathbf{V}(x)t,t\big),\quad 
	\mathbf{R}^0(z) = \mathbf{0}.
	\label{eq6_P_7}
	\end{equation}
These expressions can be written as
	\begin{equation}
	\frac{d\widetilde{\mathbf{R}}^t}{dt} 
	= \mathbf{v}_x\big(\widetilde{\mathbf{R}}^t,t\big),\quad 
	\widetilde{\mathbf{R}}^0 = \mathbf{0},\quad 
	\widetilde{\mathbf{R}}^t = \mathbf{R}^t(z)+\mathbf{V}(x)t.
	\label{eq6_P_8}
	\end{equation}
Comparison with (\ref{eq6_2A_1}) yields $\widetilde{\mathbf{R}}^t = \mathbf{R}^t(x)$ proving (\ref{eq6_P_1b}).

Fixing $t_1$, let us consider 
	\begin{equation}
	\widehat{\mathbf{R}}^{t_2} = \mathbf{R}^{t_1+t_2}(x)-\mathbf{R}^{t_1}(x)
	\label{eq6_P_A1}
	\end{equation}
as a function of $t_2$. Using (\ref{eq6_2A_1}), we have
	\begin{equation}
	\frac{d}{dt_2}\, \widehat{\mathbf{R}}^{t_2} 
	= \frac{d}{dt_2}\, \mathbf{R}^{t_1+t_2}(x) 
	= \mathbf{v}_x\big(\mathbf{R}^{t_1+t_2}(x),t_1+t_2\big),\quad 
	\widehat{\mathbf{R}}^0 = \mathbf{0}.
	\label{eq6_P_A2}
	\end{equation}
Expressing $\mathbf{v}_x$ from (\ref{eq6_2_2}) yields
	\begin{equation}
	\mathbf{v}_x\big(\mathbf{R}^{t_1+t_2}(x),t_1+t_2\big)
	= \mathbf{V} \circ s_{\mathrm{s}}^{\mathbf{R}^{t_1+t_2}(x)} \circ \Phi^{t_1+t_2}(x)
	= \mathbf{V} \circ s_{\mathrm{s}}^{\widehat{\mathbf{R}}^{t_2}} 
	\circ \Phi^{t_2}(x_1),
	\label{eq6_P_A3}
	\end{equation}
where the last equality is derived using commutation relations of Tab.~\ref{tab1}, $\widehat{\mathbf{R}}^{t_2}$ from (\ref{eq6_P_A1}) and $x_1$ from (\ref{eq6_P_2b}). Using (\ref{eq6_2_2}) and (\ref{eq6_P_A3}), we reduce (\ref{eq6_P_A2}) to the Cauchy problem
	\begin{equation}
	\frac{d}{dt_2}\, \widehat{\mathbf{R}}^{t_2} 
	= \mathbf{v}_{x_1}\big(\widehat{\mathbf{R}}^{t_2},t_2\big),\quad 
	\widehat{\mathbf{R}}^0 = \mathbf{0}.
	\label{eq6_P_A4}
	\end{equation}
According to (\ref{eq6_2A_1}), this implies that $\widehat{\mathbf{R}}^{t_2} = \mathbf{R}^{t_2}(x_1)$ and, hence, we derive (\ref{eq6_P_2b}) from (\ref{eq6_P_A1}).
\end{proof}

\begin{lemma}\label{lemma_G2}
For an incompressible flow, the expression
	\begin{equation}
	\mathbf{T}_x^t(\mathbf{r}) = \mathbf{r}+\mathbf{R}^t \circ s_{\mathrm{s}}^{\mathbf{r}} (x)
	\label{eq6_P_9}
	\end{equation}
defines a volume-preserving map $\mathbf{T}_x^t: \mathbb{R}^d \mapsto \mathbb{R}^d$ for all $x$ and $t$. 
\end{lemma}

\begin{proof}
Denoting $\widetilde{x} = s_{\mathrm{s}}^{\mathbf{r}} (x)$, we write $\mathbf{T}_x^t(\mathbf{r}) = \mathbf{r}+\mathbf{R}^t(\widetilde{x})$. Then, using (\ref{eq6_2A_1}) for $\widetilde{x}$, we obtain
	\begin{equation}
	\frac{d\mathbf{T}_x^t}{dt} = \mathbf{v}_{\widetilde{x}}
	\left(\mathbf{R}^t(\widetilde{x}),t\right).
	\label{eq6_P_12}
	\end{equation}
Substituting (\ref{eq6_2_2}) yields
	\begin{equation}
	\frac{d\mathbf{T}_x^t}{dt} = 
	\mathbf{V} \circ s_{\mathrm{s}}^{\mathbf{R}^t(\widetilde{x})} 
	\circ \Phi^t (\widetilde{x}).
	\label{eq6_P_13}
	\end{equation}
Since $\widetilde{x} = s_{\mathrm{s}}^{\mathbf{r}} (x)$, we modify expression (\ref{eq6_P_13}) as
	\begin{equation}
	\frac{d\mathbf{T}_x^t}{dt} = 
	\mathbf{V} \circ s_{\mathrm{s}}^{\mathbf{R}^t(\widetilde{x})} 
	\circ \Phi^t \circ s_{\mathrm{s}}^{\mathbf{r}} (x) 
	= \mathbf{V} \circ s_{\mathrm{s}}^{\mathbf{r}+\mathbf{R}^t(\widetilde{x})} \circ \Phi^t  (x)
	= \mathbf{V} \circ s_{\mathrm{s}}^{\mathbf{T}_x^t} \circ \Phi^t  (x),
	\label{eq6_P_14}
	\end{equation}
where we used commutation relations of Tab.~\ref{tab1}. 
Using (\ref{eq6_2_2}) in (\ref{eq6_P_14}) yields
	\begin{equation}
	\frac{d\mathbf{T}_x^t}{dt} = \mathbf{v}_x(\mathbf{T}_x^t,t).
	\label{eq6_P_15}
	\end{equation}
The incompressibility condition in Definition~\ref{def_G2} implies that the vector field in the right-hand side of (\ref{eq6_P_15}) is divergence-free. Also, since $\mathbf{R}^0 = \mathbf{0}$ in (\ref{eq6_2A_1}), the mapping $\mathbf{T}_x^0(\mathbf{r}) = \mathbf{r}$ is volume-preserving at $t = 0$. By Liouville’s theorem for divergence-free fields~\cite[\S{V.3}]{hartman2002ordinary}, $\mathbf{T}_x^t(\mathbf{r})$ is volume-preserving for any $t$. \end{proof}

\begin{proof}[Proof of Theorem~\ref{theorem_G}]
First, we prove that $\zeta$ is invariant, i.e., $\Omega^t_\sharp\zeta = \zeta$. Using (\ref{eq6_3_3}), we write
	\begin{equation}
	\Omega^t_\sharp\zeta 
	= \widehat{\Omega}^t_\sharp \mu, 
	\label{eq6_P_16}
	\end{equation}
where we introduced the mapping $\widehat{\Omega}^t: \mathcal{X} \mapsto \mathcal{Z}$ as
	\begin{equation}
	\widehat{\Omega}^t = \Omega^t \circ Q.
	\label{eq6_P_16b}
	\end{equation}
Using (\ref{eq6_3_2}) and (\ref{eq6_3_1}), we obtain
	\begin{equation}
	\widehat{\Omega}^t(x) 
	= Q \circ s_{\mathrm{s}}^{\mathbf{R}^t(z)} \circ \Phi^t \circ 
	s_{\mathrm{g}}^{\mathbf{V}(x)}(x),\quad 
	z = Q (x).
	\label{eq6_P_17}
	\end{equation}
Then, commutation relations of Tab.~\ref{tab1} yield
	\begin{equation}
	\widehat{\Omega}^t(x) 
	= Q \circ s_{\mathrm{g}}^{\mathbf{V}(x)} 
	\circ s_{\mathrm{s}}^{\mathbf{R}^t(z)+\mathbf{V}(x)t} \circ \Phi^t(x).
	\label{eq6_P_18}
	\end{equation}
Using (\ref{eq6_3_1}) and (\ref{eq6_2_1}), we obtain the identity
	\begin{equation}
	Q \circ s_{\mathrm{g}}^{\mathbf{v}}(x) 
	= s_{\mathrm{g}}^{\mathbf{V}\circ s_{\mathrm{g}}^{\mathbf{v}}(x)+\mathbf{v}}(x) 
	= s_{\mathrm{g}}^{\mathbf{V}(x)}(x) = Q(x)
	\label{eq6_P_19}
	\end{equation}
for any $x$ and $\mathbf{v}$. Using (\ref{eq6_P_1b}) and (\ref{eq6_P_19}), we reduce (\ref{eq6_P_18}) to the form
	\begin{equation}
	\widehat{\Omega}^t(x) 
	= Q \circ s_{\mathrm{s}}^{\mathbf{R}^t(x)} \circ \Phi^t(x).
	\label{eq6_P_20}
	\end{equation}
	
Now we use the homogeneity property (\ref{eq6_4}) and express (\ref{eq6_P_16}) as
	\begin{equation}
	\Omega^t_\sharp\zeta 
	= \left(\widehat{\Omega}_{\mathbf{r}}^t\right)_\sharp \mu
	\label{eq6_P_21}
	\end{equation}
with the new mapping $\widehat{\Omega}_{\mathbf{r}}^t: \mathcal{X} \mapsto \mathcal{Z}$ given by
	\begin{equation}
	\widehat{\Omega}_{\mathbf{r}}^t(x) = \widehat{\Omega}^t \circ s_{\mathrm{s}}^{\mathbf{r}}(x)
	\label{eq6_P_22}
	\end{equation}
for any $\mathbf{r} \in \mathbb{R}^d$. Using (\ref{eq6_P_20}), relations of Tab.~\ref{tab1} and (\ref{eq6_P_9}), we express (\ref{eq6_P_22}) as
	\begin{equation}
	\widehat{\Omega}_{\mathbf{r}}^t(x) 
	= Q \circ s_{\mathrm{s}}^{\mathbf{R}^t \circ
	s_{\mathrm{s}}^{\mathbf{r}}(x)} \circ \Phi^t \circ s_{\mathrm{s}}^{\mathbf{r}}(x)
	= Q \circ s_{\mathrm{s}}^{\mathbf{T}_x^t(\mathbf{r})} \circ \Phi^t(x).
	\label{eq6_P_23}
	\end{equation}

Considering a measure on $\mathcal{X} \times \mathbb{R}^d$ as a product of $\mu$ and $d$-dimensional volume measure, Fubini's theorem~\cite{rudin2006real} allows integrating equality (\ref{eq6_P_21}) with respect to $\mathbf{r}$. Taking into account the periodicity property (\ref{eq6_5}), we integrate over the periodic domain 
	\begin{equation}
	\label{eq6_P_24vol}
	\mathcal{T} = \left\{\mathbf{r} \in \mathbb{R}^d: \
	\mathbf{r} = a_1 \mathbf{e}_1+\cdots+a_d \mathbf{e}_d,\ a_1,\ldots,a_d \in [0,1] \right\}. 
	\end{equation}
Since the left-hand side in (\ref{eq6_P_21}) does not depend on $\mathbf{r}$, this yields
	\begin{equation}
	V_{\mathcal{T}}\,\Omega^t_\sharp\zeta 
	= \int_{\mathcal{T}}
	\left(\widehat{\Omega}_{\mathbf{r}}^t\right)_\sharp \mu\, d\mathbf{r}
	\label{eq6_P_24}
	\end{equation}
with the volume $V_{\mathcal{T}} = \int_{\mathcal{T}} d\mathbf{r}$. Now we change the integration variable as $\mathbf{r}' = \mathbf{T}_x^t(\mathbf{r})$. There is no extra (Jacobian determinant) factor in the new integral expression, because $\mathbf{T}_x^t$ is volume-preserving by Lemma~\ref{lemma_G2}. As a result, we write (\ref{eq6_P_24}) as
	\begin{equation}
	\Omega^t_\sharp\zeta 
	= \frac{1}{V_{\mathcal{T}}}\int_{\mathcal{T}'}
	\left(Q \circ s_{\mathrm{s}}^{\mathbf{r}'} \circ \Phi^t
	\right)_\sharp \mu\, d\mathbf{r}', \quad
	\mathcal{T}' = \mathbf{T}_x^t(\mathcal{T}),
	\label{eq6_P_25}
	\end{equation}
where we substituted expression (\ref{eq6_P_23}) written in terms of the new vector $\mathbf{r}'$. 
Since $\mu$ is invariant and homogeneous, the push-forward in (\ref{eq6_P_25}) is evaluated as
	\begin{equation}
	\left(Q \circ s_{\mathrm{s}}^{\mathbf{r}'} \circ \Phi^t
	\right)_\sharp \mu  = 
	Q_\sharp \mu  = \zeta,
	\label{eq6_P_26}
	\end{equation}
where we also used (\ref{eq6_3_3}). This expression does not depend on $\mathbf{r}'$. Using (\ref{eq6_P_26}) in (\ref{eq6_P_25}) yields
	\begin{equation}
	\Omega^t_\sharp\zeta 
	= \left(\frac{1}{V_{\mathcal{T}}}\int_{\mathcal{T}'} d\mathbf{r}'\right) \zeta
	= \zeta,
	\label{eq6_P_27}
	\end{equation}
because $\mathcal{T}'$ has the same volume as $\mathcal{T}$. This proves the invariance of $\zeta$.

It remains to prove that $\Omega^t$ is a flow. Since $\mathbf{R}^0 = \mathbf{0}$ from (\ref{eq6_2A_1}), expression (\ref{eq6_3_2}) yields $\Omega^0 = Q$, which is the identity map in $\mathcal{Z}$. We have to check the composition relation 
	\begin{equation}
	\Omega^{t_2} \circ \Omega^{t_1}(z) = \Omega^{t_1+t_2}(z).
	\label{eq6_P_28}
	\end{equation}
Denoting $z_1 = \Omega^{t_1}(z)$ and using (\ref{eq6_3_2}) and (\ref{eq6_3_1}), we write
	\begin{equation}
	z_1 = Q \circ s_{\mathrm{s}}^{\mathbf{R}^{t_1}(z)} \circ \Phi^{t_1}(z)
	= s_{\mathrm{g}}^{\mathbf{V}(x_1)} 
	\circ s_{\mathrm{s}}^{\mathbf{R}^{t_1}(z)} \circ \Phi^{t_1}(z)
	= Q(x_1),\quad
	x_1 = s_{\mathrm{s}}^{\mathbf{R}^{t_1}(z)} \circ \Phi^{t_1}(z).
	\label{eq6_P_29}
	\end{equation}
Using (\ref{eq6_3_2}) and (\ref{eq6_P_29}) we obtain
	\begin{equation}
	\Omega^{t_2} \circ \Omega^{t_1}(z) 
	= \Omega^{t_2}(z_1) 
	= Q \circ s_{\mathrm{s}}^{\mathbf{R}^{t_2}(z_1)} \circ \Phi^{t_2}(z_1)
	= Q \circ s_{\mathrm{s}}^{\mathbf{R}^{t_2}(z_1)} \circ \Phi^{t_2}
	\circ s_{\mathrm{g}}^{\mathbf{V}(x_1)} 
	\circ s_{\mathrm{s}}^{\mathbf{R}^{t_1}(z)} \circ \Phi^{t_1}(z).
	\label{eq6_P_30}
	\end{equation}
Commutation relations of Tab.~\ref{tab1} and property (\ref{eq6_P_19}) yield
	\begin{equation}
	\Omega^{t_2} \circ \Omega^{t_1}(z) 
	= Q \circ s_{\mathrm{s}}^{\mathbf{r}} \circ \Phi^{t_1+t_2}(z),
	\quad
	\mathbf{r} = \mathbf{R}^{t_2}(z_1)+\mathbf{V}(x_1)t_2+\mathbf{R}^{t_1}(z).
	\label{eq6_P_31}
	\end{equation}
Applying identities of Lemma~\ref{lemma_G1b} with $x_1$ from (\ref{eq6_P_29}), we have
	\begin{equation}
	\mathbf{r} = \mathbf{R}^{t_2}(x_1)+\mathbf{R}^{t_1}(z) = \mathbf{R}^{t_1+t_2}(z).
	\label{eq6_P_33}
	\end{equation}
Using this expression with (\ref{eq6_3_2}) in (\ref{eq6_P_31}), we obtain (\ref{eq6_P_28}).
\end{proof}

\begin{proof}[Proof of Proposition~\ref{prop_G1}] 
It is a direct consequence of identity (\ref{eq6_P_19}).
\end{proof}

\begin{proof}[Proof of Proposition~\ref{prop_Gsym}]
Let us consider the first relation in (\ref{eq6_4_S1}). Using (\ref{eq6_3_1}), we obtain
	\begin{equation}
	Q \circ s_{\mathrm{r}}^{\mathbf{Q}}(x) 
	= s_{\mathrm{g}}^{\mathbf{V}\circ s_{\mathrm{r}}^{\mathbf{Q}}(x)}
	\circ  s_{\mathrm{r}}^{\mathbf{Q}}(x)
 	= s_{\mathrm{g}}^{\mathbf{Q}^{-1}\mathbf{V}(x)}
	\circ  s_{\mathrm{r}}^{\mathbf{Q}}(x)
	= s_{\mathrm{r}}^{\mathbf{Q}}(x) \circ s_{\mathrm{g}}^{\mathbf{V}(x)}(x) 
	= s_{\mathrm{r}}^{\mathbf{Q}} \circ Q(x),
	\label{eq6_PP6_1}
	\end{equation}
where we used commutation relations of Tab.~\ref{tab1}. This proves the commutativity of $Q$ with $s_{\mathrm{r}}^{\mathbf{Q}}$. One can check that similar derivations yield the commutativity of $Q$ with the remaining generators $s^a_{\mathrm{ts}}$ and $s^b_{\mathrm{ss}}$ of the group (\ref{eq6_G1}).
\end{proof}

For the proof of Theorem~\ref{theorem_G2} we need 
\begin{lemma}
The following identities hold:
	\begin{eqnarray}
	s^{\mathbf{Q}}_{\mathrm{r}} \circ s_{\mathrm{s}}^{\mathbf{R}^t(x)}(x) 
	& = & \displaystyle
	s_{\mathrm{s}}^{\mathbf{R}^t \circ 
	s^{\mathbf{Q}}_{\mathrm{r}}(x)} \circ s^{\mathbf{Q}}_{\mathrm{r}}(x),
	\label{eq6_P_37stA}	
	\\[2pt]
	s^{a}_{\mathrm{ts}} \circ s_{\mathrm{s}}^{\mathbf{R}^t(x)}(x) 
	& = & \displaystyle
	s_{\mathrm{s}}^{\mathbf{R}^{at}\circ s^{a}_{\mathrm{ts}}(x)} \circ s^{a}_{\mathrm{ts}}(x),
	\label{eq6_P_37stB}	
	\\[2pt]
	s^{b}_{\mathrm{ss}} \circ s_{\mathrm{s}}^{\mathbf{R}^t(x)}(x) 
	& = & \displaystyle
	s_{\mathrm{s}}^{\mathbf{R}^t \circ s^b_{\mathrm{ss}}(x)} \circ s^{b}_{\mathrm{ss}}(x).
	\label{eq6_P_37stC}	
	\end{eqnarray}
\end{lemma}
\begin{proof}
Using commutation relations of Tab.~\ref{tab1}, equalities (\ref{eq6_P_37stA})--(\ref{eq6_P_37stC}) reduce to the relations
	\begin{equation}
	\mathbf{R}^t \circ s^{\mathbf{Q}}_{\mathrm{r}}(x) = \mathbf{Q}^{-1}\mathbf{R}^t(x),\quad	
	\mathbf{R}^{at}\circ s^a_{\mathrm{ts}}(x) = \mathbf{R}^t(x),\quad
	\mathbf{R}^t \circ s^b_{\mathrm{ss}}(x) = b\mathbf{R}^t(x).
	\label{eq6_P_37}
	\end{equation}
Since the proof is similar for each relation in (\ref{eq6_P_37}), we demonstrate it only in the case of temporal scaling $s^{a}_{\mathrm{ts}}$. In this case, we represent the left-hand side of the second equality in (\ref{eq6_P_37}) as
	\begin{equation}
	\mathbf{R}^{at}\circ s^{a}_{\mathrm{ts}}(x) = \mathbf{R}^{at}(x'), \quad
	x' = s^{a}_{\mathrm{ts}}(x).
	\label{eq6_P_37bb}
	\end{equation}
The function $\mathbf{r}(t) = \mathbf{R}^t(x')$ solves the problem (\ref{eq6_2A_1}) written as
	\begin{equation}
	\frac{d\mathbf{r}}{dt} = \mathbf{v}_{x'}(\mathbf{r},t),\quad 
	\mathbf{r}(0) = \mathbf{0}.
	\label{eq6_P_37b}
	\end{equation}
We express the vector field using (\ref{eq6_2_2}) and $x' = s^{a}_{\mathrm{ts}}(x)$ as
	\begin{equation}
	\mathbf{v}_{x'}(\mathbf{r},t) 
	=  \mathbf{V} \circ s_{\mathrm{s}}^{\mathbf{r}} \circ \Phi^t \circ s^{a}_{\mathrm{ts}}(x)
	= \mathbf{V} \circ s^{a}_{\mathrm{ts}} \circ s_{\mathrm{s}}^{\mathbf{r}} \circ \Phi^{t/a} (x)
	= \frac{1}{a}\, \mathbf{V} \circ s_{\mathrm{s}}^{\mathbf{r}} \circ \Phi^{t/a} (x)
	= \frac{\mathbf{v}_x(\mathbf{r},t/a)}{a},
	\label{eq6_P_37c}
	\end{equation}
where we used commutation relations of Tab.~\ref{tab1} and properties (\ref{eq6_4_S1}). 
Using (\ref{eq6_P_37c}) in (\ref{eq6_P_37b}) and comparing with (\ref{eq6_2A_1}) one finds $\mathbf{r}(t) = \mathbf{R}^{t/a}(x)$. Hence, $\mathbf{R}^{at}(x') = \mathbf{R}^t(x)$, proving the second equality of (\ref{eq6_P_37}).
\end{proof}

\begin{proof}[Proof of Theorem~\ref{theorem_G2}]
Using (\ref{eq6_3_3}) and commutation property of Proposition~\ref{prop_Gsym}, we have 
	\begin{equation}
	s_\sharp \zeta = \left(s \circ Q\right)_\sharp \mu
	= \left(Q \circ s\right)_\sharp \mu 
	=  Q_\sharp \widetilde{\mu},\quad
	\widetilde{\mu} = s_\sharp \mu,
	\label{eq6_P_35ex}
	\end{equation}
for any $s \in \mathcal{S}$. As we mentioned in Section~\ref{subsec_symG}, the measure $\widetilde{\mu} = s_\sharp \mu$ is homogeneous. Hence, $\widetilde{\mu}$ satisfies conditions of Theorem~\ref{theorem_G}, which asserts that the measure $s_\sharp \zeta = Q_\sharp \widetilde{\mu}$ is invariant for the normalized flow $\Omega^t$. 

It remains to prove the commutation relations. 
Using (\ref{eq6_3_2}), commutativity properties of Proposition~\ref{prop_Gsym}  and Tab.~\ref{tab1}, and relation (\ref{eq6_P_37stA}), yields
	\begin{equation}
	\begin{array}{rcl}
	s^{\mathbf{Q}}_{\mathrm{r}} \circ \Omega^t(z) 
	& = & s^{\mathbf{Q}}_{\mathrm{r}} \circ Q \circ s_{\mathrm{s}}^{\mathbf{R}^t(z)} \circ \Phi^t(z)
	= Q \circ \Phi^t \circ s^{\mathbf{Q}}_{\mathrm{r}} \circ s_{\mathrm{s}}^{\mathbf{R}^t(z)}(z)
	\\[7pt]
	& = & 
	Q \circ \Phi^t \circ s_{\mathrm{s}}^{\mathbf{R}^t \circ 
	s^{\mathbf{Q}}_{\mathrm{r}}(z)} \circ s^{\mathbf{Q}}_{\mathrm{r}}(z)
	= Q \circ s_{\mathrm{s}}^{\mathbf{R}^t \circ 
	s^{\mathbf{Q}}_{\mathrm{r}}(z)} \circ \Phi^t \circ s^{\mathbf{Q}}_{\mathrm{r}}(z)
	= \Omega^t \circ s^{\mathbf{Q}}_{\mathrm{r}}(z),
	\end{array}
	\label{eq6_P_38}
	\end{equation}
proving the first relation in (\ref{eq6_4_S5}). Similarly, one can prove the other two relations in (\ref{eq6_4_S5}). 
\end{proof}

\section{Combining temporal scalings with Galilean transformations}
\label{sec_fus}

In this section, we generalize the quotient construction to the group 
	\begin{equation}
	\label{eq6.5.0}
	\mathcal{H} = 
	\{s^{a}_{\mathrm{ts}} \circ s^{\mathbf{v}}_{\mathrm{g}}:
	\ a > 0,\ \mathbf{v} \in \mathbb{R}^d\}, 
	\end{equation}
which includes both temporal scalings and Galilean transformations. These are all elements of our spatiotemporal symmetry group that do not commute with the flow; see Tab.~\ref{tab1} in Section~\ref{secEuler}. Now, the equivalence relation is considered with respect to any map from $\mathcal{H}$ as
	\begin{equation}
	x \sim x' \quad \textrm{if} \quad x' = h(x),\ h \in \mathcal{H},
	\label{eq6_5_0N}
	\end{equation}
and the representative set is introduced by selecting a unique element in each equivalence class. 

\subsection{Two-step quotient construction}

We perform our construction in two steps, by utilizing the quotient constructions introduced in Sections \ref{sec4} and \ref{sec_Gal} separately for temporal scalings and Galilean transformations. In the first step, we consider the equivalence with respect to Galilean transformations only, by following the theory of Section \ref{sec_Gal}. Assuming that the system has the properties of homogeneity, periodicity and incompressibility, Theorems~\ref{theorem_G} and \ref{theorem_G2} provide the representative set $\mathcal{Z}$, the flow $\Omega^t$, the invariant measure $\zeta$ and the symmetry group $\mathcal{S}$. One can see that this system possesses the properties required in Section~\ref{sec4}, considering $\mathcal{Z}$, $\Omega^t$ and $\zeta$ in place of $\mathcal{X}$, $\Phi^t$ and $\mu$. Recall that the results of Section~\ref{sec4} are based on composition and commutation relations (\ref{eq2_0ha})--(\ref{eq2_2}). Thus, we have

\begin{corollary}
Consider the group (\ref{eq6_G1}) as the direct sum $\mathcal{S} = \mathcal{H}_{\mathrm{ts}}+\mathcal{G}$ with
	\begin{equation}
	\mathcal{H}_{\mathrm{ts}} = \big\{s^{a}_{\mathrm{ts}}:\ a > 0\big\},\quad
	\mathcal{G} = \big\{
	s^{\mathbf{Q}}_{\mathrm{r}} \circ 
	s^{b}_{\mathrm{ss}}:\ 
	\mathbf{Q} \in \mathrm{O}(d),\,
	b > 0\big\}.
	\label{eq7_1}
	\end{equation}
By Theorem~\ref{theorem_G2}, relations (\ref{eq2_0ha})--(\ref{eq2_2}) are satisfied for $h^a = s^{a}_{\mathrm{ts}} \in \mathcal{H}_{\mathrm{ts}}$, $g \in \mathcal{G}$ and $\Omega^t$ in place of $\Phi^t$.
\end{corollary}

In the second step, Theorems~\ref{theorem0} and \ref{prop_g} of Section~\ref{sec4} with $\mathcal{X}$, $\mu$ and $\Phi^t$ replaced by $\mathcal{Z}$, $\zeta$ and $\Omega^t$ provide the ``second-generation'' normalized system defined on the representative set $\mathcal{Y}$ with the flow $\Psi^t$ and the invariant measure $\nu$. In this final system, each element $y \in \mathcal{Y}$ represents an equivalence class with respect to the full group (\ref{eq6.5.0}), and elements of the symmetry group $g \in \mathcal{G}$ define transformations $\nu \mapsto g_\star\nu$ preserving invariance of normalized measures. 

\subsection{The final normalized system}
\label{subsec_full}

We now describe all components in our final construction explicitly. The two-step normalization of the system yields the nested representative sets $\mathcal{Y} \subset \mathcal{Z} \subset \mathcal{X}$ with the two projectors
	\begin{equation}
	\label{eq7_2}
	\mathcal{X} \xmapsto{Q} \mathcal{Z} \xmapsto{P} \mathcal{Y}.
	\end{equation}
The projectors $Q:\mathcal{X} \mapsto \mathcal{Z}$ and $P:\mathcal{Z} \mapsto \mathcal{Y}$ are defined by (\ref{eq6_3_1}) and (\ref{eq2_2P}) as
	\begin{eqnarray}
	\label{eq7_3a}
	z &=& Q(x) = s_{\mathrm{g}}^{\mathbf{V}(x)}(x) \in \mathcal{Z}, 
	\\[2pt]
	\label{eq7_3b}
	y &=& P(z) = s_{\mathrm{ts}}^{A(z)}(z) \in \mathcal{Y}. 
	\end{eqnarray}
Here $\mathbf{V}: \mathcal{X} \mapsto \mathbb{R}^d$ and $A: \mathcal{Z} \mapsto \mathbb{R}_+$ are  measurable functions, which satisfy conditions
	\begin{equation}
	\begin{array}{c}
	\displaystyle
	A \circ s_{\mathrm{ts}}^a(z) = \frac{A(z)}{a},\quad 
	\mathbf{V} \circ s_{\mathrm{g}}^{\mathbf{v}}(x) 
	= \mathbf{V}(x)-\mathbf{v}, 
	\\[7pt]
	\displaystyle
	\mathbf{V} \circ s^{\mathbf{Q}}_{\mathrm{r}}(x) = \mathbf{Q}^{-1}\mathbf{V}(x),\quad
	\mathbf{V} \circ s^a_{\mathrm{ts}}(x) = \frac{\mathbf{V}(x)}{a},\quad
	\mathbf{V} \circ s^b_{\mathrm{ss}}(x) = b\mathbf{V}(x),
	\end{array}
	\label{eq7_4}
	\end{equation}
see (\ref{eq2_2Ap}), (\ref{eq6_2_1}) and (\ref{eq6_4_S1}). The final normalized flow $\Psi^\tau:\mathcal{Y} \mapsto \mathcal{Y}$ is given by combining (\ref{eq2_AmF}) and (\ref{eq6_3_2}) as
	\begin{equation}
	\Psi^\tau = P\circ \Omega_A^\tau,
	\label{eq7_5}
	\end{equation}
where $\Omega_A^\tau$ denotes a change of time (\ref{eq2_A1}) in the flow  
	\begin{equation}
	\Omega^t: \mathcal{Z} \mapsto \mathcal{Z}, \quad
	\Omega^t(z) 
	= Q \circ s_{\mathrm{s}}^{\mathbf{R}^t(z)} \circ \Phi^t(z), \quad
	z \in \mathcal{Z}.
	\label{eq7_6}
	\end{equation}
Here the function $\mathbf{R}^t(z)$ is defined by equations (\ref{eq6_2A_1}) and (\ref{eq6_2_2}). 

The invariant measure $\nu$ of the flow $\Psi^\tau$ is given by combining (\ref{eq2_Am}) and (\ref{eq6_3_3}) as
	\begin{equation}
	\nu = P_\sharp \zeta_A =  (P \circ Q)_\sharp \mu_{A \circ Q},
	\quad \zeta = Q_\sharp \mu,
	\label{eq7_7}
	\end{equation}
where the subscripts $A$ and $A \circ Q$ denote the change-of-time transformations (\ref{eq2_A1ac}); in the second equality we used relation (\ref{eq2_D2exD}) of Lemma~\ref{lemmaMs}. By Theorem~\ref{prop_g}, elements of the symmetry group $g \in \mathcal{G}$ define the transformations 
	\begin{equation}
	g_\star \nu = (P \circ {g})_\sharp\nu_C, \quad C = A \circ g,
	\label{eq7_8}
	\end{equation}
providing invariant measures $g_\star \nu$ for the same flow $\Psi^\tau$. Let us summarize these findings as

\begin{theorem}
\label{th_main}
Given measurable function $\mathbf{V}: \mathcal{X} \mapsto \mathbb{R}^d$ and $A: \mathcal{Z} \mapsto \mathbb{R}_+$ satisfying conditions (\ref{eq7_4}) and assuming the properties of homogeneity, periodicity and incompressibility (see Section~\ref{sec_Gal}), expressions (\ref{eq7_5}) and (\ref{eq7_6}) define the normalized flow $\Psi^\tau$ in the representative set $\mathcal{Y}$ with the invariant measure (\ref{eq7_7}). Group elements $g \in \mathcal{G}$ define statistical symmetries in the normalized system: they generate invariant measures by means of transformation (\ref{eq7_8}).
\end{theorem}

Combining Propositions~\ref{prop_h} and \ref{prop_G1} we see that normalized measures are not sensitive to both temporal scalings and Galilean transformations. 

\begin{corollary} \label{corollary_H}
All invariant measures $\widetilde{\mu} = h_\sharp \mu$ with $h \in \mathcal{H}$ yield the same normalized measure $\nu$ by expression (\ref{eq7_7}).
\end{corollary}

Also, combining Corollaries~\ref{corr_sym_mes} and \ref{corr_sym_mes2} yields the following symmetry relation between the original and normalized systems.

\begin{corollary}\label{corr_sym_mes_main}
If the measure $\mu$ is symmetric with respect to a composition $g \circ h$ for some $g \in \mathcal{G}$ and $h \in \mathcal{H}$, then the normalized measure $\nu$ is symmetric with respect to $g$:
	\begin{equation}
	\label{eq2_A8ex_M}
	(g \circ h)_\sharp \mu = \mu \quad \Rightarrow \quad g_\star \nu = \nu.
	\end{equation}
\end{corollary}

The representative set $\mathcal{Y}$ defines a configuration space for our normalized system, which depends on a choice of the functions $\mathbf{V}(x)$ and $A(z)$ in (\ref{eq7_3a}) and (\ref{eq7_3b}).
The following statement extends Theorem~\ref{theorem3} from Section~\ref{subsec_SNM} ensuring that symmetry relations are not sensitive to a choice of $\mathcal{Y}$ (for the proof see Section~\ref{proof_th7}). 

\begin{theorem}
\label{theorem_indep}
If the normalized invariant measure $\nu$ in (\ref{eq7_7}) is symmetric with respect to $g \in \mathcal{G}$, i.e. $g_\star \nu = \nu$, for some representative set $\mathcal{Y}$, then the same is true for any representative set.
\end{theorem}

Since the group $\mathcal{G}$ in (\ref{eq7_1}) contains spatial scalings, the theory of Section~\ref{sec_int} is applicable for the study of structure functions and intermittency in our final normalized system. 
The central feature of this theory is  the notion of hidden symmetry: the normalized measure can be symmetric despite the original measure is not, i.e., $g_\star \nu = \nu$ while $(g \circ h)_\sharp \mu \ne \mu$ for any $h \in \mathcal{H}$. As we discuss in the next section, including Galilean transformations into the full quotient construction is crucial for applications in fluid dynamics. 

Similarly to Proposition \ref{prop_erg} in Section \ref{subsec_red}, we can express statistical averages in (\ref{eq2_Avr2}) for any measurable test function $\psi(y)$ in the normalized system through analogous averages (\ref{eq2_Avr1})  in the original system. Using (\ref{eq7_7}), the derivation analogous to (\ref{eq2_prP3_B1}) yields
	\begin{equation}
	\label{eq7_9}
	\langle \psi \rangle_\nu 
	= \frac{\langle \varphi \rangle_\mu}{\langle A\circ Q \rangle_\mu}, 
	\quad \varphi(x) = \psi\circ P \circ Q(x) \, A\circ Q(x).
	\end{equation}
This identity relates ensemble averages like in the second equality of (\ref{eq2_Avr3}). We cannot generalize the first equality in (\ref{eq2_Avr3}), which relates temporal averages for particular solutions. Technically, this is because spatial translations are not among symmetries in our final normalized system. However, the relation between temporal averages may follow from (\ref{eq7_9}) assuming the ergodicity~\cite{cornfeld2012ergodic}.

\subsection{Application to the Euler system}
\label{subsec_NST}

Here present the quotient construction applied to the Euler system (\ref{eqE3b}) that we started with in Section~\ref{secEuler}. In this system velocity fields $x = \mathbf{u}(\mathbf{r})$ are considered as elements of a configuration space $\mathcal{X}$. We proceed formally by assuming the existence of a flow (evolution) operator $\Phi^t:\mathcal{X} \mapsto \mathcal{X}$, which satisfies the commutation relations of Tab.~\ref{tab1} with symmetry maps (\ref{eq3}). 

The full quotient construction of Section~\ref{subsec_full} requires two projectors (\ref{eq7_2}), and we denote the respective velocity fields as
	\begin{equation}
	\label{eq_D_PV}
 	z = Q(x) = \widetilde{\mathbf{u}}(\mathbf{r}), \quad 
	y = P(z)= \mathbf{U}(\mathbf{r}). 
 	\end{equation}
As suggested earlier in Section~\ref{subsec_incomp}, we define $\mathbf{V}(x) \in \mathbb{R}^d$ as the velocity vector at the origin:
	\begin{equation}
	\label{eq_D_V}
	\mathbf{V}(x) = \mathbf{u}(\mathbf{0}).
	\end{equation}
Then, relations (\ref{eq7_3a}) and (\ref{eq7_3b}) with symmetries (\ref{eq3}) yield
	\begin{equation}
	\widetilde{\mathbf{u}}(\mathbf{r})
	= \mathbf{u}(\mathbf{r})-\mathbf{u}(\mathbf{0}), \quad
	\mathbf{U}(\mathbf{r}) = \frac{\mathbf{u}(\mathbf{r})-\mathbf{u}(\mathbf{0})}{A(z)}.
	\label{eq_D_1}
	\end{equation}
One can check that conditions (\ref{eq7_4}) are satisfied for $\mathbf{V}(x)$ from (\ref{eq_D_V}) and $A(z)$ being a positive-homogeneous function of degree $1$, i.e. $A(\alpha z) = \alpha A(z)$ for $\alpha > 0$. For example, one can take
	\begin{equation}
	A(z)
	= \left(\int K(r)\,\| \widetilde{\mathbf{u}}(\mathbf{r}) \|^2 \, d\mathbf{r}\right)^{1/2},
	\label{eq_D_2}
	\end{equation}
where $K: \mathbb{R}_+ \mapsto \mathbb{R}_+$ is some positive function vanishing (or decaying rapidly) at large $r = \|\mathbf{r}\|$. {\color{black}Definitions based on vorticity in place of velocity can also be used in (\ref{eq_D_2}), which would resemble the definition (\ref{eq3_B1}) for the shell model.}

The Euler system (\ref{eqE3b}) already includes the incompressibility condition. Under extra assumptions of homogeneity and periodicity, Theorem~\ref{th_main} yields the normalized flow $\Psi^\tau$. This flow has the invariant measure $\nu$, which describes the probability distribution of normalized velocity fields $y = \mathbf{U}(\mathbf{r})$ and is given explicitly in terms of the original  distribution $\mu$ for $x = \mathbf{u}(\mathbf{r})$. Also, symmetries of the group $\mathcal{G}$ (rotations and spatial scalings) extend to the normalized system in the statistical sense: they generate invariant measures $g_\star \nu$ for the flow $\Psi^\tau$. In applications, invariant probability measures are usually accessed with the ergodicity hypothesis, which allows substituting averages with respect to a measure by averages with respect to time. For this purpose, it is useful to have explicit relations between solutions for different flows. We devote the rest of this subsection to this issue.

Let $\Phi^t(x) = \mathbf{u}(\mathbf{r},t)$ be the velocity field describing a solution with the initial condition $x = \mathbf{u}(\mathbf{r})$ at $t = 0$. Similarly, we denote by $\Omega^t(z) = \widetilde{\mathbf{u}}(\mathbf{r},t)$ the velocity field generated by the flow (\ref{eq7_6}). As we explained in Sections~\ref{subsec_incomp} and \ref{subsec_NFG}, the field $\widetilde{\mathbf{u}}(\mathbf{r},t)$ is obtained by following the original system in the reference frame moving along a Lagrangian trajectory $\mathbf{r} = \mathbf{R}^t(x)$. Thus, we have
	\begin{equation}
	\widetilde{\mathbf{u}}(\mathbf{r},t) 
	= \mathbf{u}\left(\mathbf{R}^t+\mathbf{r},t\right)-\mathbf{u}\left(\mathbf{R}^t,t\right), 
	\label{eq_D_3_A1}
	\end{equation}
where $\mathbf{R}^t$ is defined by equations (\ref{eq6_2A_1}) and (\ref{eq6_2_5}) as
	\begin{equation}
	\frac{d\mathbf{R}^t}{dt} =  \mathbf{u}(\mathbf{R}^t,t),\quad 
	\mathbf{R}^0 = \mathbf{0}.
	\label{eq_D_3_A1R}
	\end{equation}
One can also derive these relations directly from expressions (\ref{eq7_6}) and (\ref{eq_D_1}) with the help of identity (\ref{eq6_P_1b}) and commutation relations of Tab.~\ref{tab1} for symmetries (\ref{eq3}).  The final velocity field $\Psi^\tau(y) = \mathbf{U}(\mathbf{r},\tau)$ of the normalized flow (\ref{eq7_5}) is obtained as 
	\begin{equation}
	\mathbf{U}(\mathbf{r},\tau) 
	= \frac{\widetilde{\mathbf{u}}(\mathbf{r},t)}{a_z(t)},
	\quad
	\tau = \int_0^t a_z(s)\,ds, \quad
	a_z(t) = A \circ \Omega^t(z),
	\label{eq_D_3}
	\end{equation}
where the first relation is given by the projector $P$ and the second relation introduces the change of time; see (\ref{eq2_A1}).
Solution (\ref{eq_D_3}) has the physical meaning of the velocity field, which is considered in a reference frame moving along a Lagrangian trajectory, and having the temporal scale adjusted dynamically by the function $a_z(t) $. 

Finally, let us describe solutions obtained by the spatial scaling from (\ref{eq3}). In order to comply with our notations in Sections~\ref{sec_shell1} and \ref{sec_int} we introduce the map
	\begin{equation}
	g = s^2_{\mathrm{ss}},\quad
	g: \mathbf{u}(\mathbf{r}) \mapsto 2\,\mathbf{u}\left(\frac{\mathbf{r}}{2}\right),
	\label{eq_D_S1}
	\end{equation}
which decreases the spatial scale by the factor of two. Then, $g^m = s^b_{\mathrm{ss}}$ with $b = 2^m$.
According to Tab.~\ref{tab1} and Theorem~\ref{theorem_G2}, $g$ commutes with both $\Phi^t$ and $\Omega^t$. Hence, the scaled solution $\Phi^t \circ g^m(x) = \mathbf{u}^{(m)}(\mathbf{r},t)$ for the original system is expressed as
	\begin{equation}
	 \mathbf{u}^{(m)}(\mathbf{r},t) = b\,\mathbf{u}\left(\frac{\mathbf{r}}{b},t\right), \quad 
	 b = 2^m.
	\label{eq_D_S2}
	\end{equation}
Using (\ref{eq_D_3_A1}), we obtain the corresponding scaled solution $\Omega^t \circ g^m(z) = \widetilde{\mathbf{u}}^{(m)}(\mathbf{r},t)$ in $\mathcal{Z}$ as
	\begin{equation}
	\widetilde{\mathbf{u}}^{(m)}(\mathbf{r},t) 
	= b\,\widetilde{\mathbf{u}}\left(\frac{\mathbf{r}}{b},t\right)
	= b\,\mathbf{u}\left(\mathbf{R}^t+\frac{\mathbf{r}}{b},t\right)
	-b\,\mathbf{u}\left(\mathbf{R}^t,t\right).
	\label{eq_D_S3}
	\end{equation}
The scaled solution $\mathbf{U}^{(m)}(\mathbf{r},\tau^{(m)})$ of the flow $\Psi^\tau$ is obtained from (\ref{eq_D_S3}) similarly to (\ref{eq_D_3}) as 
	\begin{equation}
	\mathbf{U}^{(m)}(\mathbf{r},\tau^{(m)}) 
	= \frac{\widetilde{\mathbf{u}}^{(m)}(\mathbf{r},t)}{a_z^{(m)}(t)},
	\quad
	\tau^{(m)} = \int_0^t a_z^{(m)}(s)\,ds,\quad
	a_z^{(m)}(t) = A \circ \Omega^t \circ g^m(z).
	\label{eq_D_S4}
	\end{equation}
This transformation features the change of both state and time.
{\color{black}Similarly to the shell model in Section~\ref{subsec_shell_norm}, one can show that $\mathbf{U}^{(m)}(\mathbf{r},\tau^{(m)})$ satisfies the same system of normalized Euler equations as $\mathbf{U}(\mathbf{r},\tau)$, therefore, demonstrating the hidden scaling symmetry. We refer to a subsequent paper~\cite{mailybaev2022hidden}, where such derivations were carried out explicitly.}

\subsection{Hidden symmetry in developed turbulence of the Navier--Stokes system} \label{subs_6_4}

Let us apply the developed formalism to the analysis of turbulence in the Navier--Stokes system. We will use similar arguments as in Section~\ref{subsec_conj} but now applied to the full normalized system. In the dimensionless form, the Navier--Stokes system is obtained from the Euler system by adding a forcing term $\mathbf{f}$ and a viscous term $\mathrm{Re}^{-1}\Delta\mathbf{u}$ on the right-hand side of the first equation in (\ref{eqE3b})~\cite{frisch1999turbulence}. The regime of developed turbulence corresponds to large $\mathrm{Re} \gg 1$ and features the so-called inertial interval of scales $\ell$ expressed by the Kolmogorov theory~\cite{kolmogorov1941local,frisch1999turbulence} as 
	\begin{equation}
	\label{eqNS_C1}
	\mathrm{Re}^{-3/4} \ll \ell \ll 1.
	\end{equation}
Here $\mathrm{Re}^{-3/4}$ represents the so-called Kolmogorov viscous scale and $\ell \sim 1$ corresponds to the scales at which external forces are applied. In the inertial interval (\ref{eqNS_C1}), the flow is described asymptotically by the Euler system. We already discussed a similar dynamics for the shell model in Section~\ref{subsec_conj} (see Fig.~\ref{figHS}), where the scale is defined as $\ell \sim 1/k_n$. Similarly to the shell model, we now apply the spatial scaling $g^m$ from (\ref{eq_D_S1}), which modifies the inertial interval (\ref{eqNS_C1}) as
	\begin{equation}
	\label{eqNS_C2}
	\mathrm{Re}^{-3/4} \ll \frac{\ell}{b} \ll 1, \quad b = 2^m.
	\end{equation}
This interval extends to all scales $\ell > 0$ by considering the double limit: first taking $\mathrm{Re}\to \infty$ and then $m \to \infty$. One expects that the limiting dynamics (if it exists) is governed by the Euler system.

It is known that the scaling invariance is broken in the developed hydrodynamic turbulence due to the intermittency phenomenon~\cite{frisch1999turbulence,falkovich2009symmetries}, which precludes the convergence of the double limit for the turbulent statistics. As we have shown in this paper (see Sections~\ref{subsec_conj} and \ref{sec_intS}), the intermittency is not an obstacle for a similar convergence in the normalized system. As in Section~\ref{subsec_conj}, we denote by $\mu^{\mathrm{Re}}$ the probability measure of the statistically stationary state in the Navier--Stokes system for a given Reynolds number, and by $\nu^{\mathrm{Re}}$ the corresponding normalized measure. Then the limiting normalized measure is defined as the double limit
	\begin{equation}
	\label{eqNS_T2}
	\nu^\infty = \lim_{m \to \infty} \lim_{{\mathrm{Re}} \to \infty} 
	g^m_\star\nu^{\mathrm{Re}}.
	\end{equation}
Existence of this limit implies that the limiting normalized measure is symmetric: $g_\star\nu^\infty = \nu^\infty$. 

{\color{black}Once the hidden symmetry is established, the theory of Section~\ref{sec_int} applies. This theory explains the power-law scaling for structure functions, and associates the scaling exponents to Perron--Frobenius eigenvalues defined in terms of $\nu^\infty$. In this theory, structure functions are written in the generalized form (\ref{eq3_17}). For example, consider the standard structure function $S_p(\ell) = \langle \|\delta_\ell \mathbf{u}\|^p \rangle$, where $\ell = 2^{-n}$ and $\delta_\ell \mathbf{u} = \mathbf{u}(\mathbf{r}')-\mathbf{u}(\mathbf{r})$ is a difference of fluid velocities at a distance $\ell  = \|\mathbf{r}'-\mathbf{r}\| > 0$. It is expressed in the form (\ref{eq3_17}) with $k_n = 1/\ell$, the function $F(x)$ defined as the average of $\|\mathbf{u}(\mathbf{r})-\mathbf{u}(\mathbf{0})\|^p$ at distances $\|\mathbf{r}\| = 1$, and the operator $g$ given by (\ref{eq_D_S1}) and describing the doubling of spatial resolution. One can write similar expressions for the longitudinal and transverse structure functions.} 

Current understanding of the Navier--Stokes system does not allow a rigorous study of the limit (\ref{eqNS_T2}); see e.g. \cite{bardos2013mathematics}. Nevertheless, the convergence can be verified numerically using expressions (\ref{eq_D_S4}) and (\ref{eq_D_S3}) with the ergodicity assumption. In this numerical analysis, the measure $g^m_\star\nu^{\mathrm{Re}}$ is approximated by the temporal statistics of the velocity field $\mathbf{U}^{(m)}(\mathbf{r},\tau^{(m)})$ obtained from a solution $\mathbf{u}(\mathbf{r},t)$ of the Navier--Stokes system for a large Reynolds number. Hence, the convergence in (\ref{eqNS_T2}) implies that this statistics is independent of $m$ at the scales of  inertial interval (\ref{eqNS_C2}). We emphasize that Galilean transformations play important role in this construction: they yield the Quasi--Lagrangian form of the velocity fields (\ref{eq_D_3_A1}) and (\ref{eq_D_S3}) considered in the reference frame moving with a Lagrangian particle. Indeed, subtracting the Lagrangian particle speed in (\ref{eq_D_S3}) eliminates the so-called sweeping effect~{\color{black}\cite{belinicher1987scale,l1991scale,biferale1999multi,frisch1999turbulence} (crucial for multi-time statistics and caused by large-scale motions), which otherwise would prevent the limit (\ref{eqNS_T2}). We refer the reader to the subsequent work~\cite{mailybaev2022hidden}, where the numerical verification of hidden symmetry for the Navier--Stokes system was carried out.

Finally, let us remark on the connection of hidden scaling symmetry} with the co-called \textit{Kolmogorov multipliers} defined as~\cite{kolmogorov1962refinement,chen2003kolmogorov}
	\begin{equation}
	w_{ij}(\mathbf{r};\ell,\ell') =
	\frac{\delta_i u_j(\mathbf{r},\ell)}{\delta_i u_j(\mathbf{r},\ell')},\quad
	\delta_i \mathbf{u}(\mathbf{r},\ell) = \mathbf{u}(\mathbf{r}+\ell\mathbf{e}_i)-\mathbf{u}(\mathbf{r}),
	\label{eq_D_5}
	\end{equation}
where the indices $i$ and $j$ denote vector components, $\mathbf{e}_i$ are unit vectors in $\mathbb{R}^3$, and $\ell,\ell' \in \mathbb{R}_+$ are two positive scales.
Using (\ref{eq_D_S4}) and (\ref{eq_D_S3}), the multipliers evaluated along the Lagrangian trajectory $\mathbf{r} = \mathbf{R}^t$ are expressed as
	\begin{equation}
	w_{ij}(\mathbf{r};\ell,\ell') = \frac{U_j^{(m)}(\mathbf{e}_i)}{U_j^{(m)}(\gamma\mathbf{e}_i)},\quad
	\ell = 2^{-m},\quad \gamma = \frac{\ell'}{\ell}.
	\label{eq_D_6}
	\end{equation}
{\color{black}It was first conjectured by Kolmogorov~\cite{kolmogorov1962refinement} and observed systematically both in numerical simulations and experimental data~\cite{chen2003kolmogorov} that  multipliers (\ref{eq_D_5}) have scale-invariant statistics depending only on the ratio $\gamma = \ell'/\ell$ and the vector indices. Using representation (\ref{eq_D_6}) this scale invariance becomes the direct consequence of the hidden scaling symmetry: the latter implies that the statistics does not depend on $m$. Strictly speaking, there are some reservations to this argument, because the statistics of normalized variables (\ref{eq_D_6}) is considered with respect to a different time $\tau^{(m)}$. However, this argument can be extended to the original time $t$ as shown in~\cite{mailybaev2022hidden}.} 

\subsection{Proof of Theorem~\ref{theorem_indep}}
\label{proof_th7}

Let us consider a different choice of the representative set denoted by tildes as
	\begin{equation}
	\label{eq7_3_1}
	\mathcal{X} \xmapsto{\widetilde{Q}} \widetilde{\mathcal{Z}} 
	\xmapsto{\widetilde{P}} \widetilde{\mathcal{Y}}.
	\end{equation}
This system is defined by two functions $\widetilde{\mathbf{V}}(x)$ and $\widetilde{A}(\widetilde{z})$ satisfying symmetry relations (\ref{eq7_4}). Recall that the independence of condition $g_\star \nu = \nu$ to a choice of $\widetilde{A}(\widetilde{z})$ has already been proven in Theorem~\ref{theorem3} of Section~\ref{subsec_SNM}. Thus, {\color{black}we can assume a specific form of this function as}
	\begin{equation}
	\label{eq7_3_2}
	\widetilde{A} = A \circ Q.
	\end{equation}
The first condition in (\ref{eq7_4}) is verified for (\ref{eq7_3_2}) as
	\begin{equation}
	\label{eq7_3_3}
	\begin{array}{rcl}
	\widetilde{A} \circ  s_{\mathrm{ts}}^a(z) 
	&=&\displaystyle
	A \circ Q \circ s_{\mathrm{ts}}^a(z)
	= A \circ s_{\mathrm{g}}^{\mathbf{V}\circ s_{\mathrm{ts}}^a(z)} \circ s_{\mathrm{ts}}^a(z)
	= A \circ s_{\mathrm{g}}^{\mathbf{V}(z)/a} \circ s_{\mathrm{ts}}^a(z) 
	\\[7pt]
	&=& \displaystyle
	A \circ s_{\mathrm{ts}}^a(z) \circ s_{\mathrm{g}}^{\mathbf{V}(z)}(z)
	= \frac{A \circ s_{\mathrm{g}}^{\mathbf{V}(z)}(z)}{a} 
	= \frac{A \circ Q(z)}{a} 
	= \frac{\widetilde{A}(z)}{a} ,
	\end{array}
	\end{equation}
where we consecutively used (\ref{eq7_3_2}), (\ref{eq7_3a}), the fourth equality in (\ref{eq7_4}), commutation relation from Tab.~\ref{tab1}, the first equality in (\ref{eq7_4}), (\ref{eq7_3a}) and (\ref{eq7_3_2}).

\begin{lemma}
For the choice (\ref{eq7_3_2}), the following relations hold: 
	\begin{equation}
	Q \circ \widetilde{Q} = Q, \quad 
	Q \circ \widetilde{P} \circ \widetilde{Q} = P \circ Q,\quad
	\widetilde{A} \circ \widetilde{Q} = A \circ Q.
	\label{eq7_3_4r}
	\end{equation}
\end{lemma}
\begin{proof}
The first relation is obtained using (\ref{eq7_3a}) and the second equality in (\ref{eq7_4}) as
	\begin{equation}
	Q \circ \widetilde{Q}(x) 
	= s_{\mathrm{g}}^{\mathbf{V} \circ s_{\mathrm{g}}^{\widetilde{\mathbf{V}}(x)}(x)}
	\circ s_{\mathrm{g}}^{\widetilde{\mathbf{V}}(x)}(x)
	= s_{\mathrm{g}}^{\mathbf{V}(x)-\widetilde{\mathbf{V}}(x)}
	\circ s_{\mathrm{g}}^{\widetilde{\mathbf{V}}(x)}(x)
	= s_{\mathrm{g}}^{\mathbf{V}(x)}(x) = Q(x).
	\label{eq7_3_4a}
	\end{equation}
Denoting 
	\begin{equation}
	\label{eq7_3_4nt}
	\widetilde{z} = \widetilde{Q}(x), \quad 
	\widetilde{a} = \widetilde{A}(\widetilde{z}),
	\end{equation}
we derive
	\begin{equation}
	\begin{array}{rcl}
	Q \circ \widetilde{P} \circ \widetilde{Q}(x) 
	&=& 
	Q \circ \widetilde{P}(\widetilde{z})
	= s_{\mathrm{g}}^{\mathbf{V} \circ \widetilde{P}(\widetilde{z})} \circ \widetilde{P}(\widetilde{z})
	= s_{\mathrm{g}}^{\mathbf{V} \circ s_{\mathrm{ts}}^{\widetilde{a}}(\widetilde{z})}
	\circ s_{\mathrm{ts}}^{\widetilde{a}}(\widetilde{z})
	\\[5pt]
	&=& s_{\mathrm{g}}^{\mathbf{V}(\widetilde{z})/\widetilde{a}}
	\circ s_{\mathrm{ts}}^{\widetilde{a}}(\widetilde{z})
	= s_{\mathrm{ts}}^{\widetilde{a}} \circ s_{\mathrm{g}}^{\mathbf{V}(\widetilde{z})}(\widetilde{z})
	= s_{\mathrm{ts}}^{\widetilde{a}} \circ Q(\widetilde{z}).
	\end{array}
	\label{eq7_3_4b}
	\end{equation}
where we consecutively used (\ref{eq7_3a}), (\ref{eq7_3b}), the fourth equality in (\ref{eq7_4}), commutation relation from Tab.~\ref{tab1}, and again (\ref{eq7_3a}). Using (\ref{eq7_3_4nt}), (\ref{eq7_3_4a}) and (\ref{eq7_3_2}), we obtain
	\begin{equation}
	\label{eq7_3_4fa}
	Q(\widetilde{z}) = Q \circ \widetilde{Q}(x) = Q(x),
	\quad
	\widetilde{a} = \widetilde{A}(\widetilde{z}) = A \circ Q \circ \widetilde{Q}(x) = A \circ Q(x).
	\end{equation}
Substituting these relations into the right-hand side of (\ref{eq7_3_4b}) and using (\ref{eq7_3b}) yields
	\begin{equation}
	Q \circ \widetilde{P} \circ \widetilde{Q}(x) 
	= s_{\mathrm{ts}}^{A \circ Q(x)} \circ Q(x)
	= P \circ Q(x).
	\label{eq7_3_4fn}
	\end{equation}
Using the first equality of (\ref{eq7_3_4r}) and (\ref{eq7_3_2}), we have
	\begin{equation}
	\widetilde{A} \circ \widetilde{Q} = A \circ Q \circ \widetilde{Q} = A \circ Q.
	\label{eq7_3_5}
	\end{equation} 
\end{proof}

Let us write invariant measures (\ref{eq7_7}) for the two normalized systems as
	\begin{equation}
	\nu = (P \circ Q)_\sharp \mu_{A \circ Q}, \quad
	\widetilde{\nu} = (\widetilde{P} \circ \widetilde{Q})_\sharp \mu_{\widetilde{A} \circ \widetilde{Q}}.
	\label{eq7_3_4}
	\end{equation}
Substituting (\ref{eq7_3_4}) into (\ref{eq7_8}), we expresses the invariant measure $g_\star \nu$ as
	\begin{equation}
	g_\star \nu 
	=  (P \circ {g})_\sharp\nu_C
	=  (P \circ {g})_\sharp \Big((P \circ Q)_\sharp \mu_{A \circ Q}\Big)_{A \circ g}.
	\end{equation}
We modify this expression using relation (\ref{eq2_C2g}) as
	\begin{equation}
	g_\star \nu 
	=  (P \circ {g} \circ P \circ Q)_\sharp \mu_F,
	\label{eq7_3_7}
	\end{equation}
where 
	\begin{equation}
	F = (A \circ g \circ P \circ Q)\, A\circ Q.
	\label{eq7_3_7b}
	\end{equation}
Using (\ref{eq2_B2x}) and commutativity of $g$ with $Q$ (see Proposition~\ref{prop_Gsym} from Section~\ref{subsec_sym_G}), we further reduce (\ref{eq7_3_7}) to the form
	\begin{equation}
	g_\star \nu 
	=  (P \circ {g} \circ Q)_\sharp \mu_F
	=  (P \circ Q \circ {g})_\sharp \mu_F.
	\label{eq7_3_7z}
	\end{equation}
	
The similar derivation for the measure $g_\star \widetilde{\nu}$ in the other normalized system yields
	\begin{equation}
	g_\star \widetilde{\nu} 
	=  (\widetilde{P} \circ \widetilde{Q} \circ {g})_\sharp \mu_{\widetilde{F}},\quad
	\widetilde{F} = (\widetilde{A} \circ g \circ \widetilde{P} \circ \widetilde{Q})\, 
	\widetilde{A}\circ \widetilde{Q}.
	\label{eq7_3_8}
	\end{equation}
For this function $\widetilde{F}$, we have
	\begin{equation}
	\widetilde{F} 
	= (A \circ Q \circ g \circ \widetilde{P} \circ \widetilde{Q})\,(A\circ Q)
	= (A \circ g \circ Q \circ \widetilde{P} \circ \widetilde{Q})\,(A\circ Q)
	= (A \circ g \circ P \circ Q)\,(A\circ Q) = F,
	\label{eq7_3_9}
	\end{equation}
where we used (\ref{eq7_3_2}), (\ref{eq7_3_4r}), commutativity of $g$ with $Q$ by Proposition~\ref{prop_Gsym}, and (\ref{eq7_3_7b}).

Let us assume that $g_\star \widetilde{\nu} = \widetilde{\nu}$. This equality can be written  using (\ref{eq7_3_8}), (\ref{eq7_3_9}) and (\ref{eq7_3_4}) as
	\begin{equation}
	(\widetilde{P} \circ \widetilde{Q} \circ g)_\sharp \mu_F 
	= (\widetilde{P} \circ \widetilde{Q})_\sharp \mu_{\widetilde{A} \circ \widetilde{Q}}.
	\label{eq7_3_10}
	\end{equation}
Applying the push-forward $Q_\sharp$ to both sides of this equality and using (\ref{eq7_3_4r}) yields
	\begin{equation}
	(P \circ Q \circ g)_\sharp \mu_F 
	= (P \circ Q)_\sharp \mu_{A \circ Q}.
	\label{eq7_3_11}
	\end{equation}
According to (\ref{eq7_3_7z}) and (\ref{eq7_3_4}), this yields the symmetry property $g_\star \nu = \nu$, proving that this property does not depend on a choice of the representative set.

\section{Conclusion}

In this work we studied symmetries given by a sum of spatiotemporal scaling and Galilean groups, which are represented by maps in an infinite-dimensional configuration space. Here we understand symmetries in the statistical sense, i.e., as maps that preserve the invariance of probability measures with respect to a flow (evolution operator). We focused on the equivalence relation imposed by the two symmetries, which do not commute with the flow: temporal scalings and Galilean transformations. Equivalence classes with respect to these symmetries define the so-called quotient space. In the noncommutative case, the equivalence relation is not preserved by the flow and, therefore, the dynamics cannot be extended to the quotient space in general. 

In this paper, we have shown that, despite of noncommutativity, a quotient-like construction is possible due to the specific form of commutation relations. This yields a normalized system with a corresponding flow, invariant measure and symmetries, which are restricted to a representative set containing a single element within each equivalence class. Such normalized flow and invariant measure are explicitly related to the flow and  invariant measure of the original system. In this construction, temporal scalings induce a state-dependent time synchronization. The role of Galilean transformations is two-fold: they impose extra conditions of homogeneity and incompressibility, and they induce the normalized system resembling the Quasi--Lagrangian representation of fluid dynamics. 

Our construction leads to the notion of a \textit{hidden symmetry}: this is a symmetry, which is broken in the original system but restored in the normalized system. As an application, we show that the hidden symmetry implies asymptotic power law scaling for structure functions in the theory of turbulence, with the exponents expressed as Perron--Frobenius eigenvalues. These exponents can be anomalous, i.e., depending nonlinearly on the order of a structure function. {\color{black}This theory is verified both numerically and analytically for anomalous scaling exponents in shell models of turbulence~\cite{mailybaev2021solvable,mailybaev2022shell}.} 

Finally, we formulated the quotient construction and the concept of hidden scaling symmetry for the incompressible Euler and Navier--Stokes systems. In this paper, we mostly focused on general properties of this construction, which have potential applications in the theory of turbulence. {\color{black}For the detailed formalism of hidden symmetry applied to the Navier--Stokes equations see~\cite{mailybaev2022hidden}.}

There are several aspects to address in future developments. The hidden scaling symmetry can be verified by numerical methods in specific systems and, if confirmed, used in theoretical and applied studies. Important theoretical questions (not discussed in the present work) are related to the role of inviscid invariants and the form of dissipative mechanism in the normalized system. {\color{black} We revealed a peculiar role of incompressibility, which was obtained as a necessary condition from the analysis of symmetry groups alone. Thus,} the proposed quotient construction does not apply to the compressible fluid dynamics, e.g. the turbulence in Burgers equation~\cite{frisch2001burgulence}. 

\vspace{5mm}\noindent
\textbf{Acknowledgments.} The author is grateful to Gregory L. Eyink and Simon Thalabard for fruitful discussions and comments on the manuscript. The work was supported by CNPq grants 308721/2021-7 and FAPERJ grant E-26/201.054/2022. 

\bibliographystyle{plain}
\bibliography{refs}

\begin{thebibliography}{10}

\bibitem{bardos2013mathematics}
C.~W. Bardos and E.~S. Titi.
\newblock Mathematics and turbulence: where do we stand?
\newblock {\em Journal of Turbulence}, 14(3):42--76, 2013.

\bibitem{belinicher1987scale}
V.~I. Belinicher and V.~S. L’vov.
\newblock A scale-invariant theory of fully developed hydrodynamic turbulence.
\newblock {\em Soviet Physics -- JETP}, 66(2):303--313, 1987.

\bibitem{benzi1993intermittency}
R.~Benzi, L.~Biferale, and G.~Parisi.
\newblock On intermittency in a cascade model for turbulence.
\newblock {\em Physica D}, 65(1-2):163--171, 1993.

\bibitem{biferale2003shell}
L.~Biferale.
\newblock Shell models of energy cascade in turbulence.
\newblock {\em Annual Review of Fluid Mechanics}, 35:441--468, 2003.

\bibitem{biferale1999multi}
L.~Biferale, G.~Boffetta, A.~Celani, and F.~Toschi.
\newblock Multi-time, multi-scale correlation functions in turbulence and in
  turbulent models.
\newblock {\em Physica D}, 127(3-4):187--197, 1999.

\bibitem{biferale2018rayleigh}
L.~Biferale, G.~Boffetta, A.~A. Mailybaev, and A.~Scagliarini.
\newblock {Rayleigh-Taylor turbulence with singular nonuniform initial
  conditions}.
\newblock {\em Physical Review Fluids}, 3(9):092601(R), 2018.

\bibitem{biferale2017optimal}
L.~Biferale, A.~A. Mailybaev, and G.~Parisi.
\newblock Optimal subgrid scheme for shell models of turbulence.
\newblock {\em Physical Review E}, 95(4):043108, 2017.

\bibitem{brading2003symmetries}
K.~Brading and E.~Castellani.
\newblock {\em Symmetries in physics: philosophical reflections}.
\newblock Cambridge University Press, 2003.

\bibitem{cardy1996scaling}
J.~Cardy.
\newblock {\em Scaling and renormalization in statistical physics}.
\newblock Cambridge University Press, 1996.

\bibitem{chang1997conditioning}
J.~T. Chang and D.~Pollard.
\newblock Conditioning as disintegration.
\newblock {\em Statistica Neerlandica}, 51(3):287--317, 1997.

\bibitem{chen2003kolmogorov}
Q.~Chen, S.~Chen, G.~L. Eyink, and K.~R. Sreenivasan.
\newblock Kolmogorov’s third hypothesis and turbulent sign statistics.
\newblock {\em Physical Review Letters}, 90(25):254501, 2003.

\bibitem{constantin2007regularity}
P.~Constantin, B.~Levant, and E.~S. Titi.
\newblock Regularity of inviscid shell models of turbulence.
\newblock {\em Physical Review E}, 75(1):016304, 2007.

\bibitem{cornfeld2012ergodic}
I.~P. Cornfeld, S.~V. Fomin, and Y.~G. Sinai.
\newblock {\em Ergodic theory}.
\newblock Springer, 2012.

\bibitem{deimling2010nonlinear}
K.~Deimling.
\newblock {\em Nonlinear functional analysis}.
\newblock Springer, Berlin, 1985.

\bibitem{eckmann1985ergodic}
J.-P. Eckmann and D.~Ruelle.
\newblock Ergodic theory of chaos and strange attractors.
\newblock {\em Reviews of Modern Physics}, 57(3):617--656, 1985.

\bibitem{eyink2006turbulent}
G.~L. Eyink.
\newblock Turbulent cascade of circulations.
\newblock {\em Comptes Rendus Physique}, 7(3-4):449--455, 2006.

\bibitem{eyink2007turbulent}
G.~L. Eyink.
\newblock Turbulent diffusion of lines and circulations.
\newblock {\em Physics Letters A}, 368(6):486--490, 2007.

\bibitem{eyink2003gibbsian}
G.~L. Eyink, S.~Chen, and Q.~Chen.
\newblock Gibbsian hypothesis in turbulence.
\newblock {\em Journal of Statistical Physics}, 113(5-6):719--740, 2003.

\bibitem{falkovich2009symmetries}
G.~Falkovich.
\newblock Symmetries of the turbulent state.
\newblock {\em Journal of Physics A}, 42(12):123001, 2009.

\bibitem{falkovich2001particles}
G.~Falkovich, K.~Gawedzki, and M.~Vergassola.
\newblock Particles and fields in fluid turbulence.
\newblock {\em Reviews of Modern Physics}, 73(4):913, 2001.

\bibitem{fefferman2006existence}
C.~L. Fefferman.
\newblock {Existence and smoothness of the Navier-Stokes equation}.
\newblock In J.~Carlson, A.~Jaffe, and A.~Wiles, editors, {\em The millennium
  prize problems}, pages 57--67. AMS, 2006.

\bibitem{frisch1999turbulence}
U.~Frisch.
\newblock {\em {Turbulence: the Legacy of A.N.~Kolmogorov}}.
\newblock Cambridge University Press, 1995.

\bibitem{frisch2001burgulence}
Uriel Frisch and J{\'e}r{\'e}mie Bec.
\newblock Burgulence.
\newblock In A.~Yaglom, F.~David, and M.~Lesieur, editors, {\em New trends in
  turbulence}, pages 341--383. Springer, 2001.

\bibitem{gibbon2008three}
J.~D. Gibbon.
\newblock {The three-dimensional Euler equations: Where do we stand?}
\newblock {\em Physica D}, 237:1894--1904, 2008.

\bibitem{gledzer1973system}
E.~B. Gledzer.
\newblock System of hydrodynamic type admitting two quadratic integrals of
  motion.
\newblock {\em Soviet Physics -- Doklady}, 18:216, 1973.

\bibitem{hartman2002ordinary}
P.~Hartman.
\newblock {\em Ordinary differential equations}.
\newblock SIAM, 2002.

\bibitem{kolmogorov1941local}
A.~N. Kolmogorov.
\newblock {The local structure of turbulence in incompressible viscous fluid
  for very large Reynolds numbers}.
\newblock {\em Doklady Akademii Nauk SSSR}, 30(4):299--303, 1941.

\bibitem{kolmogorov1962refinement}
A.~N. Kolmogorov.
\newblock {A refinement of previous hypotheses concerning the local structure
  of turbulence in a viscous incompressible fluid at high Reynolds number}.
\newblock {\em Journal of Fluid Mechanics}, 13(1):82--85, 1962.

\bibitem{kraichnan1965lagrangian}
R.~H. Kraichnan.
\newblock Lagrangian-history closure approximation for turbulence.
\newblock {\em Physics of Fluids}, 8(4):575--598, 1965.

\bibitem{kraichnan1975remarks}
R.~H. Kraichnan.
\newblock Remarks on turbulence theory.
\newblock {\em Advances in Mathematics}, 16(3):305--331, 1975.

\bibitem{lax2007linear}
P.~D. Lax.
\newblock {\em Linear algebra and its applications}.
\newblock Wiley, New Jersey, 2007.

\bibitem{l1991scale}
V.~S. L'vov.
\newblock {Scale invariant theory of fully developed hydrodynamic
  turbulence-Hamiltonian approach}.
\newblock {\em Physics Reports}, 207(1):1--47, 1991.

\bibitem{l1998improved}
V.~S. L'vov, E.~Podivilov, A.~Pomyalov, I.~Procaccia, and D.~Vandembroucq.
\newblock Improved shell model of turbulence.
\newblock {\em Physical Review E}, 58(2):1811, 1998.

\bibitem{l1997temporal}
V.~S. L'vov, E.~Podivilov, and I.~Procaccia.
\newblock Temporal multiscaling in hydrodynamic turbulence.
\newblock {\em Physical Review E}, 55(6):7030, 1997.

\bibitem{mailybaev2016spontaneously}
A.~A. Mailybaev.
\newblock Spontaneously stochastic solutions in one-dimensional inviscid
  systems.
\newblock {\em Nonlinearity}, 29(8):2238--2252, 2016.

\bibitem{mailybaev2020hidden}
A.~A. Mailybaev.
\newblock Hidden scale invariance of intermittent turbulence in a shell model.
\newblock {\em Physical Review Fluids}, 6(1):L012601, 2021.

\bibitem{mailybaev2021solvable}
A.~A. Mailybaev.
\newblock Solvable intermittent shell model of turbulence.
\newblock {\em Communications in Mathematical Physics}, 388:469--478, 2021.

\bibitem{mailybaev2022shell}
A.~A. Mailybaev.
\newblock Shell model intermittency is the hidden self-similarity.
\newblock {\em Physical Review Fluids}, 7:034604, 2022.

\bibitem{mailybaev2021spontaneously}
A.~A. Mailybaev and A.~Raibekas.
\newblock {Spontaneously stochastic Arnold's cat}.
\newblock {\em Preprint arXiv:2111.03666}, 2021.

\bibitem{mailybaev2022hidden}
A.~A. Mailybaev and S.~Thalabard.
\newblock {Hidden scale invariance in Navier-Stokes intermittency}.
\newblock {\em Philosophical Transactions of the Royal Society A},
  380:20210098, 2022.

\bibitem{oberlack2022turbulence}
M.~Oberlack, S.~Hoyas, S.~V. Kraheberger, F.~Alc{\'a}ntara-{\'A}vila, and
  J.~Laux.
\newblock Turbulence statistics of arbitrary moments of wall-bounded shear
  flows: a symmetry approach.
\newblock {\em Physical Review Letters}, 128(2):024502, 2022.

\bibitem{oberlack2010new}
M.~Oberlack and A.~Rosteck.
\newblock New statistical symmetries of the multi-point equations and its
  importance for turbulent scaling laws.
\newblock {\em Discrete \& Continuous Dynamical Systems Series S},
  3(3):451--471, 2010.

\bibitem{ohkitani1989temporal}
K.~Ohkitani and M.~Yamada.
\newblock {Temporal intermittency in the energy cascade process and local
  Lyapunov analysis in fully developed model of turbulence}.
\newblock {\em Progress of Theoretical Physics}, 81(2):329--341, 1989.

\bibitem{frisch1985singularity}
G.~Parisi and U.~Frisch.
\newblock On the singularity structure of fully developed turbulence.
\newblock In M.~Ghil, R.~Benzi, and G.~Parisi, editors, {\em Predictability in
  Geophysical Fluid Dynamics}, pages 84--87. North-Holland, Amsterdam, 1985.

\bibitem{rudin2006real}
W.~Rudin.
\newblock {\em Real and complex analysis}.
\newblock McGraw-Hill, 2006.

\bibitem{she1994universal}
Z.-S. She and E.~Leveque.
\newblock Universal scaling laws in fully developed turbulence.
\newblock {\em Physical Review Letters}, 72(3):336, 1994.

\bibitem{sreenivasan1991fractals}
K.~R. Sreenivasan.
\newblock Fractals and multifractals in fluid turbulence.
\newblock {\em Annual Review of Fluid Mechanics}, 23(1):539--604, 1991.

\bibitem{teschl2012ordinary}
G.~Teschl.
\newblock {\em Ordinary differential equations and dynamical systems}.
\newblock AMS, 2012.

\bibitem{thalabard2020butterfly}
S.~Thalabard, J.~Bec, and A.~A. Mailybaev.
\newblock From the butterfly effect to spontaneous stochasticity in singular
  shear flows.
\newblock {\em Communications Physics}, 3(1):1--8, 2020.

\bibitem{PhysRevX.11.021063}
N.~Vladimirova, M.~Shavit, and G.~Falkovich.
\newblock Fibonacci turbulence.
\newblock {\em Physical Review X}, 11:021063, 2021.

\bibitem{waclawczyk2014statistical}
M.~Wac{\l}awczyk, N.~Staffolani, M.~Oberlack, A.~Rosteck, M.~Wilczek, and
  R.~Friedrich.
\newblock {Statistical symmetries of the Lundgren-Monin-Novikov hierarchy}.
\newblock {\em Physical Review E}, 90(1):013022, 2014.

\end{thebibliography}

\end{document}